\documentclass[11pt]{article}
\usepackage{mathptmx}
\usepackage{amsmath,amsthm,amssymb}
\usepackage{graphicx}
\usepackage{float}
\usepackage{lmodern,microtype}
\usepackage[margin=1in]{geometry}

\usepackage[utf8]{inputenc}
\usepackage[round]{natbib}
\newcommand{\cmmnt}[1]{}
\usepackage{setspace}
\usepackage{tabularx}
\usepackage{subcaption}
\usepackage{nicefrac}
\usepackage{etoolbox}
\usepackage{xcolor}
\usepackage{tikz}
\usepackage{hyperref}
\newcommand{\R}{\ensuremath{\mathbb R}}

\usepackage{bbm}
\DeclareMathOperator{\Cov}{\mathbbm{Cov}}

\usepackage{amsmath,amssymb}

\usepackage{url}

\def\co{{\rm co}\hskip 1pt}

\def\Var{\mathbb{V}} 
\def\E{{\mathbb E}}

\def\proj{{\rm proj}}
\def\proj0{{\rm proj}_{\co M_0}}

\def\R{{\mathbb{R}}}

\newtheorem{theorem}{Theorem}
\newenvironment{theo}{\vskip.2cm\begin{theorem}}%
{\end{theorem}\vskip.2cm}
\newtheorem{lemma}{Lemma}
{\end{lemma}\vskip.2cm}
\newtheorem{corolla}{Corollary}
\newenvironment{coro}{\vskip.2cm\begin{corolla}}%
{\end{corolla}\vskip.2cm}
\newtheorem{defini}{Definition}
{\end{defini}\vskip.2cm}
\newtheorem{proposi}{Proposition}
{\end{proposi}\vskip.1cm}
\newtheorem{cla}{Claim}
{\end{cla}\vskip.2cm}
\newtheorem{assump}{Assumption}
\newenvironment{assu}{\vskip.1cm\begin{assump}}%
{\end{assump}\vskip.1cm}
\newtheorem{hypoth}{Assumption}

{\end{hypoth}\vskip.1cm}

\usepackage[english]{isodate}
\graphicspath{ {./Image/} {./Images/} }
\title{\textbf{Shock Propagation and Macroeconomic Fluctuations}}

\author{Antoine Mandel\thanks{
          Paris School of Economics, Universit\'{e} Paris 1 Pantheon-Sorbonne. Maison des Sciences \'{E}conomiques, 106-112 Boulevard de l'h\^{o}pital 75647 Paris Cedex 13, France.} \and   Vipin P. Veetil\thanks{Economics Area, Indian Institute of Management Kozhikode, Kerala 673 570, India.}}

\date{\small\printdayoff\today}
\begin{document}
\maketitle
\begin{abstract}
\setstretch{1.3}
We study how idiosyncratic firm-level shocks generate aggregate volatility and tail risk when they propagate through a production network under overlapping adjustment: new productivity draws arrive before the economy reaches the static equilibrium associated with earlier draws. Each innovation generates a `productivity wave' that mixes and dissipates over time as it travels through the production network. Macroeconomic fluctuations emerge from the interference between these waves of different vintages. The interference between these waves is governed by the dominant transient eigenvalue of the production network, and therefore so are the macroeconomic fluctuations they generate. In such a dynamic regime, the tail of the degree distribution is a markedly weaker determinant of macro fluctuations than in the fully adjusted static benchmark. And the macroeconomic significance of the degree-heterogeneity of production networks cannot be known without knowing the rate at which the economy converges to equilibrium or equivalently the spectral properties of the production network. More concretely, once we permit the time-averaging of shocks, granular shocks may account for only a small fraction of the empirically observed aggregate volatility.
\end{abstract}
\vspace{1cm}

\newpage
\setstretch{1.4}
\section{Introduction}
What determines the contribution of granular shocks to aggregate volatility? Economists have for decades recognized that the macroeconomic significance of granular shocks depends on the rate of decay of aggregate volatility with an increase in the size of the economy.\footnote{If the shocks are independent across firms, aggregate volatility falls in inverse proportion to the square root of the number of firms, by the Central Limit Theorem. But when the firms interact with each other on a production network, the shocks are no longer independent as the shock to one firm affects its input sellers and output buyers. In a network setting, therefore, aggregate volatility falls more slowly, at a rate governed by the power-law exponent of the degree distribution of the production network, and the fatter the tail of the degree distribution the slower the decay. See \citet{acemoglu2012network} for a derivation of the rate of decay of aggregate volatility within a network setting and \citet{mandel2023voltest} for estimates of the power-law exponent for the US production network.} Few have however considered how the contribution of granular shocks to aggregate volatility could depend on the rate at which the economy converges to equilibrium. It is indeed true that one source of the dampening of the impact of productivity shocks on aggregate output is the fact that an increase in the productivity of one firm partially nullifies the decrease in the productivity of another. There is however a second, wholly distinct, source of dampening. One draw of productivity shocks can dampen the impact of the previous draw if the new shocks arrive before the economy has reached equilibrium. To understand the origins of such dampening, let us begin by noting the simple fact that each draw of firms' productivities defines a specific level of aggregate output. If each round of productivity shocks arrives after the economy reaches equilibrium, then aggregate volatility is a measure of the differences between equilibrium outputs. If however the shocks arrive \emph{before} the economy reaches equilibrium, then aggregate volatility can no longer be measured by the differences between the equilibrium outputs pertaining to different productivity draws. Rather, aggregate volatility must be measured by the changes in the realized temporal sequence of outputs. The distribution of this temporal sequence of outputs can be different, very different, from the distribution of equilibrium outputs for any given sequence of productivity draws.

\smallskip
Consider a network economy, with firms as buyers and sellers of intermediate inputs, that is periodically hit by idiosyncratic productivity shocks. After the economy is hit by a draw of productivities, it begins moving towards the equilibrium consistent with the draw. This journey is characterized by a reallocation of intermediate inputs: firms that received a positive productivity shock embark on the process of increasing output by drawing more inputs and firms that received a negative productivity shock slowly release resources by contracting output.\footnote{Within a network setting, a firm that experiences a positive productivity shock increases output for every given level of input. By posting a lower price it then delivers more physical input to its customers at unchanged nominal outlay, so the expansion travels downstream.} These granular expansions and contractions cause aggregate output to journey from its pre-shock value to the value consistent with the productivity draw, with the direction of the journey being dictated by the cross-section of the production shocks. Such journeys though often embarked upon are seldom completed. In fact, a new productivity draw that arrives before the economy reaches the equilibrium pertaining to the old draw is capable of reversing the very direction of travel. If, for instance, the new draw is such that highly connected firms become less productive, then aggregate output may descend well before having reached the peak pertaining to the old productivity draw. And there is no guarantee that this descent will be completed. This journey too may be abandoned like the one before it. As output descends towards the trough defined by the last productivity draw, another draw might prod it to begin the ascent towards yet another new equilibrium output.\footnote{In the limit, the economy is perennially going everywhere, forever getting nowhere.} What all of this means is that in a system with periodic productivity shocks, the differences among the outputs realized at each time step will tend to be smaller than the differences between the peaks and troughs defined by the equilibrium outputs pertaining to the shocks. The temporal sequence of output tends not to reach the peaks and troughs flagged by static equilibrium because its movements are open to reversals initiated by a new shock before having reached the destination pointed to by an old shock. Put differently, realized macroeconomic fluctuations can be lower, and perhaps significantly so, than the equilibrium benchmark once one permits the time-averaging of granular shocks.

\smallskip
We formalize this notion of `time averaging' of granular shocks in a production network by embedding an i.i.d. productivity process into the decentralized adjustment dynamics of \citet{gualdi2016} and \citet{veetil2026costinflation}. Those frameworks extend the static network environment of \citet{acemoglu2012network} to a setting with explicit time, and we take from the second the equilibrium and stability results on which our comparisons rest. The economy consists of a finite set of firms linked as buyers and sellers of intermediate inputs, with no household, so that money circulates in a closed circuit and decreasing returns supply the leak through which disturbances dissipate. Technologies are Cobb-Douglas with Hicks-neutral firm productivity, and at each date the cross section of log productivities is drawn i.i.d. from a mean-zero distribution with constant variance. A key distinction from the equilibrium benchmark is informational and institutional: firms do not solve for general-equilibrium prices after each draw. Instead, they update posted prices and input demands using only locally observed demand and supply conditions, through a simple non-t\^{a}tonnement procedure. When productivities are held fixed, this bottom-up adjustment converges to the static equilibrium associated with that productivity profile. When new draws arrive before convergence, the realized allocation inherits partially completed adjustments from earlier draws. We measure aggregate output with a fixed nonnegative aggregation vector and establish the population risk comparisons for every such vector. Household-derived GDP or consumption weights, once expressed on the modeled quantities, are members of this class, so a household can supply the weights without being part of the modeled economy.

\smallskip
This dynamic structure admits a reduced-form representation that is convenient for aggregation. The system's state evolves as a linear Markov process driven by the shock stream, with propagation governed by iterates of the production network operator, and decreasing returns to scale make the process stable. This representation splits aggregate output into two channels. The first is a granular (Perron) channel, through which highly connected firms move the aggregate directly. The second is a reallocation (transient) channel, through which the network mixes shocks across firms. Each period's innovation launches a new `productivity wave' through the network. Its aggregate footprint is spread over future dates and decays geometrically at two rates, the granular component at a rate pinned down by returns to scale and the reallocation component at a rate pinned down by the dominant transient eigenvalue of the production network.\footnote{The Perron mode of the production network moves the log output of every firm alike and loads on the productivity shocks in proportion to equilibrium size, which is why we call it the granular channel. Under decreasing returns it decays at the rate given by returns to scale. The next eigenvalue governs the slowest-decaying transient mode and therefore controls the rate at which relative outputs converge to their fully adjusted allocation.} In this environment, realized aggregate output at any date can be viewed as a superposition of partially propagated waves from past vintages. One can therefore define two distinct aggregate output sequences for a given realized shock history: (i) the static series obtained by mapping each contemporaneous shock into its fully adjusted equilibrium outcome, and (ii) the dynamic series produced by overlapping adjustment. Throughout the paper, we refer to volatility and tail-risk objects computed from these two series as static and dynamic aggregate volatility and tail risk, respectively.

\smallskip
 Once shocks propagate over time, realized aggregates are shaped not only by \emph{who is large} in the network, but also by \emph{how quickly the network mixes and dissipates disturbances} from one date to the next. This change in perspective delivers a significant result: for any measure of aggregate output that loads on the network's transient modes, the macroeconomic significance of fat tails in the degree distribution depends on the spectral gap of the production network. When mixing is sufficiently slow, the dynamic series aggregates many partially propagated vintages at any given date, and this intertemporal superposition can wash out cross-sectional concentration. Put differently, degree heterogeneity is most visible in aggregate risk when mixing is sufficiently fast, that is, when the spectral gap is large, so that each vintage is processed largely in isolation and the economy is repeatedly exposed to the full cross-sectional amplitude of the current draw.

\smallskip
We establish these comparisons not only in large-economy asymptotics but also at finite horizons. For a finite number of shocks, we characterize the sampling distributions of static and dynamic realized volatility and show that the static first-order stochastically dominates the dynamic. At the population level the same gap between static and dynamic volatility is exponentially magnified in tail risk. Time dynamics, that is, dampen tail risk much more sharply than they dampen volatility, and the reason is rather simple. Realized output is a superposition of partially propagated waves from many vintages. Volatility is driven by typical fluctuations of this superposition, whereas tail events require an unusually sustained alignment of waves, a rare failure of cancellation across vintages. Overlapping adjustments make such alignment harder, and therefore compress the tails disproportionately. All of this is to say that the problem of macro risks generated by firm-level innovations cannot be reduced, even to a first-order approximation, to measuring the tails of the degree distribution of the production network. Instead, it is jointly governed by cross-sectional concentration, by returns to scale, by the production network's spectral gap, and by the weights through which aggregate output is measured.

\smallskip
We measure the empirical consequences of our theoretical claims with computational experiments on a reconstructed buyer--seller network of the US economy. The network aggregates to the empirically observed input-output tables and is faithful to the observed distribution of firm sizes. We compute the static aggregate volatility by running the system to equilibrium for each set of shocks. We compute dynamic aggregate volatility by letting the system run forward in time with a single set of productivity shocks at each time step. We find that dynamic volatility in this synthetic US economy is an order of magnitude lower than its static counterpart. Which is to say that our theoretical results are of empirical significance: once granular shocks interact with each other over time, they explain very little of the observed macroeconomic fluctuations.

\subsection{Related literature}
\label{sec:literature}

Our paper belongs to the literature that runs from \citet{lucas1977} through \citet{gabaix2011} and \citet{acemoglu2012network} and asks when idiosyncratic shocks add up to sizeable aggregate fluctuations. That literature is concerned with averaging across firms, and with how production linkages and size heterogeneity slow such averaging. We are concerned with averaging across time. When new productivity draws arrive before the economy has adjusted to earlier draws, the waves launched by successive draws interfere and partly cancel, a source of attenuation that an economy jumping from one equilibrium to the next does not have.\footnote{The idea that adjustment takes time, and that treating the economy as a sequence of instantaneous equilibria can miss important dynamics, has a long history, see for instance \citet{fisher2013stability} and \citet{leijonhufvud1993towards}.}

\smallskip
A large literature since \citet{acemoglu2012network} enriches the propagation of shocks with CES technologies, nominal rigidities, international linkages, and financial frictions.\footnote{A non-exhaustive list includes \citet{carvalho2014micro}, \citet{acemoglu2015networks}, \citet{diGiovanni2014firms}, \citet{carvalho2019primer}, \citet{baqaee2019micro}, and \citet{baqaee2022micro}.} These papers evaluate the economy at the equilibrium associated with each draw of shocks, so that aggregate fluctuations are governed by firm sizes and buyer--seller linkages alone. Some, like \citet{pasten2018price}, study how a given shock interacts with nominal frictions as it propagates through the network, but there too new shocks arrive only after the old ones are fully processed. The overlap between shocks of different vintages, which is what we study, is absent from this entire program. Our results speak to all of it, because whenever propagation takes time the mapping from micro shocks to macro fluctuations depends not only on who is large but also on how quickly the network mixes and dissipates disturbances.\footnote{A variant of our `wave view' of the transmission of shocks can be found in \citet{veetil2026costinflation}, which uses spectral decomposition to study the relative price effects of inflation. Also see \citet{liu2020}, \citet{Iyetomi2011}, and \citet{McNerney2025}.}

\smallskip
The interaction between shocks of different vintages is an old question, but its formal analysis in the context of granular shocks is recent.\footnote{One finds an early discourse on the matter in \citet{Hicks1985Methods}, though Hicks was less concerned with granular shocks.} \citet{Bouchaud2014}, \citet{Mandel2015}, and \citet{Dessertaine2022} study granular shocks as they propagate through a network over time, and they conclude that propagation over time amplifies the macroeconomic impact of granular shocks. The most explicit statement is in Section 3.4 of \citet{Dessertaine2022}, where aggregate volatility rises as convergence to equilibrium slows. Our conclusions differ for two reasons. The first is that this literature is concerned with markets that do not clear and with the inventories that accumulate as a result, whereas in our economy markets clear at every date and there are no inventories. We change the classical benchmark along one dimension only, the arrival of new shocks before the old ones are fully processed, and our results can therefore be read to mean that, absent inventory-style disequilibrium channels, the production network need not by itself amplify granular shocks through time. The second reason is that this literature measures aggregate volatility as dispersion in the level of output, whereas we measure it, as the empirical literature does, as dispersion in growth rates. When shocks do not overlap the two measures are proportional to each other. When shocks overlap they can move in opposite directions, because slower convergence spreads each shock over more dates, so that levels accumulate more past shocks while consecutive levels lie closer together. Our conclusion that growth volatility falls as convergence slows therefore does not contradict theirs, and Appendix~\ref{app:levels_vs_increments} shows that in our economy level dispersion also rises as convergence slows, though more slowly than in the fully processed benchmark.

\smallskip
Much of what we have said relates to the old ``correspondence debate'' on when the predictions of a static equilibrium analysis remain informative about an economy whose adjustments take time. Samuelson's correspondence principle justifies comparative statics when the equilibrium is dynamically stable under a plausible local law of motion \citep{Samuelson1941StabilityComparativeStatics,Samuelson1942StabilityLinearNonlinear}, and a closely related line of work studies the stability of general equilibrium adjustment mechanisms directly, tatonnement included \citep{ArrowBlockHurwicz1959StabilityCE2}. Our emphasis differs from theirs, because even when the equilibrium is locally stable, the empirical content of the static approximation depends on how fast the economy converges relative to how often shocks arrive. When shocks arrive before the economy has relaxed, the appropriate comparison is no longer between two equilibria, but between the equilibrium mapping and the dynamics of deviations around it, and it is this comparison, at the time scale on which volatility is measured, that the paper isolates.

\subsection{Organization of the paper}

The rest of the paper proceeds as follows. Section~\ref{sec:model} sets up the production network, the static equilibrium that each productivity draw defines, and the adjustment process through which firms move towards that equilibrium. Section~\ref{sec:volatility_leontief} compares static and dynamic aggregate volatility and tail risk when shocks arrive without end. We first let each shock travel a fixed number of steps through the network before the next one arrives, which shows where the overlap between shocks of different vintages comes from. We then let the network itself decide how far each shock has travelled when the next one arrives, and obtain closed forms for the attenuation of volatility and of tail risk. The section closes by asking how much of this attenuation an observer sees, and shows that an observer who weights firms by their equilibrium size sees no trace of the network's mixing in aggregate volatility, though other observers may see more or less. Section~\ref{sec:eigenmodes} makes the same comparison for a finite number of shocks, and shows that static volatility exceeds dynamic volatility not only on average but in distribution. Section~\ref{sec:spectral_gap_mediates_fat_tails} asks how much of the static granular benchmark survives once shocks overlap in time, and shows that the macroeconomic visibility of fat tails in the degree distribution depends on how fast the economy converges to equilibrium, so that fatter tails need not raise aggregate volatility. Section~\ref{sec:abm} reports the computational experiments on the reconstructed US production network, and Section~\ref{sec:conclusion} concludes.

\smallskip
All proofs are in Appendix~\ref{app:proofs}. Appendix~\ref{app:reversal} studies the episodes in which dynamic output moves by more than static output from one date to the next, derives how often they occur, shows that in large economies they are of vanishing size, and extends the results to one-sided shocks. Appendix~\ref{app:levels_vs_increments} shows that nearly all of our comparisons survive when volatility is measured in levels rather than in growth rates. Appendix~\ref{app:fat_tails_degree_approx} collects the derivations behind Section~\ref{sec:spectral_gap_mediates_fat_tails}, Appendix~\ref{app:abm_eigenvalue} records the transient spectrum of the reconstructed networks, and Appendix~\ref{app:notation} collects the notation.

%==========================================================
\section{The Model}
\label{sec:model}
%==========================================================
Before turning to the model we fix some notation. For a random variable $X$, let $\E[X]$ and $\Var[X]$ denote expectation and variance, and let $\Cov[X,Y]$ denote covariance. We use $X^t$ for powers, and keep $X_t$ and $X^{[t]}$ for indexing. Scalars are in lowercase ($x$), vectors in bold lowercase ($\mathbf x$), and matrices in bold uppercase ($\mathbf X$). For any firm-level object $x$, the variable $x_{i,t}$ denotes firm $i$'s value at calendar date $t$. Let $\mathbf 1$ be the all-ones vector in $\R^n$.\footnote{Throughout, we condition on the production network and suppress this conditioning in notation.}

%====================================================
\subsection{Economic environment}
\label{subsec:environment}
%====================================================

The economy is a production network of $n$ firms indexed by $N:=\{1,\ldots,n\}$, and it contains no
household.  Our model is that of \citet{veetil2026costinflation}, to which we add Hicks-neutral productivity shocks.\footnote{The
economy consists of firms linked through intermediate-input purchases, and money circulates among
firms according to the spending rule of that model.  Aggregate output is measured with a fixed aggregation vector, and we admit a class of such vectors, household-derived weights included, without requiring a household to select one.  The companion paper indexes input shares by the transpose
of the convention used here.}

\smallskip
Firm $i$'s spending shares across its suppliers are $a_{ji}\ge 0$, where $a_{ji}$ is the share of firm
$i$'s expenditure devoted to input $j$.  Each firm spends its available balance on intermediate inputs,
a behavioural rule we maintain throughout, and conditional on that budget its Cobb--Douglas
technology assigns the share $a_{ji}$ to input $j$.  The shares therefore sum to one for each buyer,
\begin{equation}
a_{ji}\ge 0,
\qquad
\sum_{j\in N} a_{ji}=1
\qquad\text{for each }i\in N,
\label{eq:colsum_one}
\end{equation}
so the matrix $\mathbf A=(a_{ji})_{j,i\in N}$ is column-stochastic.  Money therefore circulates in a
closed circuit among firms, and the network is the only friction in the economy.

We maintain the network conditions of \citet[Assumption~1]{veetil2026costinflation} throughout, of which two are used directly. The matrix $\mathbf A$ is primitive, so the circuit is ergodic and admits a unique positive stationary
distribution of balances, and its subdominant eigenvalue is separated in modulus from the remainder of
the spectrum.

\smallskip
Firm $j$ produces with a Cobb--Douglas technology of common returns to scale $\varsigma\in(0,1)$ and
Hicks-neutral productivity $z_{j,t}>0$,
\begin{equation}
q_{j,t+1}
=
z_{j,t}\prod_{i\in N} x_{ij,t}^{\,\varsigma a_{ij}},
\qquad j\in N.
\label{eq:prod_drs}
\end{equation}
Here $x_{ij,t}$ is the quantity of input $i$ that firm $j$ purchases at date $t$.  Money moves one step ahead of goods, since an order placed at date $t$ is settled at the price its seller then posts, and the
inputs it delivers produce output at date $t+1$.

Since the shares $a_{ij}$ sum to one across suppliers,
the input elasticities sum to $\varsigma$, and the constant-returns case is the boundary $\varsigma=1$.
Decreasing returns make the economy dissipative. Disturbances to output propagate through the network by the operator $\mathbf W=\varsigma\mathbf A^\top$, whose spectral radius is $\varsigma<1$, so the effect of a fixed disturbance on output dissipates over time.  This restriction concerns real propagation only.  The nominal spending
circuit conserves the money stock, because every firm spends its whole balance, so $\varsigma$ is a
production elasticity and no fraction of expenditure is withheld from the circuit.  The restriction $\varsigma<1$ does real work. At $\varsigma=1$ the real block
$\mathbf I-\varsigma\mathbf A^\top$ is singular, because the Perron root of $\mathbf A^\top$ is one, and
the model no longer determines a unique level system.\footnote{The distinction between a determinate
level system and a balanced growth path is the one drawn by the Hawkins--Simon viability condition and
by the von Neumann growth model.  See \citet{hawkins1949}, \citet{vonneumann1945}, and
\citet{sraffa1960}, and the discussion in \citet{veetil2026costinflation}.}

%==========================================================
\subsection{Equilibrium, transient dynamics, and stability}
\label{subsec:equil_conv}
%==========================================================

We first recall from \citet{veetil2026costinflation} that the economy has a unique equilibrium when productivities are held fixed, and that it converges to that equilibrium from any initial condition. We then let productivities change from date to date and derive the law of motion of output.

\smallskip
Each firm $j$ enters date $t$ with a nominal balance $m_{j,t}$ that finances its input purchases.  Each
firm spends that balance in full, and with Cobb--Douglas technologies the expenditure shares
conditional on that budget are constant, so firm $i$ spends $a_{ji}m_{i,t}$ with supplier $j$, and the revenue this spending generates becomes that supplier's balance at the next date.  This full-spending rule
closes the monetary circuit, and balances evolve autonomously through it,
\begin{equation}
\mathbf m_{t+1}=\mathbf A\,\mathbf m_t.
\label{eq:money_circuit}
\end{equation}
Primitivity of $\mathbf A$ gives a unique positive stationary distribution $\mathbf m^\ast$, normalized
by $\mathbf 1^\top\mathbf m^\ast=1$, which attracts every initial balance vector of equal total mass.
Each seller clears its market by posting the price that equates the nominal demand it faces with its
current output, $p_{i,t}=\mathcal D_{i,t}/q_{i,t}$, and the settled orders deliver the physical inputs
that produce next-period output.

\smallskip
Holding productivities fixed, this economy has a unique positive stationary equilibrium
$(\mathbf m^\ast,\mathbf q^\ast,\mathbf p^\ast)$, with $\mathbf A\mathbf m^\ast=\mathbf m^\ast$ and
equilibrium output determined by the real block $\mathbf I-\varsigma\mathbf A^\top$.
\citet{veetil2026costinflation} proves existence and uniqueness in its Proposition~1. Its Lemma~1 proves that balances, output, and prices converge to that equilibrium from every positive initial configuration of balances, orders in transit, and output. Balance deviations decay at the subdominant rate of $\mathbf A$, and output and price deviations at the geometric factor $\max\{\lambda_2(\mathbf A),\varsigma\}$.  Adding Hicks-neutral productivity to \eqref{eq:prod_drs}
shifts the technological constant of the log-output recursion and leaves both arguments unchanged.

\smallskip
We now let productivities change from date to date. At each date $t=1,2,\ldots$ a new innovation $\boldsymbol{\epsilon}_t$ arrives and sets productivities to $\mathbf z_t=\exp(\boldsymbol{\epsilon}_t)$, while the effects of earlier innovations are still working their way through the production network. We fix the stochastic structure of the innovations once and for all.

\begin{assu}[Innovations]
\label{assu:innovations}
The innovations $\{\boldsymbol{\epsilon}_t\}_{t\ge1}$ are i.i.d.\ across dates, with independent
components $\epsilon_{i,t}\sim\mathcal N(0,\sigma^2)$ for some $\sigma^2\in(0,\infty)$, so that
$\boldsymbol{\epsilon}_t\sim\mathcal N(\mathbf 0,\sigma^2\mathbf I)$.
\end{assu}

\noindent
Assumption~\ref{assu:innovations} asks for more than most of the paper uses. Every variance comparison and every variance ratio in the paper rests on independence across dates and on the common variance $\sigma^2$ alone. We use Gaussianity only when we compare distributions rather than variances, which is to say for tail probabilities, for the finite-horizon distributions of realized volatility, and for the frequency with which dynamic output moves by more than static output. And the zero mean is a normalization rather than a restriction. A common innovation mean shifts the stationary path of log output by a constant that cancels from every increment (Lemma~\ref{lem:mean_invariance} in Appendix~\ref{app:reversal}), so that everything we say about increments holds when the innovations are drawn from a one-sided distribution.

\smallskip
We now derive the law of motion of output. Substituting the input demands $x_{ij,t}=a_{ij}m_{j,t}/p_{i,t}$ and the market-clearing prices into \eqref{eq:prod_drs}, and taking logarithms, gives the closed recursion
\begin{equation}
\log \mathbf q_{t+1}
=
\mathbf W\,\log \mathbf q_t
+
\boldsymbol{\epsilon}_t
+
\mathbf b_0(\mathbf A,\varsigma),
\qquad
\mathbf W:=\varsigma\,\mathbf A^\top,
\quad
\rho(\mathbf W)=\varsigma<1,
\label{eq:logq_recursion}
\end{equation}
once balances have settled at $\mathbf m^\ast$. Here $\mathbf b_0(\mathbf A,\varsigma)$ is a deterministic vector that depends only on the network and on returns to scale. The operator $\mathbf W$ is what propagates output disturbances through the network, since firm $i$'s output loads on its suppliers' outputs through the input shares, scaled by returns to scale. Iterating the recursion forward gives
\begin{equation}
\log \mathbf q_{t+1}
=
\mathbf W^{t+1}\log \mathbf q_0
+
\sum_{\tau=0}^{t}\mathbf W^{\tau}\boldsymbol{\epsilon}_{t-\tau}
+
\Bigg(\sum_{\tau=0}^{t}\mathbf W^{\tau}\Bigg)\mathbf b_0(\mathbf A,\varsigma).
\label{eq:quantity_history_representation}
\end{equation}
The first term carries the dependence on initial conditions and vanishes geometrically as $t\to\infty$, since $\rho(\mathbf W)=\varsigma<1$, which is the global stability established by \citet{veetil2026costinflation}. The last term is deterministic and settles to a constant, the steady-state component of output. All of the stochastic variation in output is therefore in the middle term, which sums the innovations of every past date, each propagated through the network for as many steps as have elapsed since it arrived. Which is to say that output at any date superposes partially propagated innovations of every vintage.

\smallskip
We close by recording the spectrum of $\mathbf W$. Because the columns of $\mathbf A$ sum to one, the transpose $\mathbf A^\top$ is row-stochastic, and $\mathbf W=\varsigma\mathbf A^\top$ has constant row sums $\varsigma$. The Perron eigenvalue and eigenvectors are therefore\footnote{That the Perron eigenvalue
is exactly $\varsigma$, with right eigenvector $\mathbf v_1=\mathbf 1$, is forced by the closed monetary circuit, because every firm spends its whole balance, so that each row of $\mathbf W$ sums to the same
constant.  We therefore attach no independent empirical significance to the value $\varsigma$, or to
functions of it such as $(1-\varsigma)^2/(1+\varsigma)$.  The network enters only through the transient spectrum
$\{\varsigma\lambda_r(\mathbf A)\}_{r\ge2}$.}
\begin{equation}
\lambda_1=\varsigma\in(0,1),
\qquad
\mathbf v_1=\mathbf 1,
\qquad
\mathbf u_1=\mathbf m^\ast.
\label{eq:W_spectral}
\end{equation}
Here $\mathbf v_1$ and $\mathbf u_1$ are the right and left Perron eigenvectors and $\mathbf m^\ast$ is
the stationary distribution of balances, a size weighting that concentrates on highly connected firms.
The transpose puts the size-weighted shock $\mathbf m^{\ast\top}\boldsymbol{\epsilon}_t$, and not the cross-sectional mean $\mathbf 1^\top\boldsymbol{\epsilon}_t$, in the common direction. And the real block inherits the nominal spectrum scaled by returns to scale. If $\{\lambda_r(\mathbf A)\}$ are the eigenvalues of the monetary circuit, those of $\mathbf W$ are $\{\varsigma\lambda_r(\mathbf A)\}$, so the dominant transient eigenvalue of output is $\varsigma\lambda_2(\mathbf A)$, strictly below $\lambda_1=\varsigma$.

%====================================================
\subsection{Aggregate output and its fluctuations}
\label{subsec:aggregate_output}
%====================================================

We define aggregate output by summing the log outputs of firms with weights, and we ask of the weights only that they be nonnegative and sum to one,
\begin{equation}
\boldsymbol\gamma\in\R_+^n,
\qquad
\boldsymbol\gamma^\top\mathbf 1=1.
\label{eq:admissible_weights}
\end{equation}
We call any such vector admissible. The quantity history representation \eqref{eq:quantity_history_representation} writes $\log\mathbf q_{t+1}$ as the sum of a term inherited from initial conditions, a term driven by the innovations, and a deterministic term induced by the primitives. Our volatility objects are unchanged by deterministic additive terms, so we define aggregate output on the term driven by the innovations.

\begin{defini}[Aggregate output]
\label{def:aggregate_output}
Fix an admissible weight vector $\boldsymbol\gamma$. Aggregate output at date $t+1$ is
\begin{equation}
y_{t+1}
:=
\boldsymbol\gamma^\top \sum_{\tau=0}^{t}\mathbf W^{\tau}\boldsymbol{\epsilon}_{t-\tau}.
\label{eq:aggregate_output_shock_component}
\end{equation}
\end{defini}

\noindent
This definition encompasses GDP measured with household expenditure shares as the weights, the sales-weighted indices of the granular-origins literature, and a variety of other measures. Whatever we say about aggregate volatility in this paper therefore applies to GDP, but is far more general. In fact, the only weights that \eqref{eq:admissible_weights} rules out are weights of mixed sign, which measure differences between firms rather than an aggregate, and weights orthogonal to the vector of ones, which is the right Perron eigenvector of $\mathbf W$. Weights orthogonal to that vector would remove from the aggregate the common component of output, through which the size-weighted shock $\mathbf m^{\ast\top}\boldsymbol{\epsilon}_t$ enters, and would thereby prevent granular shocks from accumulating in the aggregate at all. Every admissible weight vector instead loads on the common component with weight one, since $\boldsymbol\gamma^\top\mathbf 1=1$.

\smallskip
Definition~\ref{def:aggregate_output} subtracts from $\boldsymbol\gamma^\top\log\mathbf q_{t+1}$ the term inherited from $\log\mathbf q_0$ and the term induced by $\mathbf b_0(\mathbf A,\varsigma)$. Both are deterministic, so the population variances and tail probabilities that we define in this subsection are the same for $y_t$ and for $\boldsymbol\gamma^\top\log\mathbf q_t$. At any finite date, however, the two subtracted terms still change from one date to the next, by the deterministic amount $\boldsymbol\gamma^\top\mathbf W^{t}\bigl((\mathbf W-\mathbf I)\log\mathbf q_0+\mathbf b_0(\mathbf A,\varsigma)\bigr)$, which decays at the rate $\varsigma$. Our finite-window statistics are therefore statistics of $y_t$ itself, which is to say of the measured aggregate of an economy that starts at the deterministic equilibrium of the settled circuit, $\log\mathbf q_0=(\mathbf I-\mathbf W)^{-1}\mathbf b_0(\mathbf A,\varsigma)$, at which that amount is zero at every date. From any other start, a measured increment carries the deterministic amount in addition to $\Delta y_t$.

\smallskip
The choice of weights matters, and one choice matters more than any other. The stationary balance distribution $\mathbf m^\ast$ is admissible, and its entry $m^\ast_i$ is firm $i$'s share of equilibrium gross sales, because a firm's balance is its revenue of the preceding date. We call $\boldsymbol\gamma=\mathbf m^\ast$ the equilibrium gross-sales weights, or size weights for short. Because $\mathbf m^{\ast\top}\mathbf W=\lambda_1\mathbf m^{\ast\top}$, size-weighted aggregation collapses \eqref{eq:aggregate_output_shock_component} to a scalar autoregression at the Perron rate alone,
\[
\boldsymbol\gamma=\mathbf m^\ast
\quad\Longrightarrow\quad
y_{t+1}=\sum_{\tau=0}^{t}\lambda_1^{\,\tau}\,\mathbf m^{\ast\top}\boldsymbol{\epsilon}_{t-\tau},
\]
and biorthogonality, $\mathbf u_1^\top\mathbf v_r=0$ for $r\ge2$, removes every transient mode before it reaches the aggregate. An observer who weights firms by their equilibrium size therefore cannot see the network's transient spectrum in the aggregate at any horizon, however slowly the network mixes. We now record the condition that keeps the dominant transient mode in the aggregate, and we maintain it wherever the transient channel is at issue.

\begin{assu}[Misaligned weights]
\label{assu:misalignment}
The aggregate loading of the dominant transient mode is nonzero:
\[
b:=\boldsymbol\gamma^\top\tilde{\mathbf v}_2\neq 0.
\]
\end{assu}

\noindent
Here $\tilde{\mathbf v}_2$ is the right eigenvector of $\mathbf W$ associated with its dominant transient eigenvalue, and the condition does not depend on how that eigenvector is scaled. The assumption restricts the observer's weights and says nothing about the economy. It holds for every weight vector outside the measure-zero set of those orthogonal to $\tilde{\mathbf v}_2$, and that set contains size weights, at which every transient loading vanishes and not merely the dominant one.

\smallskip
We measure aggregate fluctuations with one-step changes in aggregate output,\footnote{Volatility in macro and micro empirical practice is typically computed from (detrended)
growth rates rather than from dispersion in levels \citep{HP97,BK99,CP06,BFJST18}.  The convention is
also conceptually natural here.  With overlapping propagation, the level $y_t$ mixes the contemporaneous
impact of the new vintage with the unfinished effects of older vintages.  One-step changes $\Delta y_t$
instead difference out what is carried from $t-1$ to $t$, isolating the net movement of the
superposition from one date to the next.}
\[
\Delta y_t:=y_t-y_{t-1}.
\]
We compare the dynamic aggregate output of Definition~\ref{def:aggregate_output} with the static aggregate output $y_t^\ast:=\boldsymbol\gamma^\top(\mathbf I-\mathbf W)^{-1}\boldsymbol{\epsilon}_t$. The static aggregate output is the value the aggregate would take if the economy adjusted fully to the innovation $\boldsymbol{\epsilon}_t$ before the next innovation arrived, and its increments are $\Delta y_t^\ast:=y_t^\ast-y_{t-1}^\ast$. We now define aggregate volatility and tail risk, each in a static and a dynamic version and each as a population object and as a realized statistic over a finite window.

\begin{defini}[Aggregate volatility: static and dynamic]
\label{def:agg_vol_static_dynamic}
The static and dynamic population volatilities are
\[
\phi^\ast := \lim_{t\to\infty}\Var(\Delta y_t^\ast),
\qquad
\phi := \lim_{t\to\infty}\Var(\Delta y_t),
\]
and the static and dynamic realized volatilities over a window $t=2,\ldots,T$ with $T\ge 2$ are
\[
\hat\phi^\ast
:=
\frac{1}{T-1}\sum_{t=2}^{T}\bigl(\Delta y_t^\ast\bigr)^2,
\qquad
\hat\phi
:=
\frac{1}{T-1}\sum_{t=2}^{T}\bigl(\Delta y_t\bigr)^2.
\]
\end{defini}

\begin{defini}[Aggregate tail risk: static and dynamic]
\label{def:agg_tail_static_dynamic}
Fix a threshold $c>0$. The static and dynamic population lower-tail probabilities are
\[
\omega_c^\ast := \lim_{t\to\infty}\Pr\bigl(\Delta y_t^\ast<-c\bigr),
\qquad
\omega_c := \lim_{t\to\infty}\Pr\bigl(\Delta y_t<-c\bigr),
\]
and the static and dynamic realized lower-tail frequencies over a window $t=2,\ldots,T$ with $T\ge 2$ are
\[
\hat\omega_{c}^\ast
:=
\frac{1}{T-1}\sum_{t=2}^{T}\mathbf 1\{\Delta y_t^\ast<-c\},
\qquad
\hat\omega_{c}
:=
\frac{1}{T-1}\sum_{t=2}^{T}\mathbf 1\{\Delta y_t<-c\}.
\]
\end{defini}

\noindent
The four limits exist because each increment series is a stable linear filter of the i.i.d.\ innovation stream of Assumption~\ref{assu:innovations}, so that its moments converge geometrically as the influence of initial conditions decays. Both increments have mean zero, so the realized volatilities are uncentered sums of squares. The window length $T$ is understood from context in every hatted object, unless we compare window lengths explicitly.

\section{Asymptotic Volatility Bounds}
\label{sec:volatility_leontief}
%==========================================================

How much smaller is aggregate volatility when shocks overlap in time than when each shock is fully processed before the next one arrives? We take up this question for an economy that has been receiving shocks for so long that the influence of its initial condition has died out and shocks of every vintage are present in output. We proceed in two steps. We first let each shock travel a fixed number of layers through the production network before the next shock arrives, and we call the resulting economy the $\mathcal L$-economy after that number of layers. In the $\mathcal L$-economy the superposition of partially propagated shocks can be seen directly, and it gives us our first comparison between static and dynamic volatility, a bound that holds for every production network and every admissible weight vector, with no restriction on the network's transient spectrum. We then drop the fixed number of layers, since nothing in the economy separates the arrival of shocks from their propagation, and let the decay of the network's own transients decide how far each shock has travelled when the next one arrives. Writing aggregate output in terms of the eigenmodes of the propagation operator and keeping the dominant transient mode, we obtain closed forms for growth-rate volatility and for tail risk, and in particular the factor by which overlap deflates the static benchmark. We also ask how much of that deflation an observer sees, which turns on the weights with which the observer aggregates, and we end the section with tail risk, which overlap deflates by far more than it deflates volatility.

\subsection{The $\mathcal L$-economy}

Consider a productivity shock $\boldsymbol\epsilon$.  The shock propagates downstream through the
production network as each firm transmits output changes to the firms that buy inputs from it.  This transmission occurs in layers. The original recipients of the shock pass the disturbance to their buyers, who in turn pass it on to their buyers.  The aggregate effect is the sum of the incremental contributions from each layer of propagation. Because returns to scale are below one, this layer-by-layer expansion converges.  We now formalize finite-depth processing and quantify how rapidly the depth-$\mathcal L$ aggregate approaches the fully processed benchmark.

\smallskip
In our network representation, $a_{ji}$ is the share of firm $i$'s spending devoted to input $j$.
Multiplication by $\mathbf W$ therefore maps disturbances at suppliers into disturbances faced by their buyers, so a shock to firm $j$ transmits along the directed links $j\to i$ from supplier to buyer.
A depth-$\ell$ term $\mathbf W^\ell\boldsymbol{\epsilon}$ collects the component arriving
exactly $\ell$ links downstream. Formally, $(\mathbf W^\ell)_{ki}$ aggregates weights over
directed length-$\ell$ paths from $i$ to $k$, so that the entry carries a shock at source $i$ to destination $k$. Here $\ell$ indexes processing depth within a
period rather than calendar time.\footnote{%
This ``rounds of requirements'' logic is the standard Leontief interpretation, in which producing a unit of activity requires direct inputs, producing those inputs
requires second-round inputs, and total requirements are obtained by summing
across depths.  See \citet{leontief1936quantitative} and
\citet[Sec.~2.4]{millerblair2009input} for textbook treatments, and
\citet[App.~A]{elliottgolubjackson2014contagion} for a Neumann-series
representation in a network setting.}

Since $\boldsymbol{\gamma}$ is the aggregation vector of Definition~\ref{def:aggregate_output}, the change in aggregate output is $\boldsymbol{\gamma}^\top$ applied to the vector of firm-level output changes. In particular, the incremental
aggregate contribution arriving at depth $\ell$ is
$\boldsymbol{\gamma}^\top\mathbf W^\ell\boldsymbol{\epsilon}$, and processing the shock through
$\mathcal L$ layers amounts to summing the first $\mathcal L$ layer contributions:
\begin{equation}
y^\ast_{\mathcal L}(\boldsymbol{\epsilon})
:=
\boldsymbol{\gamma}^\top\mathbf S_{\mathcal L}\,\boldsymbol{\epsilon},
\qquad
\mathbf S_{\mathcal L}:=\sum_{\ell=0}^{\mathcal L-1}\mathbf W^\ell .
\label{eq:y_L_single_shock}
\end{equation}
We call $\mathbf S_{\mathcal L}$ the block operator of the first $\mathcal L$ layers. The block operator satisfies
$\mathbf S_1=\mathbf I$ and $(\mathbf I-\mathbf W)\mathbf S_{\mathcal L}=\mathbf I-\mathbf W^{\mathcal L}$.

Since returns to scale satisfy $\varsigma<1$, we have $\rho(\mathbf W)=\varsigma<1$ and higher powers of $\mathbf W$ decay. The Neumann series therefore converges,
\begin{equation}
\sum_{\ell=0}^{\infty}\mathbf W^\ell=(\mathbf I-\mathbf W)^{-1},
\label{eq:neumann_series}
\end{equation}
and aggregating with
$\boldsymbol{\gamma}$ gives the fully processed mapping
\begin{equation}
y^\ast(\boldsymbol{\epsilon})
=
\boldsymbol{\gamma}^\top(\mathbf I-\mathbf W)^{-1}\boldsymbol{\epsilon}.
\label{eq:y_star_map}
\end{equation}

\medskip
\noindent
The truncated aggregate $y^\ast_{\mathcal L}$ is the $\mathcal L$-th partial sum of the Neumann series \eqref{eq:neumann_series}, aggregated with $\boldsymbol{\gamma}$.\footnote{We reserve the subscripted symbol $y^\ast_{\mathcal L}$ for this static object, the depth-$\mathcal L$ truncation of the fully processed mapping applied to a single draw. The bracketed symbol $y_t^{[\mathcal L]}$ denotes the dynamic series in which one draw arrives per date and each vintage advances $\mathcal L$ layers between dates.} We now state the exact layer weights and bound the gap between the fully processed aggregate output $y^\ast$ and its depth-$\mathcal L$ truncation $y^\ast_{\mathcal L}$.

\begin{lemma}[Neumann convergence, exact layer weights, and finite-depth error]
\label{lem:neumann_convergence_L}
Fix an admissible weight vector $\boldsymbol{\gamma}$. Then
$(\mathbf I-\mathbf W)^{-1}=\sum_{\ell=0}^\infty \mathbf W^\ell$,
and the layer weights are exact: for every $\ell\ge0$,
\begin{equation}
\bigl\|\boldsymbol{\gamma}^\top \mathbf W^\ell\bigr\|_1
=
\varsigma^{\,\ell},
\qquad\text{equivalently}\qquad
\bigl\|\mathbf W^\ell\bigr\|_\infty
=
\varsigma^{\,\ell}.
\label{eq:layer_weight_identity}
\end{equation}
Consequently, for every $\mathcal L\ge 1$,
\begin{equation}
\Big|
y^\ast(\boldsymbol{\epsilon})-y^\ast_{\mathcal L}(\boldsymbol{\epsilon})
\Big|
=
\Big|\sum_{\ell\ge\mathcal L}\boldsymbol{\gamma}^\top \mathbf W^\ell\boldsymbol{\epsilon}\Big|
\le
\frac{\varsigma^{\,\mathcal L}}{1-\varsigma}\,\|\boldsymbol{\epsilon}\|_\infty,
\label{eq:finite_depth_error_bound_simple}
\end{equation}
and $y^\ast_{\mathcal L}(\boldsymbol{\epsilon})\to y^\ast(\boldsymbol{\epsilon})$ geometrically as
$\mathcal L\to\infty$.
\end{lemma}

\noindent\emph{Proof in Appendix~\ref{app:proof:lem:neumann_convergence_L}.}

\subsection{Superposition of shocks of different vintages}
\label{subsec:superposition_vintages}

Consider an economy that experiences primitive productivity shocks at discrete dates
$t=0,1,\ldots,T$.  When $\mathcal L$ is finite, each new shock arrives before earlier shocks have
been fully processed, that is, before the economy has effectively reached the fully processed
Leontief outcome associated with the earlier vintage.  It helps to think of each shock as generating a ``wave'' that propagates downstream, layer by layer, through the production network.  With only $\mathcal L$ layers processed per calendar date, older waves remain
present when new ones arrive, so the realized aggregate at date $t$ is a superposition of multiple
vintages, each at a different stage of processing.

\smallskip
To formalize this superposition, we index exposures by the block of layers that a vintage occupies.
During the date in which it arrives, a primitive shock $\boldsymbol{\epsilon}$ traverses the first
$\mathcal L$ layers.  Its contribution to aggregate output in that date is
$\boldsymbol{\gamma}^\top\mathbf S_{\mathcal L}\boldsymbol{\epsilon}=y^\ast_{\mathcal L}(\boldsymbol{\epsilon})$,
the depth-$\mathcal L$ truncation of Lemma~\ref{lem:neumann_convergence_L}.  During each subsequent
date it advances a further $\mathcal L$ layers.  After $\tau$ dates it therefore occupies the block of
layers $\tau\mathcal L,\ldots,(\tau+1)\mathcal L-1$ and contributes
$\boldsymbol{\gamma}^\top\mathbf W^{\tau\mathcal L}\mathbf S_{\mathcal L}\boldsymbol{\epsilon}$.

\smallskip
The overlap mechanism is easiest to see with two shocks, a vintage
$\boldsymbol{\epsilon}_0$ at date $0$ and a second vintage
$\boldsymbol{\epsilon}_1$ at date $1$, with
$\boldsymbol{\epsilon}_t\equiv \mathbf 0$ for $t\ge 2$.
In the $\mathcal L$--economy each vintage advances by $\mathcal L$ layers per date.
Hence, at date $t\ge 1$, the first vintage occupies the block beginning at depth $t\mathcal L$,
while the second occupies the block beginning at depth $(t-1)\mathcal L$, so
\begin{equation}
y_t^{[\mathcal L]}
=
\boldsymbol{\gamma}^\top\!\Big(
\mathbf W^{t\mathcal L}\mathbf S_{\mathcal L}\boldsymbol{\epsilon}_0
+
\mathbf W^{(t-1)\mathcal L}\mathbf S_{\mathcal L}\boldsymbol{\epsilon}_1
\Big).
\label{eq:two_vintages_depth_multiple}
\end{equation}

\smallskip
With one shock arriving each date, the same accounting applies.
At date $t$, the vintage from date $t-\tau$ occupies the block beginning at depth $\tau\mathcal L$.
The aggregate level can therefore be written as
\begin{equation}
y_t^{[\mathcal L]}
=
\boldsymbol{\gamma}^\top\!\Bigg(
\sum_{\tau=0}^{t}
\mathbf W^{\tau\mathcal L}\mathbf S_{\mathcal L}\boldsymbol{\epsilon}_{t-\tau}
\Bigg),
\qquad t\ge 0.
\label{eq:superposition_finite_L_w_multiples}
\end{equation}
In \eqref{eq:superposition_finite_L_w_multiples} the calendar date $t$ fixes the present, the date at which output is evaluated. The lag $\tau$ is
the age of a vintage, the number of dates since it arrived, so $\boldsymbol{\epsilon}_{t-\tau}$
is the shock launched $\tau$ dates ago. Both $t$ and $\tau$ are measured in dates, but $t$ is a fixed
point while $\tau$ runs over the past vintages still alive in the superposition. The depth $\mathcal L$
is the number of network layers a vintage advances per date, and it converts a vintage's age into its propagation depth, since a vintage of age $\tau$ occupies the block of $\mathcal L$ layers beginning at depth
$\tau\mathcal L$, the power on the operator $\mathbf W$.

\smallskip
As $t$ grows, more vintages are simultaneously present and the superposition lengthens.
The layer-weight identity \eqref{eq:layer_weight_identity} gives
$\|\boldsymbol{\gamma}^\top \mathbf W^{\tau\mathcal L}\mathbf S_{\mathcal L}\|_1
=\varsigma^{\tau\mathcal L}(1-\varsigma^{\mathcal L})/(1-\varsigma)$, so sufficiently old
vintages become negligible. Extending the shock sequence to infinity, we can write
\begin{equation}
y_t^{[\mathcal L]}
=
\boldsymbol{\gamma}^\top\!\Bigg(
\sum_{\tau=0}^{\infty}
\mathbf W^{\tau\mathcal L}\mathbf S_{\mathcal L}\boldsymbol{\epsilon}_{t-\tau}
\Bigg),
\qquad t\in\mathbb Z.
\label{eq:superposition_infinite_vintage_L}
\end{equation}
The series converges absolutely. At $\mathcal L=1$ the block is a single layer, $\mathbf S_1=\mathbf I$, and \eqref{eq:superposition_infinite_vintage_L} is the shock-driven component of Definition~\ref{def:aggregate_output} with the date index shifted by the one-period production lag. And because the blocks of successive vintages tile the layers, the exposures of the vintages present at any date sum to the fully processed exposure,
\begin{equation}
\sum_{\tau=0}^{\infty}\mathbf W^{\tau\mathcal L}\mathbf S_{\mathcal L}
=
(\mathbf I-\mathbf W^{\mathcal L})^{-1}(\mathbf I-\mathbf W^{\mathcal L})(\mathbf I-\mathbf W)^{-1}
=
(\mathbf I-\mathbf W)^{-1},
\label{eq:block_tiling}
\end{equation}
for every $\mathcal L\ge1$.  The first equality uses the Neumann expansion of
$(\mathbf I-\mathbf W^{\mathcal L})^{-1}$, which converges because
$\rho(\mathbf W^{\mathcal L})=\varsigma^{\mathcal L}<1$, and the identity
$(\mathbf I-\mathbf W)\mathbf S_{\mathcal L}=\mathbf I-\mathbf W^{\mathcal L}$.

\smallskip
We refer to $\{y_t^{[\mathcal L]}\}$ as dynamic output, because shocks arrive before earlier vintages are fully processed and realized aggregates
reflect the superposition of partially propagated waves.
By contrast, the static benchmark processes each period's draw to infinite depth before the next draw arrives. The corresponding per-date aggregate is the static aggregate output of Section~\ref{subsec:aggregate_output},
\[
y_t^\ast
=
\boldsymbol{\gamma}^\top(\mathbf I-\mathbf W)^{-1}\boldsymbol{\epsilon}_t,
\]
and $\{y_t^\ast\}$ is i.i.d.\ whenever $\{\boldsymbol{\epsilon}_t\}$ is i.i.d.

\smallskip
For a given shock realization, dynamic levels mix multiple vintages, while static levels reflect
only the fully processed current vintage.  A large upstream innovation can therefore influence
dynamic output for many dates as its wave gradually moves downstream.  If, before this wave has
fully propagated, an offsetting innovation arrives elsewhere in the network, the two waves coexist
in levels and partially cancel.  Dynamic output is thus a superposition of many partially processed shocks, in which the effect of no single shock is fully realized.

\subsection{Growth rates and shock-interference}
\label{subsec:growth_interference}

We now take one-step differences of log levels to study growth-rate fluctuations.  A
one-step change subtracts what is mechanically inherited from $t-1$ and therefore isolates the
net new movement of the wave superposition between two adjacent dates\footnote{%
Differencing acts as a high-pass filter on the superposition of partially processed vintages.  A
large innovation to an upstream firm typically generates a downstream ``wave'' whose impact on
aggregate levels unfolds gradually as it reaches deeper layers. In that case the level
sequence $\{y_t^{[\mathcal L]}\}$ can remain persistently displaced even after the arrival date of
the shock.  But if this propagated component changes only slowly from $t-1$ to $t$, as is natural when
propagation is incremental across layers, then most of it is mechanically inherited and cancels in
$\Delta y_t^{[\mathcal L]}:=y_t^{[\mathcal L]}-y_{t-1}^{[\mathcal L]}$.  Consequently, growth rates
put relatively little weight on slow-moving propagated effects and instead emphasize movements that
come from new innovations or sharp changes in the cross-sectional shock pattern across
adjacent vintages.}:
\[
\Delta y_t^{[\mathcal L]}:=y_t^{[\mathcal L]}-y_{t-1}^{[\mathcal L]}
\]

Differencing \eqref{eq:superposition_infinite_vintage_L} yields the moving-average representation
\begin{equation}
\Delta y_t^{[\mathcal L]}
=
\boldsymbol{\gamma}^\top\!\Bigg(
\sum_{\tau=0}^{\infty}
\mathbf W^{\tau\mathcal L}\mathbf S_{\mathcal L}\,
\Delta\boldsymbol{\epsilon}_{t-\tau}
\Bigg),
\qquad
\Delta\boldsymbol{\epsilon}_{t}
:=
\boldsymbol{\epsilon}_t-\boldsymbol{\epsilon}_{t-1}.
\label{eq:Dy_t_L_superposition}
\end{equation}
The static benchmark has no overlap, so static growth depends only on the current shock difference:
\[
\Delta y_t^\ast
:=
y_t^\ast-y_{t-1}^\ast
=
\boldsymbol{\gamma}^\top(\mathbf I-\mathbf W)^{-1}\Delta\boldsymbol{\epsilon}_t
=
\boldsymbol{\gamma}^\top\!\Bigg(
\sum_{\tau=0}^{\infty}
\mathbf W^{\tau\mathcal L}\mathbf S_{\mathcal L}
\Bigg)\Delta\boldsymbol{\epsilon}_t .
\]
The second equality is the tiling identity \eqref{eq:block_tiling}.
Comparing \eqref{eq:Dy_t_L_superposition} to the static expression shows that dynamic growth attaches
the block weight $\boldsymbol{\gamma}^\top \mathbf W^{\tau\mathcal L}\mathbf S_{\mathcal L}$ to the
lagged shock difference $\Delta\boldsymbol{\epsilon}_{t-\tau}$.  Static growth attaches the
same block weights, summed into the fully processed
resolvent weight $\boldsymbol{\gamma}^\top(\mathbf I-\mathbf W)^{-1}$, to the current shock
difference $\Delta\boldsymbol{\epsilon}_{t}$.  Equivalently, dynamic growth can be decomposed into
static growth plus an overlap correction:
\begin{equation}
\Delta y_t^{[\mathcal L]}
=
\Delta y_t^\ast
\;+\;
\underbrace{\boldsymbol{\gamma}^\top\!
\sum_{\tau=1}^{\infty}\mathbf W^{\tau\mathcal L}\mathbf S_{\mathcal L}
\bigl(\Delta\boldsymbol{\epsilon}_{t-\tau}-\Delta\boldsymbol{\epsilon}_{t}\bigr)}_{\text{shock-interference correction}}.
\label{eq:Dy_dynamic_static_plus_correction}
\end{equation}

Block by block, the correction compares the vintage that actually sits in the block with the current
vintage, which under full processing would sit there instead.  At $\mathcal L=1$ the blocks are single layers and the correction is
$\boldsymbol{\gamma}^\top\sum_{\tau\ge1}\mathbf W^{\tau}(\Delta\boldsymbol{\epsilon}_{t-\tau}-\Delta\boldsymbol{\epsilon}_{t})$.

\smallskip
Suppose firm $i$ receives a large positive innovation at date $t-1$.  In the static benchmark, this shock is
propagated to all depths before any new innovation arrives, so its downstream wave is realized
without interference.  In the dynamic economy, the positive wave reaches different downstream
neighborhoods only gradually.  The wave loads first at the shallow block emphasized by
$\boldsymbol{\gamma}^\top\mathbf S_{\mathcal L}$, and only later at the deeper blocks emphasized
by $\boldsymbol{\gamma}^\top\mathbf W^{\mathcal L}\mathbf S_{\mathcal L}$,
$\boldsymbol{\gamma}^\top\mathbf W^{2\mathcal L}\mathbf S_{\mathcal L}$,
and so on.  Now suppose that at date $t$ an offsetting negative innovation hits firms in precisely
the downstream neighborhoods that the older positive wave is just beginning to reach.  Then the
contemporaneous growth rate loads on a difference of adjacent vintages along the relevant exposure direction, as the negative wave arriving at $t$ subtracts from, and can even overturn, the marginal contribution of the positive wave from $t-1$ before that older wave has finished
propagating through the network.  The correction term in \eqref{eq:Dy_dynamic_static_plus_correction} captures this interference mechanism.

The correction is large when the exposure-weighted shock differences fluctuate sharply across
dates, or switch sign.  In that case
$\boldsymbol{\gamma}^\top\mathbf W^{\tau\mathcal L}\mathbf S_{\mathcal L}\Delta\boldsymbol{\epsilon}_{t-\tau}$ differs
systematically from $\boldsymbol{\gamma}^\top\mathbf W^{\tau\mathcal L}\mathbf S_{\mathcal L}\Delta\boldsymbol{\epsilon}_{t}$
at many depth blocks.  When the relevant components evolve slowly, the comparisons in the correction term largely cancel and dynamic growth stays close to static growth.

\smallskip
Under the two-sided innovation stream $\{\boldsymbol{\epsilon}_t\}_{t\in\mathbb Z}$ of Assumption~\ref{assu:innovations}, the superposition \eqref{eq:superposition_infinite_vintage_L} defines a stationary process. We write $\phi^{[\mathcal L]}$ and $\omega_c^{[\mathcal L]}$ for the population volatility and the population lower-tail probability of its increments, defined as in Definitions~\ref{def:agg_vol_static_dynamic} and~\ref{def:agg_tail_static_dynamic} with $\Delta y_t^{[\mathcal L]}$ in place of $\Delta y_t$. We now state our first comparison between the static and the dynamic series.

\begin{lemma}[Static bounds on dynamic volatility and tail risk in the $\mathcal L$-economy]
\label{lem:neumann_bounds_L}
Fix $\mathcal L\in\mathbb N$ and an admissible weight vector $\boldsymbol{\gamma}$. Then
\[
\phi^{[\mathcal L]} \;\le\; \phi^\ast,
\qquad\text{and}\qquad
\omega_c^{[\mathcal L]} \;\le\;\omega_c^\ast
\quad\text{for all } c>0.
\]
Moreover, as $\mathcal L\to\infty$,
\[
\Delta y_t^{[\mathcal L]}
\xrightarrow{L^2}
\boldsymbol{\gamma}^\top(\mathbf I-\mathbf W)^{-1}(\boldsymbol{\epsilon}_t-\boldsymbol{\epsilon}_{t-1})
=\Delta y_t^\ast,
\qquad\text{and hence}\qquad
\phi^{[\mathcal L]}\;\longrightarrow\;\phi^\ast.
\]
\end{lemma}

\noindent\emph{Proof in Appendix~\ref{app:proof:lem:neumann_bounds_L}.}

Lemma~\ref{lem:neumann_bounds_L} says that, when the shock stream runs forever, the static benchmark
bounds both dynamic growth volatility and dynamic lower-tail risk from above, whatever the block depth
$\mathcal L$.  The reason is the one we gave in the introduction.  With overlap, a one-step change in output picks
up the difference between what each vintage contributes today and what it contributed yesterday, so a
wave that has just arrived can be offset, layer by layer, by an older wave of the opposite sign that
is still travelling through the network.  The static benchmark never lets two waves meet, because each
draw is processed to completion before the next arrives, and so it has nothing to cancel against.

The block depth connects the two timing regimes.  At $\mathcal L=1$ each date
advances every vintage by a single layer, which is the discrete-time model of Section~\ref{sec:model}.
As $\mathcal L$ grows the current vintage traverses more layers before the next draw arrives, so its
exposure $\boldsymbol{\gamma}^\top\mathbf S_{\mathcal L}$ approaches the fully processed exposure
$\boldsymbol{\gamma}^\top(\mathbf I-\mathbf W)^{-1}$.  The older vintages, pushed to depth
$\tau\mathcal L$ and beyond, drop out of the superposition.  In the limit the increment loads on the
contemporaneous shock difference alone, through the resolvent, which is the static benchmark.  The
family $\{\phi^{[\mathcal L]}\}_{\mathcal L\ge1}$ therefore connects the dynamic economy at
$\mathcal L=1$ to the static benchmark as $\mathcal L\to\infty$, and the lemma bounds every member of
the family by the static benchmark.\footnote{The lemma orders population objects under the stationary law induced by the two-sided shock sequence, and does not assert a pointwise ordering along finite realized histories, so the sample statistics $\hat\phi$ and $\hat\omega_c$ need not satisfy the inequalities path by path. The lemma holds for every admissible weight vector and orders the two series separately for each one. The weights are fixed once and for all, because a weight vector that moved with the shocks would add the change in weights, applied to the previous date's outputs, to every measured increment, and no proof in the paper carries that term. The class includes GDP measured with fixed household expenditure shares, provided that the shares attach to the same log quantities and are normalized to sum to one, as well as equal weights and the equilibrium gross-sales weights. And an aggregation vector is not an exposure vector. A weight vector $\boldsymbol\gamma$ becomes the static exposure vector $\boldsymbol\gamma^\top(\mathbf I-\mathbf W)^{-1}$ through the resolvent. A vector taken from a formula for the response of GDP to productivity shocks therefore belongs at the exposure stage, and it would be propagated a second time if it were used as $\boldsymbol\gamma$. For the same reason we call $\mathbf m^\ast$ the equilibrium gross-sales weights rather than Domar weights.}

\subsection{From the $\mathcal L$-economy to eigenmodes}
\label{sec:ell_to_eigenmodes}

Lemma~\ref{lem:neumann_bounds_L} orders static against dynamic risk for any fixed block depth
$\mathcal L$, but it takes that depth as exogenous.  We now let the network itself decide how far a vintage is processed before the next one arrives.  The
$\mathcal L$--economy is only a pedagogical device for visualizing overlap across shock
vintages.  It treats the arrival of shocks and their propagation as if they took place in two separate
times.  In the actual discrete-time model there is no such split, since each tick of calendar time simultaneously brings a new vintage
$\boldsymbol{\epsilon}_t$ and advances every previously realized vintage one further step
through the same propagation operator $\mathbf W$.  A shock realized at date $t$ affects aggregates
$\tau$ dates later through $\mathbf W^{\tau}$, so what we called ``depth'' is nothing but elapsed calendar
time viewed through the network. Accordingly, we set $\mathcal L=1$ to eliminate parallel tracks of time and let network structure determine the effective extent of
overlap.  Networks for which the powers $\mathbf W^{\tau}$ decay rapidly process vintages quickly, so
only a few recent shocks remain economically relevant. Networks for which $\mathbf W^{\tau}$ decays
slowly keep old vintages ``alive'' for many dates, so overlap is strong.  The spectrum of $\mathbf W$ thus decides how far a vintage is processed before the next one
arrives.

\smallskip
In
Section~\ref{sec:model} we assumed that $\mathbf W$ is primitive. We now further assume that it is
diagonalizable.\footnote{%
Diagonalizability gives the scalar-mode representation that we now derive, in which each mode is an autonomous scalar autoregression.  For a defective operator the
Jordan--Chevalley decomposition writes
$\mathbf W^{\tau}=\sum_{r} \lambda_r^{\,\tau}\mathbf P_r + \mathbf N_\tau$, with $\mathbf P_r$ the spectral
projectors and $\mathbf N_\tau$ a sum of terms of order $\tau^m|\lambda_r|^{\,\tau}$ for finite $m$.
These terms decay geometrically after a single impulse, but under continually renewed innovations
they enter the stationary increment filter and change its variance, so the scalar-mode formulas do
not apply verbatim to a defective operator.  The variance bound of Lemma~\ref{lem:neumann_bounds_L} does not use the scalar-mode representation and needs no diagonalizability.}
Let $\{(\lambda_r,\mathbf v_r,\mathbf u_r)\}_{r=1}^n$ denote right and left eigenpairs of $\mathbf W$,
$\mathbf W\mathbf v_r=\lambda_r\mathbf v_r$ and
$\mathbf u_r^\top\mathbf W=\lambda_r\mathbf u_r^\top$, normalized so that
$\mathbf u_r^\top\mathbf v_s=\mathbf 1\{r=s\}$.  Then each power admits the modal decomposition
\begin{equation}
\mathbf W^{\tau}
=
\sum_{r=1}^n \lambda_r^{\,\tau}\,\mathbf v_r\mathbf u_r^\top,
\qquad \tau\ge 0.
\label{eq:A_power_modal_subsec}
\end{equation}

By \eqref{eq:W_spectral}, the Perron eigenvalue of $\mathbf W=\varsigma\mathbf A^\top$ is
$\lambda_1=\varsigma\in(0,1)$, with right eigenvector $\mathbf v_1=\mathbf 1$ and left eigenvector
$\mathbf u_1=\mathbf m^\ast$, the stationary distribution of balances.
Define the dominant transient modulus as
$\lambda_2:=\max\{|\lambda_r|:r\ge 2\}$,
which governs the slowest-decaying deviation from the Perron direction and therefore the
persistence of partially processed vintages.

\medskip
We now reduce the dynamics to two modes, the Perron mode $r=1$ and the dominant transient mode $r=2$. For the second of these to be a single real mode with a gap below it, we need the following assumption on the transient spectrum of $\mathbf W$, which we maintain throughout the two-mode analysis.

\begin{assu}[Dominant transient mode]
\label{assu:dominant_transient}
Order the eigenvalues of $\mathbf W$ other than the Perron eigenvalue $\lambda_1=\varsigma$ by modulus. (i) The dominant transient modulus $\lambda_2:=\max\{|\lambda_r|:r\ge 2\}$ is attained by a real, positive, and simple eigenvalue, which we identify with $\lambda_2$. (ii) The remaining modes are separated from it, $\lambda_3:=\max\{|\lambda_r|:r\ge 3\}<\lambda_2$.
\end{assu}

\noindent
Part~(i) says that the slowest transient of the network is a single real mode. Relative outputs then relax towards their equilibrium monotonically rather than oscillating, at the rate $\lambda_2=\varsigma\lambda_2(\mathbf A)$, with $\lambda_2(\mathbf A)<1$ the subdominant eigenvalue of the circuit, so that $0<\lambda_2<\lambda_1<1$. Were the slowest transient a complex pair, the state of the reallocation channel would be two-dimensional and oscillatory. Were it negative, overlap would amplify rather than attenuate increment volatility, because the inherited state then enters the increment with a coefficient larger than one in absolute value. Part~(ii) says that every other transient mode decays strictly faster than the dominant one after a single impulse, so that there is a single dominant transient mode.

\medskip
We now write aggregate output mode by mode. Substituting the modal decomposition \eqref{eq:A_power_modal_subsec} into the depth expansion \eqref{eq:aggregate_output_shock_component}, and writing $c_r:=\boldsymbol\gamma^\top\mathbf v_r$ for the aggregate loading of mode $r$ and $\xi_{r,t}:=\mathbf u_r^\top\boldsymbol{\epsilon}_t$ for its modal innovation, gives
\begin{equation}
y_{t+1}
=
\sum_{r=1}^n c_r\sum_{\tau=0}^{t}\lambda_r^{\,\tau}\,\xi_{r,t-\tau},
\label{eq:y_modal_exact}
\end{equation}
with no approximation. The two channels we study, the Perron mode and the dominant transient mode, are two of the terms of this sum. We now state the decomposition in the autoregressive form used throughout, dating each modal state by the newest innovation it contains, which relabels the dates of \eqref{eq:y_modal_exact} by the one-period production lag of Definition~\ref{def:aggregate_output}.

\begin{lemma}[Exact modal decomposition]
\label{lem:modal_decomposition}
Let $\mathbf W$ be diagonalizable and fix an admissible weight vector $\boldsymbol\gamma$. Aggregate output is the sum of $n$ autonomous scalar autoregressions,
\[
y_t=\sum_{r=1}^n c_r\,s_{r,t},
\qquad
s_{r,t}=\lambda_r s_{r,t-1}+\xi_{r,t},
\qquad s_{r,0}=0,
\]
and the granular loading is $c_1=1$ for every admissible weight vector.  Complex eigenvalues enter in
conjugate pairs whose combined contribution to $y_t$ is real.
\end{lemma}

\noindent\emph{Proof in Appendix~\ref{app:proof:lem:modal_decomposition}.}

Each term of the decomposition \eqref{eq:y_modal_exact} combines a filter, an eigenmode, and the aggregation vector. The filter $\sum_{\tau=0}^{t}\lambda_r^{\,\tau}\,\xi_{r,t-\tau}$ aggregates current and past vintages of the mode-$r$ innovation with geometric weights governed by $\lambda_r$. A large $|\lambda_r|$ keeps old vintages alive and increases overlap. The eigenmode $(\mathbf v_r,\mathbf u_r)$ maps the cross-section of shocks into the scalar modal innovation through the projection $\mathbf u_r^\top(\cdot)$ and re-embeds it along the direction $\mathbf v_r$. And the aggregation vector $\boldsymbol{\gamma}$ converts firm-level output into aggregate output through the loadings $c_r$, and thereby decides how visible each eigenmode is in the aggregate.

\smallskip
 In the exact decomposition \eqref{eq:y_modal_exact}, the $r=1$ (Perron) term and the
$r\ge 2$ (transient) terms are two economically distinct channels through which firm-level shocks reach
aggregate output. Grouping them and using $\mathbf u_1=\mathbf m^\ast$ and
$\boldsymbol{\gamma}^\top\mathbf v_1=\boldsymbol{\gamma}^\top\mathbf 1=1$ from \eqref{eq:W_spectral},
\begin{equation}
y_{t+1}
=
\underbrace{\mathbf m^{\ast\top}\!\sum_{\tau=0}^{t}\lambda_1^{\,\tau}\boldsymbol{\epsilon}_{t-\tau}}_{\text{granular (Perron) channel } y^{\mathrm{gr}}_{t+1}}
\;+\;
\underbrace{\sum_{r\ge 2}(\boldsymbol{\gamma}^\top\mathbf v_r)\,\mathbf u_r^\top\!\sum_{\tau=0}^{t}\lambda_r^{\,\tau}\boldsymbol{\epsilon}_{t-\tau}}_{\text{transient block}}.
\label{eq:y_perron_transient}
\end{equation}
The granular channel loads on the size-weighted shock $\mathbf m^{\ast\top}\boldsymbol{\epsilon}$, which is the classical mechanism by which highly connected firms move the aggregate, and, because
$\mathbf m^\ast$ concentrates on hubs, it does not wash out when the degree distribution is fat-tailed.
Note that its aggregate loading is $\boldsymbol\gamma^\top\mathbf v_1=1$ for every admissible weight vector, so the granular channel is always present at full strength. The transient block loads on the cross-sectional pattern of shocks relative to
the common direction, and is governed by the transient spectrum $\{\lambda_r\}_{r\ge 2}$.  The
reallocation channel, defined next, is the dominant ($r=2$) term of this block.

\smallskip
The granular channel and the transient block are time-averaged at different rates. The granular channel decays at the common rate $\lambda_1$, a single scalar that carries no
information about how the network mixes. The transient
block decays at the transient eigenvalues, the slowest being $\lambda_2$, which encode the network's
mixing and its degree heterogeneity. We therefore retain the dominant transient mode $r=2$ and develop the reallocation channel in full, which presumes Assumption~\ref{assu:misalignment}, since at size weights the loading $b$ is zero and the reallocation channel does not reach the aggregate at all.\footnote{The projection $\mathbf u_2^\top$ onto the transient direction is orthogonal to the Perron direction by biorthogonality, $\mathbf u_2^\top\mathbf v_1=0$, so it isolates the reallocation channel without any demeaning of the cross-section. Demeaning by the cross-sectional mean $\mathbf 1^\top\boldsymbol{\epsilon}$ would in fact remove the wrong direction, because under $\mathbf W=\varsigma\mathbf A^\top$ the common mode is driven by the size-weighted shock $\mathbf m^{\ast\top}\boldsymbol{\epsilon}$ and not by the cross-sectional mean.} We return to the rate at which the granular channel time-averages once the reallocation channel is in hand.

\smallskip
Since the mode enters only through the rank-one product $\mathbf v_2\mathbf u_2^\top$, its
contribution can be written equivalently with any reciprocal rescaling of the left--right
eigenvectors.  For our purpose, define the rescaling
\begin{equation}
\tilde{\mathbf u}_2 := \frac{\mathbf u_2}{\|\mathbf u_2\|},
\qquad
\tilde{\mathbf v}_2 := \|\mathbf u_2\|\,\mathbf v_2,
\label{eq:u2_normalization}
\end{equation}
so that $\|\tilde{\mathbf u}_2\|=1$ and
$\tilde{\mathbf v}_2\tilde{\mathbf u}_2^\top=\mathbf v_2\mathbf u_2^\top$, leaving the rank-one
product, and hence the channel defined next, unchanged.%
\footnote{Eigenvectors are defined only up to reciprocal scaling of the left--right pair.  Fixing
$\|\tilde{\mathbf u}_2\|=1$ is therefore without loss and is convenient because, under
Assumption~\ref{assu:innovations}, the projected
innovation $\tilde{\mathbf u}_2^\top\boldsymbol{\epsilon}_t$ then has variance $\sigma^2$.} We now use these rescaled eigenvectors to characterize the dynamics of an $n$-dimensional shock vector through the network, with scalar state variables. Define the second-mode innovation and the associated aggregate loading by\footnote{These objects retain clear economic meaning in the eigenmode representation.  The scalar
$\eta_t$ is the projection of the firm-level innovation onto the dominant transient direction, and it summarizes, at date $t$, how strongly the realized cross-section of shocks excites the slow mode.
The scalar $b$ is the aggregate visibility, or loading, of that mode, and it measures how a unit of the
mode's state translates into aggregate output through the aggregation weights $\boldsymbol{\gamma}$.}
\begin{equation}
\eta_t := \tilde{\mathbf u}_2^\top \boldsymbol{\epsilon}_t \in \mathbb R,
\qquad
b := \boldsymbol{\gamma}^\top \tilde{\mathbf v}_2 \in \mathbb R.
\label{eq:eta_b_def}
\end{equation}
In these coordinates the reallocation channel is the dominant transient term of \eqref{eq:y_perron_transient},\footnote{Under Assumption~\ref{assu:dominant_transient} the modes $r\ge3$ decay strictly faster than the dominant one after a single impulse. That ordering of decay rates does not by itself make the two-mode representation $y_t\approx y_t^{\mathrm{gr}}+y_t^{\mathrm{re}}$ accurate for stationary increment variance, because continually renewed innovations excite every mode at every date. The own contribution of a real mode $r$ to the stationary increment variance is $2\sigma^2c_r^2\|\mathbf u_r\|_2^2/(1+\lambda_r)$, which does not vanish as $\lambda_r\to0$, so a fast mode with a sizeable loading carries a sizeable share of the aggregate. We return to how all the channels combine into the aggregate, and to the error of the two-mode representation, after developing the two leading channels.}
\begin{equation}
y_{t+1}^{\mathrm{re}}
:=
b\sum_{\tau=0}^{t}\lambda_2^{\,\tau}\,\eta_{t-\tau}.
\label{eq:y_single_mode_time_subsec}
\end{equation}

\subsection{Aggregate volatility under the dominant-transient reduction}
\label{subsec:dominant_transient_vol}

In the static benchmark, each date-$t$ innovation is fully processed within the date before the
next arrives.  The reallocation component of the fully processed mapping \eqref{eq:y_star_map} is the $r=2$ term of its eigenmode expansion. Expanding $(\mathbf I-\mathbf W)^{-1}$ over the eigenmodes, that term carries the within-date resolvent weight $(1-\lambda_2)^{-1}$, so
\begin{equation}
y_t^{\ast,\mathrm{re}} = \frac{b}{1-\lambda_2}\,\eta_t,
\qquad
\Delta y_t^{\ast,\mathrm{re}}
=
\frac{b}{1-\lambda_2}\,(\eta_t-\eta_{t-1}).
\label{eq:Dy_static_lambda2}
\end{equation}

\smallskip
In the dynamic economy, by contrast, a date-$t$ shock is processed only one step before date $t+1$
arrives, so partially processed vintages remain in the system.  To make this overlap explicit
in differences, it is convenient to summarize the entire history of projected innovations
$\{\eta_{t-\tau}\}_{\tau\ge 0}$ by a single scalar state.  Define the dominant-mode state $s_t$ by
\begin{equation}
s_t := \sum_{\tau=0}^{t-1}\lambda_2^{\,\tau}\eta_{t-\tau},
\qquad\text{equivalently}\qquad
s_t=\lambda_2 s_{t-1}+\eta_{t},
\qquad s_0=0,
\label{eq:s_recursion}
\end{equation}
so that the reallocation channel is the scalar process\footnote{Here we align the dynamic series with the contemporaneous static benchmark
$y_t^{\ast,\mathrm{re}}=\tfrac{b}{1-\lambda_2}\eta_t$ by indexing the reduced state so that $s_t$, and hence
$y_t^{\mathrm{re}}=b\,s_t$, responds to the date-$t$ projected innovation $\eta_t$.  This absorbs the
one-period production lag of Definition~\ref{def:aggregate_output} (under which $y_{t+1}$ loads on
$\boldsymbol{\epsilon}_t$). Because the $\eta_t$ are i.i.d., the relabeling changes none of the volatility or
tail-risk objects, and it lines up the date-$t$ static and dynamic increments so that the two can be
compared date by date.}
\begin{equation}
\label{eq:dyn_output_approx}
y_t^{\mathrm{re}}= b\,s_t.
\end{equation}
With $s_t$ in hand,\footnote{Population statements evaluate this channel under its stationary law, in which $s_t=\sum_{\tau\ge0}\lambda_2^{\,\tau}\eta_{t-\tau}$ is the two-sided moving average. The initialization $s_0=0$ in \eqref{eq:s_recursion} matters only for finite-window statistics.} the infinite moving average becomes a one-dimensional recursion, the channel
increment has a closed form, and so do the two variances that we now compare.  Combining \eqref{eq:dyn_output_approx} with \eqref{eq:s_recursion}, the reallocation component of dynamic output growth can be written as
\begin{align}
\Delta y_t^{\mathrm{re}}
&= y_t^{\mathrm{re}}-y_{t-1}^{\mathrm{re}} \notag\\
&= b(s_t-s_{t-1}) \notag\\
&= b\bigl((\lambda_2-1)s_{t-1}+\eta_{t}\bigr).
\label{eq:Dy_dynamic_lambda2}
\end{align}

We now state the spectral counterpart of Lemma~\ref{lem:neumann_bounds_L}.\footnote{%
In standard network models, equilibrium output and influence vectors are functions of the resolvent
$(\mathbf I-\mathbf W)^{-1}$, whose spectral representation depends on the full
eigenvalue--eigenvector decomposition of $\mathbf W$.  Movements in eigenvalues therefore shift both
equilibrium exposure and the persistence of propagation.  Our results refine the usual influence-vector logic, in that the same spectral objects that determine equilibrium exposure also govern
transient survival, and hence the gap between static and dynamic volatility.}

\begin{theo}[Attenuation of growth-rate volatility by intertemporal overlap]
\label{thm:attenuation}
Maintain Assumptions~\ref{assu:innovations}, \ref{assu:misalignment}, and~\ref{assu:dominant_transient}, and evaluate the reallocation channel under its stationary law. Then the two increment variances are
\begin{align}
\Var(\Delta y_t^{\mathrm{re}})&=\frac{2b^2\sigma^2}{1+\lambda_2},
\label{eq:var_dynamic_increment}\\
\Var(\Delta y_t^{\ast,\mathrm{re}})&=\frac{2b^2\sigma^2}{(1-\lambda_2)^2},
\label{eq:var_static_increment}
\end{align}
so that
\begin{equation}
\frac{\Var(\Delta y_t^{\mathrm{re}})}{\Var(\Delta y_t^{\ast,\mathrm{re}})}
=
\mathcal R(\lambda_2),
\qquad\text{where}\quad
\mathcal R(\lambda):=\frac{(1-\lambda)^2}{1+\lambda},
\label{eq:attenuation_ratio}
\end{equation}
and the attenuation factor $\mathcal R$ is strictly decreasing on $(0,1)$ with
$\lim_{\lambda\uparrow 1}\mathcal R(\lambda)=0$.
\end{theo}

\noindent\emph{Proof in Appendix~\ref{app:proof:thm:attenuation}.}

Under full processing of shocks, the within-date mapping from the projected innovation $\eta_t$ to the channel is amplified by the resolvent weight $(1-\lambda_2)^{-1}$. Differencing does not undo this amplification, because the resolvent multiplier scales $y_t^{\ast,\mathrm{re}}$ and $y_{t-1}^{\ast,\mathrm{re}}$ by the same constant. The static increment variance therefore rises with $\lambda_2$ like $(1-\lambda_2)^{-2}$, up to $(1-\varsigma)^{-2}$ times $2b^2\sigma^2$ within the model, where $\lambda_2<\varsigma$.

The dynamic variance $\Var(\Delta y_t^{\mathrm{re}})$ differs from this benchmark in one respect only, which is that the vintages overlap. With one-step propagation, each vintage is only partially processed when the next arrives, leaving a residual imbalance that interacts with subsequent shocks. When we difference output, this inherited component is subtracted against the new disturbance, so partially processed waves from different vintages interfere and cancel. The factor $\mathcal R$ is the fraction of the fully processed variance that survives this interference. As $\lambda_2$ rises, the dominant transient becomes more persistent, waves live longer, overlap grows, and $\mathcal R$ falls like $(1-\lambda_2)^2$. So $\lambda_2$ pulls in two directions at once, raising the static benchmark through the resolvent and lowering $\mathcal R$ by making vintages overlap for longer. Within the model both forces act over the range $0<\lambda_2<\lambda_1=\varsigma$, because the transient rate is $\varsigma\lambda_2(\mathbf A)$ with $\lambda_2(\mathbf A)<1$. Over that range the static reallocation variance rises from $2b^2\sigma^2$ toward $2b^2\sigma^2/(1-\varsigma)^2$, the dynamic reallocation variance falls from $2b^2\sigma^2$ toward $2b^2\sigma^2/(1+\varsigma)$, and the channel ratio $\mathcal R(\lambda_2)$ falls from one toward the granular rate $\mathcal R(\varsigma)$. When the network mixes slowly, so that $\lambda_2$ is close to $\varsigma$, the reallocation channel survives overlap little better than the granular channel does, and the dominant transient mode is close to invisible in the aggregate ratio. The limit $\lambda_2\uparrow1$ is not reachable at fixed returns to scale, since along any sequence of economies on which $\lambda_2\to1$ the Perron rate $\varsigma$ also tends to one, and statements about slow mixing in this paper are therefore statements about $\lambda_2$ close to $\varsigma$ at fixed $\varsigma$.\footnote{Nothing in Theorem~\ref{thm:attenuation} uses Gaussianity or the centering of the innovations. The ratio compares second moments, so any i.i.d.\ innovation stream with finite variance gives the same $\mathcal R(\lambda_2)$, and a common innovation mean $\mu$ shifts the stationary path of log output by the constant $\mu/(1-\varsigma)$, which cancels from every increment (Lemma~\ref{lem:mean_invariance} in Appendix~\ref{app:reversal}). The attenuation factor, the variance bound of Lemma~\ref{lem:neumann_bounds_L}, and the magnitude bound of Corollary~\ref{coro:reversal_magnitude} in the same appendix therefore apply verbatim to one-sided innovations. The results that compare probabilities rather than second moments use the Gaussian shape of the increments and do not extend, and Appendix~\ref{app:reversal} shows how each can fail. One-sided innovations also slow convergence to the stochastic steady state, since the mean of log output must then travel to its shifted level at the rate $\varsigma$, whereas under centered innovations only second moments converge, at the faster rate $\varsigma^2$.}

\smallskip
The granular (Perron) channel of \eqref{eq:y_perron_transient} time-averages too, and we now record its rate. The granular channel is the scalar
$\mathrm{AR}(1)$ process $\mathbf m^{\ast\top}\sum_{\tau\ge0}\lambda_1^{\tau}\boldsymbol{\epsilon}_{t-\tau}$,
with i.i.d.\ driver $g_t:=\mathbf m^{\ast\top}\boldsymbol{\epsilon}_t$ of variance
$\sigma_g^2=\sigma^2\|\mathbf m^\ast\|_2^2$ (the size Herfindahl).

\begin{coro}[Time-averaging of the granular channel]
\label{coro:granular_perron_rate}
The granular (Perron) channel has dynamic and static increment variances in the ratio
\[
\frac{\Var(\Delta y_t^{\mathrm{gr}})}{\Var(\Delta y_t^{\ast,\mathrm{gr}})}
=
\frac{(1-\varsigma)^2}{1+\varsigma}
=
\mathcal R(\lambda_1),
\]
with $\mathcal R(\cdot)$ the attenuation factor of Theorem~\ref{thm:attenuation}.
\end{coro}

\noindent\emph{Proof in Appendix~\ref{app:proof:coro:granular_perron_rate}.}

\noindent
The rate $\mathcal R(\lambda_1)=(1-\varsigma)^2/(1+\varsigma)$ depends only on returns to scale, a primitive that carries no information about how the network mixes, so two economies with the same returns to scale and entirely different production networks time-average their granular channels identically. The reallocation channel decays instead at the rate at which the network mixes expansions and contractions across firms, which is the one object that distinguishes one production network from another. And because $\lambda_2<\lambda_1$ and $\mathcal R$ is decreasing, the reallocation channel survives time-averaging more than the granular channel does. Degree heterogeneity therefore reaches realized volatility through two routes that time-average at different speeds. Through the granular route it enters as the concentration of the size distribution, which scales the level of that channel's variance but leaves its rate of time-averaging at returns to scale. Through the reallocation route it enters as the rate of time-averaging itself.

\smallskip
The two channels are separable in a stronger sense than the algebra suggests.  Because the spectral
projectors annihilate one another, $\mathbf P_r\mathbf P_s=\delta_{rs}\mathbf P_r$, a disturbance to
the common level of output and a disturbance to the allocation of inputs across firms decay
independently of each other.  The first decays at the Perron rate $\lambda_1$ and the second at the
transient rate $\lambda_2$, and at no date does the network convert one into the other.  A reshuffling of inputs
across firms never becomes a movement of the aggregate level, and a movement of the level never
becomes a reshuffling.  This decoupling of the dynamics, and no assumption about the shocks, makes the two survival rates exact.  Each rate is
the time-averaging of a single autonomous process, with the driver variance and the aggregate loading
canceling inside its own ratio.\footnote{%
What decoupling does not give us is statistical separability of the two channels, and we do not
assume it.  Decoupling is a property of the dynamics (the modal amplitudes evolve autonomously), not of
the innovations.  The projected shocks $g_t=\mathbf m^{\ast\top}\boldsymbol\epsilon_t$ and
$\eta_t=\tilde{\mathbf u}_2^\top\boldsymbol\epsilon_t$ are in general correlated, with
$\Cov(g_t,\eta_t)=\sigma^2\,\mathbf m^{\ast\top}\tilde{\mathbf u}_2$.  Biorthogonality restricts only left-against-right
eigenvectors, $\mathbf u_r^\top\mathbf v_s=\delta_{rs}$, which yields the ``pure reallocation'' identities
$\mathbf u_2^\top\mathbf 1=0$ and $\mathbf m^{\ast\top}\mathbf v_2=0$, but places no restriction on the
left--left inner product $\mathbf u_1^\top\mathbf u_2$ (equivalently $\mathbf m^{\ast\top}\tilde{\mathbf u}_2$, since
$\mathbf m^\ast=\mathbf u_1$ and $\tilde{\mathbf u}_2\propto\mathbf u_2$).  That product vanishes only in non-generic cases, in
particular when $\mathbf W$ is symmetric, which in the present model means a symmetric input-share
matrix.  An undirected network with unequal degrees does not suffice.  Its column-normalized share
matrix is $\mathbf A=\mathbf B\,\mathrm{diag}(\mathbf d)^{-1}$, with $\mathbf B$ the symmetric adjacency
matrix and $\mathbf d$ the degree vector.  The operator
$\mathbf W=\varsigma\,\mathrm{diag}(\mathbf d)^{-1}\mathbf B$ is then similar to a symmetric matrix but is
not normal in the Euclidean inner product under which the innovations are isotropic, so its left
eigenvectors need not be orthogonal.  In a directed production network the hubs that
move the aggregate level are also those whose shocks set off reallocation, so a single draw excites both
modes in correlated amounts.  Total aggregate volatility then carries a level--reallocation cross-term
$\propto b\,(\mathbf m^{\ast\top}\tilde{\mathbf u}_2)$ in addition to the two channel variances.  The per-channel ratios are exact whether or not the projected shocks are correlated. The cross-term enters only if one recombines the channels into a single
aggregate ratio, and it is itself meaningful, being the macroeconomic co-movement, induced by the same shock,
of level shifts and input reallocations.}

%==========================================================
\subsection{Aggregation weights and measured aggregate volatility}
\label{subsec:aggregator_geometry}
%==========================================================

So far we have looked at the reallocation channel on its own and, in
Corollary~\ref{coro:granular_perron_rate}, at the granular channel on its own.  An observer sees both
at once, and how much of each enters the aggregate is decided by the aggregation vector rather than
by the network or by the shocks.  We now work out what the aggregation weights do, for all modes at once.  Recall from Lemma~\ref{lem:modal_decomposition} the modal loadings
$c_r=\boldsymbol\gamma^\top\mathbf v_r$ and the modal innovations
$\xi_{r,t}=\mathbf u_r^\top\boldsymbol{\epsilon}_t$.

\begin{theo}[Exact aggregate variances under any admissible weights]
\label{thm:aggregator_geometry}
Let $\mathbf W$ be diagonalizable, maintain Assumption~\ref{assu:innovations}, fix an admissible weight vector $\boldsymbol\gamma$, write $\mathbf M_r:=c_r\mathbf u_r$, and evaluate all variances under the stationary law. Then
\begin{align*}
\Var(\Delta y_t)
&=
\sigma^2\Bigg(\Big\|\sum_r \mathbf M_r\Big\|_2^2
+\sum_{\tau\ge0}\Big\|\sum_r (1-\lambda_r)\lambda_r^{\,\tau}\mathbf M_r\Big\|_2^2\Bigg),
\\
\Var(\Delta y_t^\ast)
&=
2\sigma^2\Big\|\sum_r \frac{\mathbf M_r}{1-\lambda_r}\Big\|_2^2.
\end{align*}
\end{theo}

\noindent\emph{Proof in Appendix~\ref{app:proof:thm:aggregator_geometry}.}

The two variances hold whether or not the transient eigenvalues beyond the dominant one are complex, because conjugate modes combine into real coefficient vectors $\sum_r(1-\lambda_r)\lambda_r^{\,\tau}\mathbf M_r$. Both variances carry cross-terms between modes, the level--reallocation cross-term of Section~\ref{subsec:dominant_transient_vol} among them, and those cross-terms separate the aggregate ratio of dynamic to static variance from every weighted average of the per-mode factors $\mathcal R(\lambda_r)$.\footnote{The exact variances also say how much a representation that keeps only the granular and the reallocation channels leaves out. Write the increment of the omitted block, the modes $r\ge3$ of \eqref{eq:y_perron_transient}, as the filter $\sum_{k\ge0}\mathbf h_k^\top\boldsymbol{\epsilon}_{t-k}$, with $\mathbf h_0=\sum_{r\ge3}\mathbf M_r$ and $\mathbf h_{k+1}=-\sum_{r\ge3}(1-\lambda_r)\lambda_r^{\,k}\mathbf M_r$, which are real because conjugate modes combine. Write $V_e:=\sigma^2\sum_{k\ge0}\|\mathbf h_k\|_2^2$ for its stationary variance and $V_r$ for the stationary increment variance of the retained granular and reallocation channels. Then $|\Var(\Delta y_t)-V_r|\le V_e+2\sqrt{V_eV_r}$, because the covariance between the retained and the omitted increments is bounded in modulus by $\sqrt{V_eV_r}$, whether or not the retained and omitted innovations are correlated. The spectral gap does not control $V_e$. The own contribution of a real mode $r$ to the dynamic increment variance is $2\sigma^2c_r^2\|\mathbf u_r\|_2^2/(1+\lambda_r)$, and the factor $1/(1+\lambda_r)$ lies between one half and one for every eigenvalue in $(0,1)$, so a small eigenvalue does not make a mode's contribution small. The two-mode representation is accurate when $V_e$ is small relative to $V_r$, which is a condition on the loadings of the omitted modes that a spectral gap alone does not deliver.}

\smallskip
At size weights the exact ratio needs neither diagonalizability nor any assumption on the transient spectrum, and we state it in that generality.

\begin{proposi}[Aggregate volatility under size weights]
\label{prop:singular_aggregator}
Let the innovations $\{\boldsymbol{\epsilon}_t\}$ be i.i.d.\ across dates with $\Cov(\boldsymbol{\epsilon}_t)=\sigma^2\mathbf I$, and let $\boldsymbol\gamma=\mathbf m^\ast$. Then aggregate output is the scalar autoregression driven by the size-weighted innovation $g_t=\mathbf m^{\ast\top}\boldsymbol{\epsilon}_t$,
\[
y_{t+1}=\sum_{\tau=0}^{t}\varsigma^{\,\tau}g_{t-\tau},
\]
and, under the stationary law,
\[
\frac{\Var(\Delta y_t)}{\Var(\Delta y_t^\ast)}
=
\mathcal R(\varsigma)
=
\frac{(1-\varsigma)^2}{1+\varsigma}
\]
exactly.
\end{proposi}

\noindent\emph{Proof in Appendix~\ref{app:proof:prop:singular_aggregator}.}

The first thing to take from Theorem~\ref{thm:aggregator_geometry} and Proposition~\ref{prop:singular_aggregator} is that the granular channel is not optional.  Its loading $c_1$ equals one for every admissible weight vector (Lemma~\ref{lem:modal_decomposition}), because $\mathbf v_1=\mathbf 1$ and weights sum to one, so no choice of
$\boldsymbol\gamma$ can remove the common mode or change its standalone contribution.  That
contribution is $2\sigma^2\|\mathbf m^\ast\|_2^2/(1-\lambda_1)^2$ to the static increment variance and
$2\sigma^2\|\mathbf m^\ast\|_2^2/(1+\lambda_1)$ to the dynamic one.  Its share of either aggregate variance does change with the weights, because the transient loadings and the cross-terms between
modes change.  Only the transient loadings $\{c_r\}_{r\ge2}$ are at the observer's disposal.

\smallskip
The second is that one choice of weights removes the network's transient spectrum from aggregate output.  Proposition~\ref{prop:singular_aggregator} is exact and needs no diagonalizability, no
spectral gap, no Gaussianity, and no assumption on $\lambda_2$.  Weighting firms by their equilibrium
size makes the aggregation vector a left Perron direction of the propagation operator, so aggregate
output is a scalar autoregression at the Perron rate and no transient component is measured at any
horizon.  An observer using size weights therefore recovers
$\mathcal R(\lambda_1)$ for every network, and would conclude that the speed of convergence is
governed by returns to scale alone.  Such an observer would be right about the number and wrong about its source.  The network's
transient spectrum is still there, and it still governs how fast disturbances mix across firms, but a
size-weighted aggregate never loads on it, so nothing about the network's spectrum can show up in the
ratio that observer measures.  The levels of the two variances,
$2\sigma^2\|\mathbf m^\ast\|_2^2/(1+\varsigma)$ and $2\sigma^2\|\mathbf m^\ast\|_2^2/(1-\varsigma)^2$, do
carry the network through the size Herfindahl $\|\mathbf m^\ast\|_2^2$, which multiplies both and
cancels from their ratio.  Equilibrium gross-sales weighting is the convention that our computational experiments adopt.

\smallskip
The third is that between these poles the aggregate ratio of dynamic to static variance is not a weighted average of the per-mode factors, because the exact variances of Theorem~\ref{thm:aggregator_geometry} carry cross-terms between modes. The two channel rates therefore describe the two channels and do not bracket the aggregate ratio at other weights. Assumption~\ref{assu:misalignment} is the requirement that the weights do not lie on the set of weight vectors orthogonal to $\tilde{\mathbf v}_2$, of which size weights are one point.

%==========================================================
\subsection{Tail risk}
\label{subsec:omega_dist_dom}
%==========================================================

We now ask how overlap changes the probability of a large fall in output. Under Assumption~\ref{assu:innovations} each increment is a linear combination of Gaussian innovations and therefore a centered normal random variable, so the probability that output falls by more than $c$ in one step depends on the increment variance alone,
\[
\omega_c^\ast
:=
\Pr(\Delta y_t^\ast<-c)
=
\Phi\!\left(-\frac{c}{\sqrt{\phi^\ast}}\right),
\qquad
\omega_c
:=
\Pr(\Delta y_t<-c)
=
\Phi\!\left(-\frac{c}{\sqrt{\phi}}\right).
\]
Here $\Phi$ is the standard normal distribution function. Since $\phi^\ast>\phi$, the static tail probability exceeds the dynamic one at every threshold. We measure the gap between the two regimes by the variance ratio
\begin{equation}
\kappa
\;:=\;
\frac{\phi}{\phi^\ast}
\;\in\;
(0,1),
\label{eq:kappa_def}
\end{equation}
which for the reallocation channel equals the attenuation factor $\mathcal R(\lambda_2)$ of Theorem~\ref{thm:attenuation}.\footnote{We keep the separate symbol because the tail comparison uses only the variance ratio, and so applies to any pair of centered Gaussian increments with ordered variances, aggregate output under any admissible weights included.} A gap that is linear in $\kappa$ for variances is exponential for tail probabilities.

\begin{proposi}[Tail-risk magnification of the timing wedge]
\label{prop:tail_amp}
Maintain Assumption~\ref{assu:innovations} and the stationary law, and write $x:=c/\sqrt{\phi}$. Then for every $c>0$,
\begin{equation}
\frac{\omega_c^\ast}{\omega_c}
=
\frac{\Phi\!\left(-x\sqrt{\kappa}\right)}
     {\Phi\!\left(-x\right)},
\label{eq:tail_ratio_exact}
\end{equation}
and, as $x\to\infty$,
\begin{equation}
\frac{\omega_c^\ast}{\omega_c}
\;\sim\;
\frac{1}{\sqrt{\kappa}}\,
\exp\left\{\frac{1}{2}x^2\left(1-\kappa\right)\right\}.
\label{eq:tail_ratio_asymp}
\end{equation}
\end{proposi}

\noindent\emph{Proof in Appendix~\ref{app:proof:prop:tail_amp}.}

Variance averages over the whole distribution and is driven by typical fluctuations, whereas a fixed-threshold tail probability isolates rare one-step falls. A fall of size $c$ is a fall of $x$ dynamic standard deviations but of only $x\sqrt{\kappa}$ static standard deviations, since the static increment has the larger variance $\phi/\kappa$, and the Gaussian tail responds to the gap between these two standardized thresholds exponentially in $x^2$. The magnification therefore comes from the two increment variances alone, with no particular sequence of draws needed for it, and even a modest timing wedge is exponentially magnified in tail risk at high thresholds.

% =========================================================
\section{Finite-horizon Distributions of Macroeconomic Fluctuations}
\label{sec:eigenmodes}
% =========================================================

The comparisons of Section~\ref{sec:volatility_leontief} are between population variances, which an observer sees only after an infinitely long history of shocks. An empirical researcher works instead with a window of $T$ growth rates and computes realized volatility from it, and since any one window is a single draw, what the researcher needs to know is not the mean of realized volatility but its distribution. In this section we derive that distribution for a window of any length, under static timing and under overlapping propagation, and we show that the static distribution first-order stochastically dominates the dynamic one. An observer who computes realized volatility over a window of any length is therefore less likely to record a value above any given threshold when shocks overlap in time than when each shock is fully processed before the next one arrives.

We work with the reallocation channel, and write $y_t$ for the channel $y_t^{\mathrm{re}}=b\,s_t$ and $y_t^\ast$ for its static counterpart $y_t^{\ast,\mathrm{re}}$, so that every display in this section is exact. The granular channel is a scalar autoregression of the same form at the Perron rate, with loading one under every admissible weight vector by Lemma~\ref{lem:modal_decomposition}. The argument of this section applies to it with the Perron rate in place of the dominant transient eigenvalue and the size-weighted shock in place of the projected innovation. Written in matrix form over the window, realized volatility is a quadratic form in the Gaussian innovation history, and therefore a weighted sum of independent chi-square variables whose weights are the eigenvalues of a symmetric matrix of order $T-1$. The timing regime enters only through these weights, and the dominance follows once the dynamic weights are ordered below the static ones one by one. We then show that the dominance owes nothing to the dominant-transient reduction, by proving it for aggregate output on any production network of the model and under any admissible weights, using only the nonnegativity of the propagation operator and the Gaussianity of the innovations. The ranking is a ranking of distributions over a window and not of increments date by date. Individual dates at which the dynamic increment exceeds its fully processed counterpart occur at a frequency that Appendix~\ref{app:reversal} computes exactly, and the magnitude of such reversals vanishes in large economies.

\subsection{Distribution of finite-horizon aggregate volatility}

Consider an observer who records $y_t$ at $T$ consecutive dates and computes realized volatility from the $T-1$ growth rates in the window. Each growth rate depends on every innovation up to its date, so the growth rates in the window are correlated, and the distribution of their sum of squares depends on the joint law of the $T$ innovations and not only on the variance of any one of them. The joint law is easiest to handle once the output path over the window is written in matrix form.

Let $\boldsymbol{\eta}:=(\eta_1,\ldots,\eta_T)^\top$ collect the innovations to the channel over the window and $\mathbf s:=(s_1,\ldots,s_T)^\top$ its transient state. The state at each date is the sum of the innovations up to that date, each discounted by $\lambda_2$ raised to its age, and the $T\times T$ lower-triangular matrix $\mathbf K$ collects these discount factors,
\begin{equation}
[\mathbf K]_{tq}
=
\lambda_2^{\,t-q}\mathbf 1_{\{q\le t\}},
\qquad 1\le q,t\le T,
\label{eq:K_def}
\end{equation}
so that
\begin{equation}
\mathbf K
=
\begin{pmatrix}
1 & 0 & 0 & \cdots & 0\\
\lambda_2 & 1 & 0 & \cdots & 0\\
\lambda_2^2 & \lambda_2 & 1 & \cdots & 0\\
\vdots & \vdots & \ddots & \ddots & \vdots\\
\lambda_2^{T-1} & \lambda_2^{T-2} & \cdots & \lambda_2 & 1
\end{pmatrix}.
\label{eq:K_explicit}
\end{equation}
When the window opens on a settled transient state, $s_0=0$, the recursion $s_t=\lambda_2 s_{t-1}+\eta_t$ reads
\begin{equation}
\mathbf s=\mathbf K\,\boldsymbol{\eta},
\label{eq:s_K_eta}
\end{equation}
and the output path over the window is
\begin{equation}
\mathbf y:=(y_1,\ldots,y_T)^\top = b\,\mathbf K\,\boldsymbol{\eta}.
\label{eq:y_K_eta}
\end{equation}
Under static timing each innovation is fully processed before the next one arrives, and the output path is
\begin{equation}
\mathbf y^\ast:=(y_1^\ast,\ldots,y_T^\ast)^\top
=
\frac{b}{1-\lambda_2}\,\boldsymbol{\eta}.
\label{eq:y_star_vector}
\end{equation}
The two paths differ only in the matrix that stands between the innovations and the outputs. Under static timing that matrix is the scalar resolvent gain $b/(1-\lambda_2)$ times the identity, so the output of each date depends on the innovation of that date alone. Under overlapping propagation the innovation of date $q$ enters the output of date $q$ and of every later date in the window, with the geometric weights $\lambda_2,\lambda_2^2,\ldots$, and a column of $\mathbf K$ records this trail. We call $\mathbf K$ the overlap operator. It spreads each vintage across the dates that follow it, so that the output of a date carries a component inherited from earlier vintages, and it is this inherited component that is subtracted away when we pass from the level of output to its growth rate.

Growth rates are first differences of the output path. Let $\mathbf D_T\in\mathbb R^{(T-1)\times T}$ denote the first-difference matrix,
\begin{equation}
\mathbf D_T
=
\begin{pmatrix}
-1 & 1 & 0 & \cdots & 0\\
0 & -1 & 1 & \cdots & 0\\
\vdots & & \ddots & \ddots & \vdots\\
0 & \cdots & 0 & -1 & 1
\end{pmatrix},
\label{eq:DT_matrix}
\end{equation}
so that $\mathbf D_T\mathbf x=(x_2-x_1,\;x_3-x_2,\;\ldots,\;x_T-x_{T-1})^\top$ for any $\mathbf x=(x_1,\ldots,x_T)^\top\in\mathbb R^T$. Realized volatility over the window, in the sense of Definition~\ref{def:agg_vol_static_dynamic}, is the mean of the squared growth rates,
\begin{equation}
\hat\phi
=
\frac{1}{T-1}\,(\mathbf D_T\mathbf y)^\top(\mathbf D_T\mathbf y),
\label{eq:hatphi_def_matrix}
\end{equation}
and substituting the output path \eqref{eq:y_K_eta} gives the quadratic form
\begin{equation}
\hat\phi
=
\frac{1}{T-1}\, b^2
\boldsymbol{\eta}^\top
\Big(\mathbf K^\top \mathbf D_T^\top \mathbf D_T \mathbf K\Big)
\boldsymbol{\eta}.
\label{eq:phi_dyn_quad_main}
\end{equation}
The same steps under static timing give
\begin{align}
\hat\phi^\ast
&=
\frac{1}{T-1}\,(\mathbf D_T\mathbf y^\ast)^\top(\mathbf D_T\mathbf y^\ast) \notag\\
&=
\frac{1}{T-1} \, \frac{1}{(1-\lambda_2)^2}\,b^2\,
\boldsymbol{\eta}^\top(\mathbf D_T^\top \mathbf D_T)\boldsymbol{\eta}.
\label{eq:phi_stat_quad_main}
\end{align}

The two statistics are therefore quadratic forms in the same Gaussian innovation history, and the overlap operator is the only difference between them. All randomness enters through the projected innovations $\eta_t=\tilde{\mathbf u}_2^\top\boldsymbol{\epsilon}_t$, which under the normalization \eqref{eq:u2_normalization} satisfy $\boldsymbol{\eta}\overset{d}{=}\sigma\mathbf Z$ with $\mathbf Z\sim\mathcal N(0,I_T)$. A quadratic form in a standard Gaussian vector is a weighted sum of independent chi-square variables with the eigenvalues of its matrix as weights, and the eigenvalues that matter here are those of the Gram matrices
\[
\mathbf M^\ast
:=
\mathbf D_T\mathbf D_T^\top,
\qquad
\mathbf M
:=
\mathbf D_T\mathbf K\mathbf K^\top\mathbf D_T^\top,
\]
both of order $T-1$, which we write as $\{\nu^*_j\}_{j=1}^{T-1}$ and $\{\nu_j\}_{j=1}^{T-1}$. Both distributions are then available in closed form,
\begin{align}
\hat\phi^\ast
&\ \overset{d}{=}\
\sigma^2 \frac{b^2}{(1-\lambda_2)^2}\cdot\frac{1}{T-1}\sum_{j=1}^{T-1}\nu^*_j\,Z_j^2,
\label{eq:phi_stat_chisq}
\\
\hat\phi
&\ \overset{d}{=}\
\sigma^2 b^2\cdot\frac{1}{T-1}\sum_{j=1}^{T-1}\nu_j\,Z_j^2,
\label{eq:phi_dyn_chisq}
\end{align}
where $\{Z_j\}_{j=1}^{T-1}$ are i.i.d.\ standard normals (Lemma~\ref{lem:phi_quadratic_forms} in Appendix~\ref{app:proofs} records the derivation). The timing regime is thus encoded entirely in the eigenvalues of a single matrix of order $T-1$, and the comparison of the two regimes comes down to a comparison of two sets of weights.

The two sets of weights respond differently to the network. The static weights are the eigenvalues of $\mathbf D_T\mathbf D_T^\top$, which depends on the differencing operator alone, so they depend neither on the production network nor on $\lambda_2$. They equal $2-2\cos(j\pi/T)$ for $j=1,\ldots,T-1$ and run from near zero to near four, and as $\lambda_2$ rises toward one the whole static distribution is scaled up by $(1-\lambda_2)^{-2}$ while its shape stays put. The dynamic weights are the eigenvalues of $\mathbf D_T\mathbf K\mathbf K^\top\mathbf D_T^\top$, and the overlap operator changes their shape as well as their scale. The columns of $\mathbf K$ carry the lag profile $(1,\lambda_2,\lambda_2^2,\ldots)$, and as persistence rises they become more alike, so that $\mathbf K\mathbf K^\top$ builds up correlation across calendar time, which differencing then removes. At the endpoint $\lambda_2=1$, which lies outside the model but fixes the limit of the formulas, differencing removes all of it. There $\mathbf K$ is the lower-triangular matrix of ones, $\mathbf D_T\mathbf K=[\,\mathbf 0\ \ \mathbf I_{T-1}\,]$ and $\mathbf M=\mathbf I_{T-1}$, so every dynamic weight equals one and realized dynamic volatility is $\sigma^2b^2/(T-1)$ times a $\chi^2_{T-1}$ variable. As $\lambda_2$ rises, then, the dynamic weights move from the static profile toward a flat profile at one, while the static weights stay where they are and their common scale grows. Persistence therefore stretches the static distribution through the resolvent gain, while overlap flattens the dynamic one.

Taking expectations in \eqref{eq:phi_stat_chisq} and \eqref{eq:phi_dyn_chisq} shows that the static statistic is unbiased at every window length, $\E[\hat\phi^\ast]=\phi^\ast$, whereas the dynamic statistic started from a settled state has a mean below $\phi$ that rises to $\phi$ as the window lengthens. A comparison of means, however, is what Section~\ref{sec:volatility_leontief} already gave for a long window. The observer with a single window needs the two distributions ordered, so that the dynamic statistic is less likely than the static one to exceed any given value, and this is what the ordering of the weights delivers. The window may also open on a transient state drawn from its stationary law, which is how an observer who starts recording at an arbitrary date finds the economy, and the ordering holds there too.

\begin{theo}[Stochastic dominance of static over dynamic realized volatility at every window length]
\label{thm:phi_finite_ordering}
Fix $T\ge 2$, maintain Assumption~\ref{assu:innovations}, and let $\nu_{(1)}\ge\cdots\ge\nu_{(T-1)}$ and $\nu^\ast_{(1)}\ge\cdots\ge\nu^\ast_{(T-1)}$ be the eigenvalues of $\mathbf M$ and $\mathbf M^\ast$ in descending order. Then
\begin{equation}
\nu_{(j)}\ \le\ \bigl(1-\lambda_2^{\,T}\bigr)^{2}\,\frac{\nu^\ast_{(j)}}{(1-\lambda_2)^2}
\ <\ \frac{\nu^\ast_{(j)}}{(1-\lambda_2)^2}
\qquad\text{for all } j=1,\ldots,T-1,
\label{eq:finiteT_fosd_condition}
\end{equation}
and consequently
\[
\hat\phi^\ast \succeq_{\mathrm{FOSD}} \hat\phi.
\]
The dominance also holds when the initial state $s_0$ is drawn from the stationary law $\mathcal N\bigl(0,\sigma^2/(1-\lambda_2^2)\bigr)$ independently of the innovations, in which case $\E[\hat\phi]=\phi$ for every $T\ge2$.
\end{theo}

\noindent\emph{Proof in Appendix~\ref{app:proof:thm:phi_finite_ordering}.}

The dominance is strict at every window length, and the gap between the two distributions widens as mixing slows. The rescaled static weights $\nu^\ast_{(j)}/(1-\lambda_2)^2$ grow like $(1-\lambda_2)^{-2}$, since the static weights never exceed four, whereas the dynamic weights stay bounded, the largest of them satisfying $\nu_{(1)}=\|\mathbf D_T\mathbf K\|_2^2\le 4/(1+\lambda_2)^2$.\footnote{The
matrix $\mathbf D_T\mathbf K$ carries the Toeplitz symbol
$(e^{i\theta}-1)/(1-\lambda_2 e^{-i\theta})$, whose squared modulus
$(2-2\cos\theta)/(1-2\lambda_2\cos\theta+\lambda_2^2)$ is decreasing in $\cos\theta$ and therefore
attains its supremum $4/(1+\lambda_2)^2$ at $\theta=\pi$.  The norm of a finite section of a
Toeplitz operator is bounded by the supremum of its symbol, which gives the bound at every window
length.} Persistence enlarges the component of the level that a date inherits from earlier vintages, differencing removes that component, and slow mixing therefore leaves the dynamic weights flat and bounded while it inflates the static scale. When mixing is fast the two timing regimes are close to each other, and as $\lambda_2\to 0$ the slack $\lambda_2^{\,T}(2-\lambda_2^{\,T})$ in \eqref{eq:finiteT_fosd_condition} vanishes and the two distributions coincide.

Dominance over a window is compatible with reversals at individual dates. A reversal is a one-step, pathwise event in which the dynamic increment exceeds its fully processed counterpart, and Proposition~\ref{prop:reversal_frequency} in Appendix~\ref{app:reversal} gives its frequency in closed form, free of $b$, $\sigma$, and $n$, running from about $0.11$ to $0.06$ over the calibrated range of $\lambda_2$. An ordering of the two window distributions cannot exclude such dates, because the Loewner comparison behind it holds for the Gram matrices $\mathbf M$ and $\mathbf M^\ast$ and not for the quadratic forms in $\boldsymbol\eta$ themselves, as the proof shows. The ordering is between the distributions of realized volatility over a window of length $T$, and it leaves room for realized windows in which the dynamic statistic is the larger one. Along an innovation history close to a nonzero constant vector, for instance, the static differences are near zero while the dynamic differences from a settled start are not, and such histories have positive probability under Assumption~\ref{assu:innovations}. Appendix~\ref{app:reversal} shows that in large economies the reversals are of vanishing magnitude.

\subsection{Dominance of static over dynamic volatility on any production network}
\label{subsec:full_network_fosd}

Everything in this section so far has used the dominant-transient reduction. The reduction did one job, which was to write the weights of the two chi-square forms in closed form, and the dominance itself owes nothing to it. The ordering rests on features of the model that hold on any production network. The exposure $\boldsymbol\gamma^\top\mathbf W^{k}$ of the aggregate to a vintage $k$ dates old is entrywise nonnegative, so the exposures of the vintages present at a date add up to the fully processed exposure without cancellation. And Gaussian innovations turn an ordering of covariance matrices into an ordering of weighted chi-square laws. On nonnegativity and Gaussianity alone we now prove the dominance for aggregate output itself, with every mode of the network retained and under any admissible weights. The dynamic aggregate is $y_t=\boldsymbol\gamma^\top\sum_{k\ge0}\mathbf W^{k}\boldsymbol{\epsilon}_{t-k}$, with dates relabeled by the one-period production lag of Definition~\ref{def:aggregate_output} so that $y_t$ carries $\boldsymbol{\epsilon}_t$, and the static aggregate is $y_t^\ast=\boldsymbol\gamma^\top(\mathbf I-\mathbf W)^{-1}\boldsymbol{\epsilon}_t$.

\begin{theo}[Stochastic dominance of static over dynamic realized volatility on any production network]
\label{thm:full_network_fosd}
Fix an admissible weight vector $\boldsymbol\gamma$ and a window length $T\ge2$, and maintain Assumption~\ref{assu:innovations}. Let the dynamic aggregate start from a settled configuration, $\boldsymbol{\epsilon}_t=\mathbf 0$ for $t\le0$, or from the stationary law. Then
\[
\hat\phi^\ast\ \succeq_{\mathrm{FOSD}}\ \hat\phi .
\]
\end{theo}

\noindent\emph{Proof in Appendix~\ref{app:proof:thm:full_network_fosd}.}

Beyond Assumption~\ref{assu:innovations}, the hypotheses of the theorem are those of the model itself. Nothing in the proof uses diagonalizability, a real transient spectrum, a spectral gap, or the reduction to a dominant mode, so the dominance covers a defective propagation operator and any network too large for its diagonalizability to be checked. The sign and phase of the transient eigenvalues never enter either, because primitivity makes the exposures positive whatever the transient spectrum, so the dominance covers negative and complex subdominant eigenvalues. Assumption~\ref{assu:misalignment} enters nowhere, so the dominance holds at size weights, where the reallocation channel is absent and Proposition~\ref{prop:singular_aggregator} makes aggregate output a scalar autoregression at the Perron rate. At those weights the theorem is the finite-window form of that proposition. What the general statement gives up is the closed form of the weights, which Theorem~\ref{thm:phi_finite_ordering} supplies for the reallocation channel.

%====================================================
\section{Fat Tails and Aggregate Volatility When Shocks Overlap}
\label{sec:spectral_gap_mediates_fat_tails}
%====================================================
Static aggregate volatility bounds its dynamic counterpart from above, on average by Lemma~\ref{lem:neumann_bounds_L} and Theorem~\ref{thm:attenuation} and in distribution by Theorems~\ref{thm:phi_finite_ordering} and~\ref{thm:full_network_fosd}. When shocks arrive as a stream, successive vintages overlap and partly cancel, and the aggregate is less volatile than it would be if each draw were fully processed before the next one arrived. We now ask what this means for the granular origins of aggregate fluctuations. We first compute, back of the envelope, how much of observed aggregate volatility granular shocks generate once they average over time, under the degree heterogeneity found in actual firm networks. The two leading channels through which the shocks reach the aggregate carry between about $1.5$ and $12$ percent of observed aggregate volatility, measured as standard deviations, where the static calculation puts them at between a tenth and a third. We then ask how the thickness of the tail of the degree distribution and the speed of convergence to equilibrium interact in generating aggregate volatility, and show that fatter tails need not raise it. Fat tails raise the exposure of the aggregate to the cross-section of shocks, and they slow the economy's convergence to equilibrium, and once shocks overlap the slower convergence works against the greater exposure.

The calculation does not require the aggregation weights, because Theorem~\ref{thm:aggregator_geometry} gives a rate for each channel and Proposition~\ref{prop:singular_aggregator} an exact value at one choice of weights. Degree heterogeneity reaches aggregate volatility through both channels. The granular channel carries it as the size Herfindahl $\|\mathbf m^\ast\|_2^2$ and averages over time at the rate $\mathcal R(\lambda_1)=(1-\varsigma)^2/(1+\varsigma)$ of Corollary~\ref{coro:granular_perron_rate}, which depends on returns to scale alone. The reallocation channel carries it through the transient loading $b$ and the dominant transient eigenvalue $\lambda_2$, and averages over time at the rate $\mathcal R(\lambda_2)$. Because $\mathcal R$ is decreasing and $\lambda_2<\lambda_1$, the two rates are ordered,
\begin{equation}
\mathcal R(\lambda_1)\ <\ \mathcal R(\lambda_2).
\label{eq:wedge_bracket}
\end{equation}
At size weights the aggregate ratio $\phi/\phi^\ast$ equals the granular rate exactly, by Proposition~\ref{prop:singular_aggregator}. Away from size weights it is the exact expression of Theorem~\ref{thm:aggregator_geometry}, which carries cross-terms between modes, so the two rates describe the two channels taken separately and not the aggregate.\footnote{The level--reallocation cross-term of Section~\ref{subsec:dominant_transient_vol} moves the aggregate ratio away from the granular rate, and the modes beyond the second, which attenuate less severely since $\mathcal R$ is decreasing, move it above $\mathcal R(\lambda_2)$ when the weights load on them. The aggregate is known exactly only at size weights, where it coincides with the granular rate.} Both rates are computable from primitives, and we report each as the attenuation that its channel applies to whatever static contribution it carries.

%----------------------------------------------------
\subsection{Back-of-the-envelope estimates of aggregate volatility}
\label{subsec:boe_vol}
%----------------------------------------------------

We evaluate both rates and report the range between them. The Perron rate is returns to scale, $\lambda_1=\varsigma$, and we take $\varsigma=0.8$, the value we use in our computational experiments. The dominant transient rate is $\lambda_2=\varsigma\lambda_2(\mathbf A)$, where $\lambda_2(\mathbf A)$ is the subdominant eigenvalue of the monetary circuit, which we measure directly on reconstructed US production networks in Appendix~\ref{app:abm_eigenvalue}. There it clusters in $[0.7,1.0)$, the upper end excluded by the primitivity of $\mathbf A$, so the transient rate for output lies in $[0.56,0.80)$, and over that range
\[
\mathcal R(\lambda_1)=0.022,
\qquad
\mathcal R(\lambda_2)\ \in\ (\,0.022,\ 0.124\,],
\qquad\text{equivalently}\qquad
\sqrt{\mathcal R(\lambda_1)}=0.15,
\qquad
\sqrt{\mathcal R(\lambda_2)}\ \in\ (\,0.15,\ 0.35\,].
\]
A forty-fifth of the static increment variance carried by the granular channel survives intertemporal averaging, and between a forty-fifth and an eighth of the static increment variance carried by the reallocation channel survives.\footnote{The two rates mean different things.  The granular rate is the aggregate ratio at size weights, exactly and for every network (Proposition~\ref{prop:singular_aggregator}).  The reallocation rate describes what overlap does to the reallocation channel on its own.  The aggregate approaches it only as the transient loading $b$ comes to dominate the granular channel and the cross-term between the two channels, which requires weights strongly misaligned with the size distribution.}

\smallskip
Empirical estimates of the static granular contribution exist, and they vary substantially. They vary because static aggregate volatility is very sensitive to degree heterogeneity, as is easiest to see when heterogeneity is summarized by a power-law tail exponent \citep[Section 3]{mandel2023voltest}, and because estimates of that exponent differ across datasets and estimation procedures. For the United States, reported exponents vary widely enough that the implied static contribution ranges from nearly all of observed aggregate volatility to nearly none.\footnote{See Table~5 in \citet{mandel2023voltest} and the discussion in \citet{bacilieri2023}.}

\smallskip
One of the most detailed US estimates currently available is that of \citet{mandel2023voltest}, who use firm-level buyer--seller network data with an order of magnitude more firms and connections than most other studies. Their range of tail exponents implies static contributions of between one-tenth and one-third of observed aggregate volatility, measured, like the dynamic contributions we now compute, as shares of standard deviations.

We read a published static share as carried entirely by one channel, the one the observer's weights expose, and we treat this reading as a scenario rather than an identification. An undivided static estimate does not separate the two channels, it need not refer to the aggregation convention under which a channel is exposed, and in a directed network the cross-term between the two channels prevents adding one channel's contribution to the other's. The calculations that follow are therefore conditional on both the measurement convention and the channel assignment.

Under this scenario the overlap-adjusted contribution of a channel is the static share multiplied by the channel's standard-deviation rate, $0.15$ for the granular channel and between $0.15$ and $0.35$ for the reallocation channel, so that overlap lowers the implied contribution by a factor of between three and seven. The granular channel then accounts for between about $1.5$ and $5$ percent of observed aggregate volatility, and the reallocation channel for between about $1.5$ and $12$ percent, the lower end open because $\lambda_2<0.8$ and approached only as the transient rate approaches the Perron rate. Granular shocks thus account for between almost none of observed aggregate volatility and roughly an eighth of it, where the static calculation puts them at between a tenth and a third.

Under equilibrium gross-sales weights, the weights of our computational experiments, the granular range is the whole of the aggregate, since the reallocation channel is then invisible. The ends of the reallocation range pair the smallest static share with the strongest attenuation and the largest share with the weakest. And because the tail exponent moves exposure and mixing together, as we now show, the two are not independent, and an actual economy is likely to attain a narrower range.

%----------------------------------------------------
\subsection{Fatter tails need not raise aggregate volatility}
\label{subsec:fat_tails}
%----------------------------------------------------

In the static benchmark a thicker tail of the degree distribution can move aggregate volatility a great deal, because degree heterogeneity directly reshapes the exposure of the aggregate to the cross-section of shocks. Overlap mutes that sensitivity. On any fixed network, overlap attenuates what fully processed exposure contributes to time-averaged fluctuations, so a given change in exposure produces a smaller change in volatility. And once the thickness of the tail also sets the speed at which the economy converges, the effect of fatter tails on dynamic volatility can go either way, because the tail then moves volatility through persistence as well as through exposure.

\smallskip
We now parameterize degree heterogeneity by a truncated power law with tail exponent $\alpha$, so that a lower $\alpha$ means a fatter tail, and adopt the degree proxy of Assumption~\ref{assu:degree_proxy} in Appendix~\ref{app:fat_tails_degree_approx}. Under the proxy the stationary balance vector is the degree vector, as it is under locally even weights with matched in- and out-degrees. The degree vector, centered at its realized mean on the left and at its size-weighted mean on the right and normalized by the realized degree variance, stands in for the dominant transient pair.\footnote{The proxy is an assumption, because no eigenvector equation delivers the transient direction from the degree distribution, and Appendix~\ref{app:fat_tails_degree_approx} gives a family of networks with matched degrees and even weights on which the centered degree vector lies in modes faster than the dominant one. None of the exact channel results depends on the proxy, which enters only the comparative statics of this section.} With the population moments of Proposition~\ref{prop:fat_tail_moments_and_proxy} in Appendix~\ref{app:fat_tails_degree_approx} in place of the realized moments of \eqref{eq:utilde_vtilde_alpha_app} there, the scalar loading
\[
b(\alpha):=\boldsymbol\gamma^\top \tilde{\mathbf v}_2(\alpha)
\]
admits the approximation
\begin{equation}
b(\alpha)
\ \approx\
\frac{\Delta(\alpha)}{\sqrt{n\,\Var(d)}},
\qquad
\Delta(\alpha):=\E_{\boldsymbol\gamma}[d]-\frac{\E[d^2]}{\E[d]},
\label{eq:b_alpha_explicit_fattails}
\end{equation}
where $\E_{\boldsymbol\gamma}[d]:=\boldsymbol\gamma^\top\mathbf d$ is the mean degree under the observer's weights. The gap $\Delta(\alpha)$ measures the misalignment of the weights with firm size in degrees, comparing the degree of the typical firm as the observer weights it with the degree of the typical firm as the economy weights it. It vanishes at $\boldsymbol\gamma\propto\mathbf d\propto\mathbf m^\ast$, the size weights of Proposition~\ref{prop:singular_aggregator}, and more generally on the set of weight vectors whose degree-weighted mean equals $\E[d^2]/\E[d]$. That set is the degree-proxy counterpart of the set of weight vectors orthogonal to $\tilde{\mathbf v}_2$ in Assumption~\ref{assu:misalignment}.

\smallskip
Fatter tails move both the misalignment gap and the degree dispersion, and the log-derivative of $b(\alpha)$ keeps both,
\begin{equation}
\frac{b'(\alpha)}{b(\alpha)}
\ \approx\
\frac{\Delta'(\alpha)}{\Delta(\alpha)}
\;-\;
\frac{1}{2}\,\frac{\Var(d)'(\alpha)}{\Var(d)}.
\label{eq:log_deriv_b_alpha}
\end{equation}
The first term is the change in the gap between the observer's average degree and the economy's, and the second is the change in the dispersion by which that gap is scaled. Fatter tails raise both, so the two terms work against each other, and the sign of the derivative turns on which rises faster. The sign is not the same for every weight vector. A weight vector held fixed while $\alpha$ moves can be the size weights of some tail exponent $\alpha_0$, at which $b(\alpha_0)=0$, and $|b|$ then rises on both sides of $\alpha_0$, so that on one side exposure increases as the tail thins.

We therefore state the exposure channel for a family of weight vectors that moves with the degree distribution, $\gamma_i\propto d_i^{\,\theta}$, under which $\E_{\boldsymbol\gamma}[d]=\E[d^{1+\theta}]/\E[d^{\theta}]$. Fattening the tail raises every moment of the degree distribution, and for $\theta>1$ the hub-tilted weights load on moments of higher order than the second moment in $\Var(d)$, which rise faster. Under the cutoff normalization of Proposition~\ref{prop:fat_tail_moments_and_proxy} the comparison is exact in one case. Suppose $\alpha>2$, so that the degree variance converges as the economy grows, and $\theta>\alpha-1$, so that the observer's mean degree $\E[d^{1+\theta}]/\E[d^{\theta}]$ is carried by a moment that grows with $n$. Then the gap $\Delta(\alpha)$ grows like $n^{e(\alpha,\theta)}$, with $e(\alpha,\theta)=(1+\theta-\alpha)/\alpha$ for $\theta<\alpha$ and $e(\alpha,\theta)=1/\alpha$ for $\theta\ge\alpha$, while the dispersion $\sqrt{\Var(d)}$ converges. The exponent $e$ falls as $\alpha$ rises, so the loading $|b(\alpha)|$ decreases in $\alpha$ in every sufficiently large economy. Outside this case the sign of $b'/b$ is a finite-sample question, and we impose it as a restriction.\footnote{At $\theta=1$ the weights are the size distribution itself, $\Delta\equiv0$,
and the loading is identically zero at every $\alpha$, as Proposition~\ref{prop:singular_aggregator}
requires.  As $\theta$ approaches one from either side the loading vanishes at every $\alpha$, so weights close to size weights inherit a weak exposure channel.} The exposure channel is the restriction $(b^2)'(\alpha)<0$, equivalently $b'/b<0$ wherever $b\neq0$, so that the exposure scale $b(\alpha)^2$ rises as the tail fattens. The sign of $b'(\alpha)$ on its own says nothing about exposure, because $b$ is negative for $\theta<1$ and the sign of the eigenvector pair is a convention.

\smallskip
Holding $\lambda_2$ fixed, tail thickness enters both regimes only through the common exposure scale $b(\alpha)^2$, since by Theorem~\ref{thm:attenuation} the two population variances are $\phi^\ast(\alpha)=2\sigma^2b(\alpha)^2/(1-\lambda_2)^2$ and $\phi(\alpha)=2\sigma^2b(\alpha)^2/(1+\lambda_2)$. Both therefore respond to $\alpha$ through the same exposure elasticity,
\begin{equation}
\frac{\partial}{\partial \alpha}\phi^\ast
=
2\,\frac{b'(\alpha)}{b(\alpha)}\,\phi^\ast,
\qquad
\frac{\partial}{\partial \alpha}\phi
=
2\,\frac{b'(\alpha)}{b(\alpha)}\,\phi,
\label{eq:mean_sensitivity_common}
\end{equation}
and the dynamic economy moves less in absolute terms because overlap has already lowered the level of volatility,
\begin{equation}
\Bigl|\partial_\alpha \phi\Bigr|
=
\mathcal R(\lambda_2)\,
\Bigl|\partial_\alpha \phi^\ast\Bigr|.
\label{eq:absolute_change_smaller}
\end{equation}
We now let $\lambda_2$ move with $\alpha$. Tail thickness then reshapes not only exposure, through $b(\alpha)$, but also the speed of convergence, through $\lambda_2(\alpha)$, and the monotone static ranking need not carry over to the dynamic economy.

\subsubsection{Interplay between the power-law tail and the rate of convergence to equilibrium}
\label{subsubsec:alpha_lambda2_interplay}
So far we have treated the tail exponent $\alpha$ and the dominant transient eigenvalue $\lambda_2$ as if they had nothing to do with each other. In that reading $\alpha$ told us how exposed the aggregate is to the cross-section of shocks, and $\lambda_2$ told us how strongly one vintage overlaps with the next. But in an actual production network the two are not independent, because the same structural features that generate degree concentration also shape the speed with which the network mixes.

\smallskip
Spectral theory and simulations on heterogeneous graphs both suggest that more unequal degrees go with slower convergence. Even when the graph is connected, flow between large regions may pass through a thin set of connectors, and when propagation runs mainly through hubs, the hubs become bottlenecks and disturbances take longer to spread through the economy. Appendix~\ref{app:fat_tails_degree_approx} summarizes the theoretical links between bottlenecks, spectral gaps, and mixing rates, together with results on large random graphs that support this direction of association. Guided by that evidence, we let mixing move with tail thickness, write $\lambda_2=\lambda_2(\alpha)$, and impose the reduced-form sign restriction
\begin{equation}
\lambda_2'(\alpha)<0,
\label{eq:lambda2_prime_sign_fattails}
\end{equation}
so that fatter tails (lower $\alpha$) go with slower convergence (higher $\lambda_2$).

\smallskip
Once $\lambda_2$ moves with $\alpha$, the dynamic economy is no longer a rescaling of the static benchmark, and the stationary formulas of Theorem~\ref{thm:attenuation} give the comparison in closed form. The reallocation increment variances are
\[
\phi(\alpha)=\frac{2\sigma^2 b(\alpha)^2}{1+\lambda_2(\alpha)},
\qquad
\phi^\ast(\alpha)=\frac{2\sigma^2 b(\alpha)^2}{(1-\lambda_2(\alpha))^2}.
\]
In the static benchmark the tail exponent enters through the exposure scale $b(\alpha)^2$ and through the resolvent gain $(1-\lambda_2)^{-2}$, and under the sign restrictions both rise as $\alpha$ falls. In the dynamic economy the tail exponent enters through the same exposure scale and through the overlap factor $1/(1+\lambda_2)$, which falls as $\alpha$ falls. We now record the resulting condition and a bound on how much the overlap factor can do at fixed returns to scale.

\begin{proposi}[Fatter tails and dynamic volatility]
\label{prop:fat_tails_condition}
Let $b(\alpha)$ and $\lambda_2(\alpha)$ be differentiable on an interval of tail exponents, with
$b(\alpha)\neq0$, $(b^2)'(\alpha)<0$, and $\lambda_2'(\alpha)<0$, and hold $\sigma^2$ fixed.
\begin{itemize}
\item[(i)] The static and dynamic elasticities are
\[
\frac{d\log\phi^\ast}{d\alpha}
=
2\frac{b'}{b}+\frac{2\lambda_2'}{1-\lambda_2}<0,
\qquad
\frac{d\log\phi}{d\alpha}
=
2\frac{b'}{b}-\frac{\lambda_2'}{1+\lambda_2},
\]
so lowering $\alpha$ lowers the dynamic variance $\phi$ if and only if
$-\lambda_2'/(1+\lambda_2)>-2b'/b$.
\item[(ii)] For two economies with the same returns to scale $\varsigma$ and reallocation loadings $b_1$ and $b_2$,
$\phi_2/\phi_1\ge b_2^2/\bigl((1+\varsigma)b_1^2\bigr)$, so that $\phi_2>\phi_1$ whatever the two transient eigenvalues whenever $b_2^2>(1+\varsigma)b_1^2$.
\end{itemize}
\end{proposi}
\noindent\emph{Proof in Appendix~\ref{app:proof:prop:fat_tails_condition}.}

\smallskip
Fatter tails therefore need not mean more aggregate volatility once shocks overlap in time. In the static benchmark exposure and persistence reinforce each other, because the resolvent gain rises with $\lambda_2$. In the dynamic economy they oppose each other, because the interference that removes the inherited component strengthens with $\lambda_2$. Whether the reallocation channel's contribution to observed aggregate volatility rises or falls as tails fatten depends on the persistence elasticity $-\lambda_2'/(1+\lambda_2)$ against the exposure elasticity $-2b'/b$. Persistence can do only so much at fixed returns to scale. The overlap factor $1/(1+\lambda_2)$ ranges over $(1/(1+\varsigma),1)$, so an increase of $|b|$ by more than $\sqrt{1+\varsigma}$, which is $1.34$ at $\varsigma=0.8$, raises dynamic volatility whatever happens to mixing, as part (ii) of Proposition~\ref{prop:fat_tails_condition} records.\footnote{The opposition between the two channels is specific to increments. In levels the dynamic variance is $\sigma^2b^2/(1-\lambda_2^2)$ by Lemma~\ref{lem:levels_increments_translation} in Appendix~\ref{app:levels_vs_increments}, so its elasticity $2b'/b+2\lambda_2\lambda_2'/(1-\lambda_2^2)$ has two negative terms and fatter tails raise dynamic level variance through both channels.}

%====================================================
\section{Computational Experiments on the Reconstructed US Production Network}
\label{sec:abm}
%====================================================

The closed forms of the preceding sections rest on assumptions that an actual production network need not satisfy. We assumed that the propagation operator is diagonalizable, that its dominant transient eigenvalue is real, simple, and separated from the rest of the spectrum by a gap, and that aggregate dynamics reduce to the two leading modes, and every formula took the balances of the money circuit as settled. We now relax all of these assumptions at once. We build the economy of Section~\ref{sec:model} in silico on reconstructed US production networks and run its full adjustment dynamics, with the money circuit, the market-clearing prices, and every eigenmode retained, and we read volatility and its entire distribution off the realized paths rather than from any formula. The experiments are agent-based in the sense that each firm spends its balance on inputs in its Cobb--Douglas shares, posts the price that clears its market, and produces from the inputs delivered, and aggregate output is whatever these decisions add up to.

Throughout, we weight firms by their equilibrium gross sales, the size weights $\boldsymbol\gamma=\mathbf m^\ast$ of Proposition~\ref{prop:singular_aggregator}, taken from the settled baseline balances and held fixed, the same vector in the static and in the dynamic series. Aggregate volatility is the standard deviation of the one-step growth increments of aggregate output, the square root of the variances compared in the preceding sections. At these weights the theory makes a prediction with nothing left to fit. By Proposition~\ref{prop:singular_aggregator}, aggregate output is a scalar autoregression at the Perron rate whatever the Jordan structure of the operator, so the ratio of dynamic to static volatility equals $\sqrt{\mathcal R(\varsigma)}=(1-\varsigma)/\sqrt{1+\varsigma}$ for every network, which is $0.149$ at the calibrated returns to scale. Once balances have settled the log-output recursion \eqref{eq:logq_recursion} is exactly linear, so the ratio is an identity of the simulated economy, and the experiments ask whether the identity survives the full economy on reconstructed topologies across nearly three decades of size. They also add what the identity does not fix, namely the level of static and dynamic volatility at each economy size and the dispersion of the measured ratio over finite windows.

We proceed from one network to an ensemble. On a single reconstructed network we trace the step-to-step movements of dynamic aggregate output against the static output generated by the same productivity draws, and find that dynamic output moves far less from one step to the next. We then measure how the two volatilities decay as the economy grows, and find their ratio stable across nearly three decades in the number of firms. And we repeat the comparison across an ensemble of independent reconstructions at two sizes, on which every reconstruction returns the same ratio. Appendix~\ref{app:abm_eigenvalue} records the realized subdominant eigenvalue of the monetary circuit across the reconstructions, which is the source of the calibrated range of Section~\ref{subsec:boe_vol}, and shows that the ratio is unmoved by whether that eigenvalue is real, complex, or negative.

\smallskip
The networks are reconstructed with the method of \citet{bhattathiripad2026reconstructing}, which generates a firm-to-firm network with a chosen number of firms and supplies its edge weights. The weights enter as each firm's Cobb--Douglas intermediate-input shares, with self-loops dropped and the shares renormalized so that each firm's expenditure shares sum to one. Returns to scale are $\varsigma=0.8$, so output closes about a fifth of its distance to equilibrium each step. Firm productivities are drawn i.i.d.\ as $z_i=\exp(\epsilon_i)$ with $\epsilon_i\sim\mathcal N(0,\sigma^2)$ and $\sigma=0.10$. The scale is illustrative, and the volatility ratios do not depend on it.\footnote{Firm-level volatility, measured by the standard deviation of firm sales growth, is on the order of $10\%$ \citep{comin2005rise,thesmar2011contrasting,davis2006volatility}. We take that order of magnitude as the illustrative scale without identifying the two objects, since nominal sales in this economy equal the settled balances $\mathbf m^\ast$ at every date and do not respond to productivity once balances have settled.} Two output series, one static and one dynamic, are built from one shared sequence of draws. Each network is first relaxed over $200$ shock-free steps. The dynamic series then advances the economy one step per draw and records output at every date. The static series takes each draw in turn and runs the economy in isolation for $40$ steps, by which point the mean absolute relative price change has fallen below $10^{-4}$ and output has reached the neighborhood of the equilibrium of that draw. The two series differ only in timing.

%----------------------------------------------------
\subsection{Localized movements of dynamic output}
%----------------------------------------------------

We begin with a single reconstructed network of about $10^5$ firms and $100$ draws, one productivity draw per firm at each step of the run. Figure~\ref{fig:ts} plots the two output series, each normalized by its own mean. The static path (grey) swings across roughly $[0.88,1.10]$, while the dynamic path (black) stays bunched in $[0.97,1.02]$, since dynamic output never travels the full distance to each draw's equilibrium before the next draw arrives. The standard deviation of dynamic growth increments is about $0.15$ times that of static increments, against the $0.149$ predicted at the calibrated returns to scale.

\begin{figure}[H]
\centering
\includegraphics[width=\textwidth]{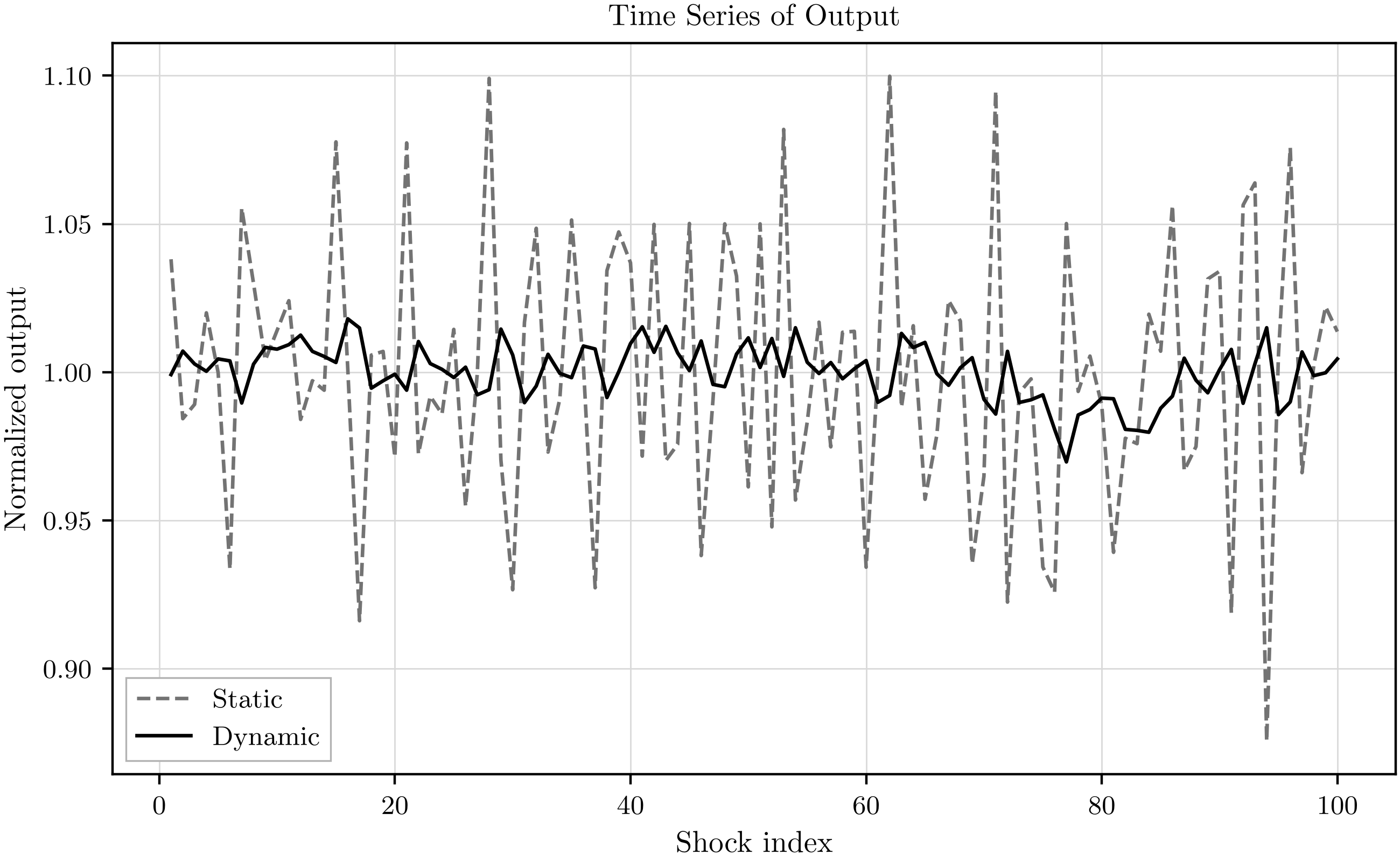}
\caption{Static and dynamic aggregate output for $100$ productivity draws on a reconstructed US
production network of about $10^5$ firms, each series normalized by its own mean.}
\label{fig:ts}
\end{figure}

The levels in the figure compress less than the increments do, and under the granular substitution Lemma~\ref{lem:levels_increments_translation} in Appendix~\ref{app:levels_vs_increments} says that they must. The dispersion of dynamic levels is $\sqrt{(1-\varsigma)/(1+\varsigma)}=0.33$ times that of static levels, while the increment ratio is $0.149$, because differencing removes the inherited component that persistence keeps in the level.

\smallskip
The dampening does not hold at every step. Figure~\ref{fig:ts_zoom} isolates three windows, around steps $43$--$44$, $58$--$59$, and $89$--$90$, in which the dynamic series moves more than the static one. In each window two consecutive draws happen to define nearly coincident static equilibria, so that the static change is close to zero by construction, while the dynamic series, still lagging an earlier draw, keeps drifting toward equilibrium and its change remains appreciably nonzero.

These windows are the finite-horizon reversals of Appendix~\ref{app:reversal}, seen through size weights. Since $\boldsymbol\gamma=\mathbf m^\ast$ sets $b=0$, the relevant reduction is the granular one, obtained from \eqref{eq:overshoot_core_f} under the substitution $(b,\lambda_2,\eta_t)\mapsto(1,\lambda_1,g_t)$. A dynamic increment then exceeds its fully processed counterpart when the difference of two adjacent size-weighted innovations is small relative to the combination of the inherited common-mode state and the new innovation. Proposition~\ref{prop:reversal_frequency} says how often such configurations occur once the run has moved away from its settled start. The reversal frequency is free of $b$, $\sigma$, and $n$, and under the granular substitution it equals $p(\lambda_1)=p(0.8)\approx 0.057$ at the calibrated returns to scale, about one step in eighteen. The three circled windows are the most visible of these episodes in this run.

The prediction presumes the alignment of Section~\ref{subsec:dominant_transient_vol}. Under that alignment the dynamic increment at a step is paired with the static increment of the draw that first enters dynamic output at that step. Pairing the two series one date apart changes the covariance of the pair, and with it the frequency, so a count of realized reversals must use the same pairing. The fat-tailed size distribution does not disturb the prediction. The granular innovation $g_t=\mathbf m^{\ast\top}\boldsymbol{\epsilon}_t$ has standard deviation $\sigma\|\mathbf m^\ast\|_2$, which need not vanish as $n$ grows, but the frequency is scale-free, and only the magnitude of reversals, governed by Corollary~\ref{coro:reversal_magnitude} of Appendix~\ref{app:reversal}, depends on the scale of the innovations.

\begin{figure}[H]
\centering
\includegraphics[width=\textwidth]{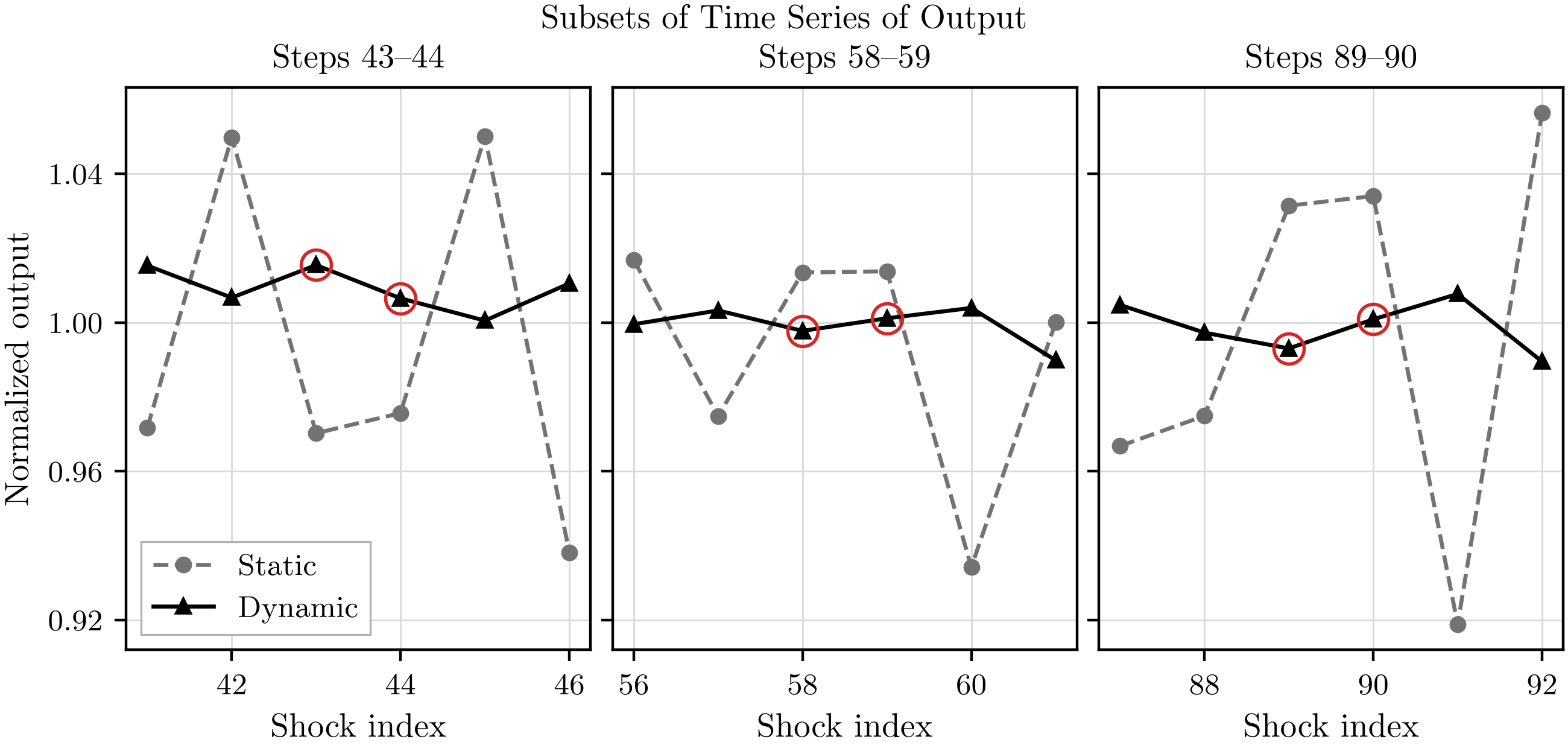}
\caption{Three windows of the series in Figure~\ref{fig:ts} (steps $43$--$44$, $58$--$59$, $89$--$90$)
in which dynamic output moves more than static output.  The reversal steps are circled.  We write
$d_{\mathrm{dyn}}$ and $d_{\mathrm{stat}}$ for the corresponding one-step changes.}
\label{fig:ts_zoom}
\end{figure}

Figure~\ref{fig:hist_ratio} turns the visual comparison into a comparison of distributions. It plots the empirical distribution of one-step output ratios $Y_s/Y_{s+1}$ for both series over $10{,}000$ draws on the same network. The dynamic distribution is a tall, narrow spike around one, while the static distribution is broad, the standard deviation of the ratio being $0.061$ for the static series against $0.0091$ for the dynamic series (a ratio of $0.149$), and the $1$--$99\%$ spread is $0.283$ against $0.042$ (a ratio of $0.148$). The two distributions have the same shape, and would at any network size, because the innovations are Gaussian and the log-output recursion is linear, so log output is Gaussian in both series, and the kurtosis of about $3.0$ in both histograms records nothing more than the exponential map from log output to the ratio. The dynamics therefore compress the width of the entire distribution by the constant factor $0.149$ and leave its shape unchanged. At a fixed threshold the compression is the tail effect of Proposition~\ref{prop:tail_amp}. A one-step change of five percent or more, which occurs in roughly two of every five static draws, does not occur once in the ten thousand dynamic draws.

\begin{figure}[H]
\centering
\includegraphics[width=\textwidth]{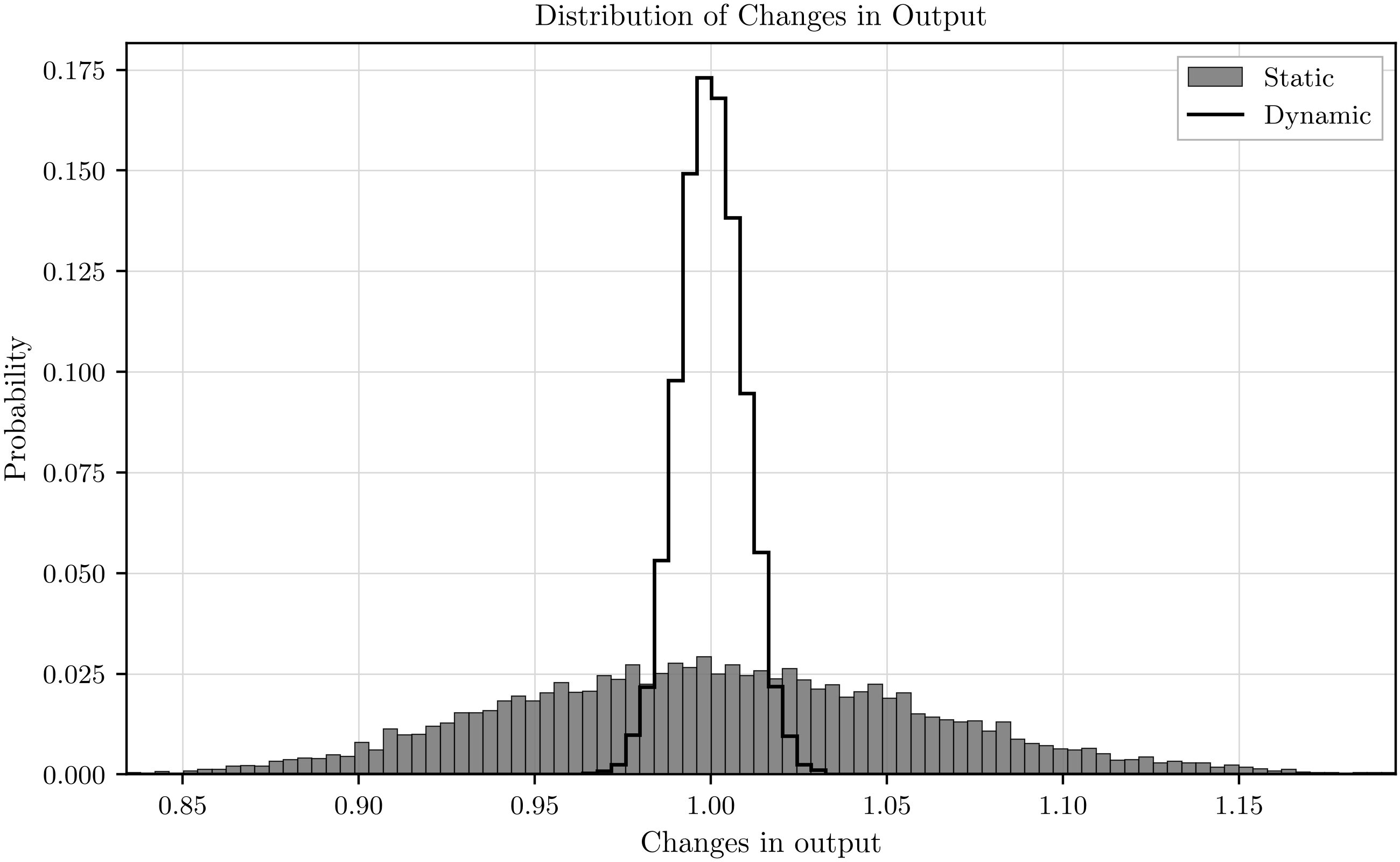}
\caption{Empirical distribution of one-step output ratios $Y_s/Y_{s+1}$ for the static and dynamic series,
over $10{,}000$ draws on a reconstructed US network of about $10^5$ firms.  The dynamic distribution
(narrow) is far more concentrated than the static (wide).  Both are close to Gaussian.}
\label{fig:hist_ratio}
\end{figure}

%----------------------------------------------------
\subsection{The scaling of aggregate volatility}
%----------------------------------------------------

One might worry that the gap between static and dynamic volatility is a small-economy phenomenon, and that it washes out as the number of firms grows and cross-sectional averaging takes hold. We therefore repeat the comparison across economy sizes from about $10^4$ to $6\times 10^6$ firms. For each size we draw reconstructed networks and shock collections, compute the static and dynamic output sequences, and record the two aggregate volatilities, averaging over $100$ reconstructions at the two smaller sizes and using a single network at the two largest.

Figure~\ref{fig:scaling} plots the mean volatilities on log--log axes. Both decline with the number of firms, and more slowly than the $1/\sqrt n$ of the central limit theorem, as a fat-tailed size distribution implies. In the static series volatility falls from $11.9\%$ at about $10^4$ firms to $1.04\%$ at about $6\times 10^6$ firms, a factor of eleven against a factor of six hundred in the number of firms, which is a slope of about $-0.4$ on the log--log axes where the central limit theorem would give $-0.5$. The two lines are parallel, because under size weights both volatilities are proportional to the size norm $\|\mathbf m^\ast\|_2$ of the granular driver in Proposition~\ref{prop:singular_aggregator}, so the dynamic-to-static ratio is about $0.15$ at every size ($0.148$, $0.150$, $0.146$, $0.158$ across the four sizes). The dampening is therefore not a small-economy phenomenon. It operates with the same force at every scale, which is the computational counterpart of Lemma~\ref{lem:neumann_bounds_L} and Corollary~\ref{coro:granular_perron_rate}, both of which order the two volatilities independently of $n$.

\begin{figure}[H]
\centering
\includegraphics[width=\textwidth]{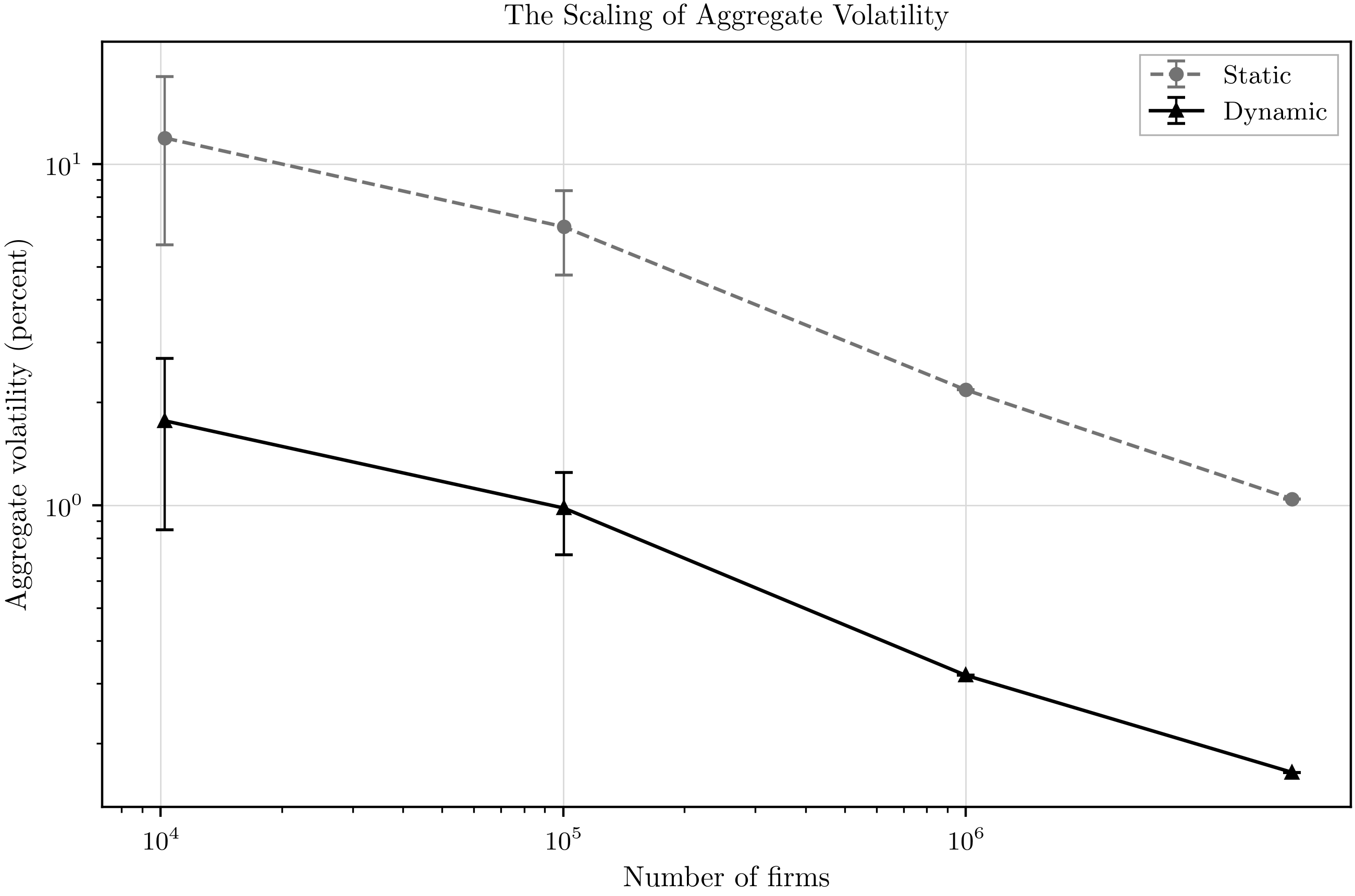}
\caption{Static and dynamic aggregate volatility (percent, log scale) against the number of firms (log
scale).  Each point at $10^4$ and $10^5$ firms is the mean over $100$ reconstructed US networks (error
bars $\pm 1$ s.d.\ across networks).  The $10^6$ and $6\times 10^6$ points are single networks.}
\label{fig:scaling}
\end{figure}

The value of the ratio, and not only the ordering, is the main result of this section. Proposition~\ref{prop:singular_aggregator} says that under size weighting the ratio equals $(1-\varsigma)/\sqrt{1+\varsigma}$ for every network, with nothing left to fit. Volatility is a standard deviation, so the measured $0.15$ corresponds to a variance ratio of about $0.022$, against the $\mathcal R(\varsigma)=(1-\varsigma)^2/(1+\varsigma)=0.0222$ of the proposition at $\varsigma=0.8$. Two of the four measured ratios sit within one percent of $0.149$, a third within two percent, and the fourth, the single-network $6\times10^6$ run, within six percent, across networks that differ in size by nearly three decades and in topology from one draw to the next. The two ratios within one percent are means over $100$ networks. The two that sit farther from $0.149$ are the single-network runs, each estimated from one window of draws on one network, and their distance from $0.149$ is the finite-window dispersion of Section~\ref{sec:eigenmodes}, whose settled start and window length these runs share.

\smallskip
The ratio measured here therefore cannot depend on the network, even though every network in the experiments is fat-tailed and freshly reconstructed. The measurements are blind to the network because of how the aggregate is built, since under size weights every transient mode is removed before it reaches the aggregate, for any transient eigenvalue. The network matters for weights that load on a transient mode, and size weights load on none. An observer who weights firms by equilibrium size will conclude that convergence speed is governed by returns to scale alone, whatever the mixing rate of the network being observed. Exhibiting the reallocation channel requires weights satisfying Assumption~\ref{assu:misalignment}, which none of our experiments imposes.

%----------------------------------------------------
\subsection{Aggregate volatility across reconstructed US production networks}
%----------------------------------------------------

Finally, we ask how the dampening behaves across independent reconstructions of the US production network. We generate an ensemble of reconstructed firm-to-firm networks at two sizes, about $10^4$ and $10^5$ firms, with $100$ independent reconstructions at each size, and, holding the rest of the calibration fixed, we run the static and dynamic experiments on every reconstructed network.

Figure~\ref{fig:boxplot} reports the results in two size panels. In each panel the dynamic box sits far below the static one. Median static volatility is $9.65\%$ against dynamic $1.45\%$ at $10^4$ firms, and $6.12\%$ against $0.92\%$ at $10^5$ firms, a ratio of about $0.15$ at both sizes. Both distributions tighten as the economy grows. The spread within each panel is the spread of the size norm $\|\mathbf m^\ast\|_2$ across reconstructions, to which both volatilities are proportional under size weights, and that spread leaves the ratio untouched, the ratio being about $0.15$ in every reconstruction, as Figure~\ref{fig:agg_vol_robust} in Appendix~\ref{app:abm_eigenvalue} shows at $10^4$ firms. The gap is thus a feature of the overlapping-adjustment mechanism that every reconstruction in the ensemble reproduces.

\begin{figure}[H]
\centering
\includegraphics[width=\textwidth]{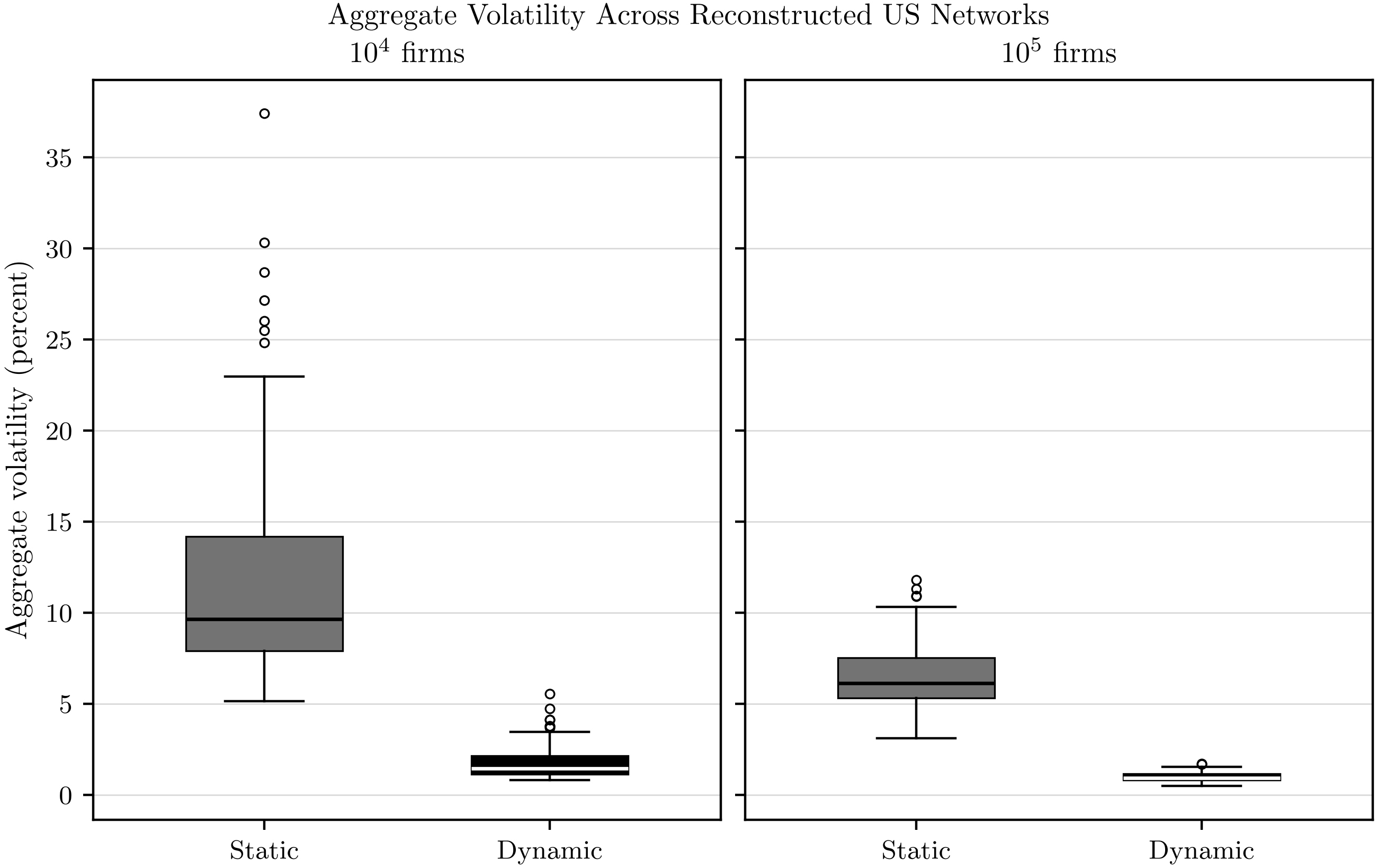}
\caption{Static and dynamic aggregate volatility across reconstructed US production networks, in two
size panels ($\approx 10^4$ and $\approx 10^5$ firms, $100$ independent reconstructions each), with
$100$ productivity draws per network.  In both panels the dynamic distribution sits well below the
static one.}
\label{fig:boxplot}
\end{figure}

Taken together, the experiments establish that static volatility exceeds dynamic volatility on real topology, as Lemma~\ref{lem:neumann_bounds_L} orders them, across nearly three decades in the number of firms and across independent reconstructions, with dynamic aggregate volatility running at about $15$ percent of its static counterpart. They also establish that the size-weighted ratio of Proposition~\ref{prop:singular_aggregator}, an identity of the linearized economy, survives the full economy with its money circuit and its prices and every mode retained, and the measured ratios scatter around it by no more than finite windows produce. Under equilibrium gross-sales weights the surviving fraction is set by returns to scale alone, and the aggregate does not see the network's transient spectrum. The empirical content of the granular-origins view therefore depends jointly on the rate at which the economy converges to equilibrium and on the weights through which convergence is observed.

\section{Concluding Thoughts} \label{sec:conclusion}
How granular innovations aggregate into macroeconomic fluctuations is a difficult problem because aggregation runs along two dimensions at once. Shocks cumulate across firms, and they cumulate across time, and neither process operates alone. A firm's innovation reaches other firms only through changes in the flow of intermediate inputs, and those flows take time to travel the supply chain. Successive draws interact only through the network, because how today's draw meets yesterday's depends on which firms were hit and how they are connected. Treating the economy as a representative firm, or as one that reaches equilibrium without the passage of time, collapses the interaction between the two processes. What we show in this paper is that much of what is macroeconomically consequential arises precisely from that interaction, as shocks mix across firms while overlapping across dates. We model this mechanism using a Neumann-series representation of the propagation operator to derive asymptotic aggregate volatility and tail risk. We then derive the full finite-horizon distribution of aggregate volatility using a dominant-transient (second-mode) reduction of the production network. We find that macroeconomic risk increases with the rate at which the economy converges to equilibrium. Slower convergence means more overlapping of shocks from the past, and therefore more `time-averaging'. Because cross-sectional averaging of shocks across firms itself occurs only through time, stronger time-averaging dampens the macroeconomic impact of shocks originating from highly connected firms. The dampening holds provided the aggregation weights do not annihilate the transient modes, as equilibrium gross-sales weights do, at which the network's transient spectrum drops out of the attenuation ratio altogether.\footnote{One might worry that granular shocks need not be i.i.d.\ over time, since many observed
series are autocorrelated.  But serial dependence in observables need not imply serial dependence
in the primitive innovations, since network propagation can generate topology-dependent temporal
correlation even when primitives are uncorrelated, and can even flip its sign
\citep{puravankara_ts_networks2026}.  More generally, autocorrelation in measured productivity is
therefore not, by itself, a reason to build persistence into the primitives. See also \citet{foerster2008} and \cite{liu2020}.}

\smallskip
Why highly connected firms matter less in a dynamic setting is easy to see. Even the largest firms in an economy are connected to only a small portion of the universe of firms.\footnote{In the United States, for instance, there are about six million firms \citep{Axtell2001Zipf}. But even the largest of these tend not to have more than ten thousand significant connections \citep{Atalay2011NetworkStructure}.} Their macroeconomic influence arises through indirect buyer--seller relations, in that large firms matter because their innovations can reach most other firms through second- or third-order connections. In a dynamic setting, these connections perennially carry disturbances from earlier innovations, originating from firms large and small. Every new innovation therefore superposes with the partially processed waves of earlier vintages, and the resulting interference and partial cancellation mute the aggregate consequences of tail innovations in highly connected firms. The slower the economy converges to equilibrium, the more residual disturbance the paths of propagation carry. When convergence is sufficiently slow, and the weights load on the transient direction at all, the macroeconomic risks generated by a fat-tailed degree distribution can therefore look much the same as those in an economy with roughly homogeneous connections. The reason is that time-averaging dampens the influence of the productivity change in one firm on the output of others.

\smallskip
The paper's results can be summarized briefly. Lemma \ref{lem:neumann_bounds_L} compares static and dynamic aggregate volatility under an infinite stream of i.i.d.\ shocks, for every fixed nonnegative aggregation vector, household-derived weights included. Dynamic aggregate volatility is lower than its static counterpart because the productivity waves generated by i.i.d.\ shocks of different vintages superpose, thereby dampening each other's impact on output. The lemma uses a Neumann-series representation of the Leontief inverse to derive this asymptotic comparison. Theorem \ref{thm:attenuation} makes the comparison exact for the dominant transient mode, where overlap deflates the static growth-rate variance of that mode by a closed-form factor that depends only on the dominant transient eigenvalue. Building on the same spectral reduction, we characterize the finite-horizon distributions of static and dynamic aggregate volatility as weighted sums of independent chi-square variables. Theorem \ref{thm:phi_finite_ordering} shows that the static distribution dominates the dynamic one in the first-order stochastic sense, at every window length and every admissible rate of convergence. Theorem \ref{thm:full_network_fosd} proves the same dominance for aggregate output on any production network of the model, under any admissible weights and without any spectral hypothesis. Theorem \ref{thm:aggregator_geometry} and Proposition \ref{prop:singular_aggregator} then supply what the preceding results leave open, namely how much of any of this an observer sees. The theorem decomposes the aggregate ratio across eigenmodes, with every admissible weight vector carrying the granular channel at full loading, and the proposition shows that weighting firms by equilibrium size removes the transient spectrum exactly, without any diagonalizability hypothesis. The two channel rates of our calibration are therefore statements about the two channels rather than about the aggregate, which is known exactly only at size weights. Together these results give the channel rates used in that calibration and the point prediction that the computational experiments check. The intermediate derivations also allow us to make stochastic-dominance statements about static and dynamic tail-risk.  Overall, the passing of time can blunt the macro consequences of degree heterogeneity. Put differently, how much fat-tails matter depends on the rate at which the economy converges to equilibrium, and on the weights through which convergence is observed.

\smallskip
We estimated the empirical significance of our theoretical insights by running computational experiments on a reconstructed production network with the universe of firms in the United States. These experiments showed that the aggregate volatility generated by granular shocks when the economy takes time to relax to equilibrium is an order of magnitude lower than that under the assumption of instantaneous convergence. Which is to say that the claim that `granular shocks can explain nearly all of the observed macroeconomic fluctuations within a general equilibrium setting' is not the same as the claim that `granular shocks can explain nearly all of the observed macroeconomic fluctuations'. The latter claim cannot be made unconditional on the rate of convergence to equilibrium, and conditional on the rate it can plausibly leave a sizeable share of aggregate volatility unexplained by granular shocks. How much of observed macroeconomic fluctuation granular shocks account for is thus a statement about a rate of convergence to equilibrium as much as about the size or the degree distribution of firms.

\newpage
\bibliography{ref}
\bibliographystyle{aea}

\newpage

\appendix

\section{Proofs}
\label{app:proofs}

\subsection{\texorpdfstring{Proofs for Section~\ref{sec:volatility_leontief}}{Proofs for Section 3}}

\subsubsection{\texorpdfstring{Proof of Lemma~\ref{lem:neumann_convergence_L}}{Proof of Lemma 1}}
\label{app:proof:lem:neumann_convergence_L}
\begin{proof}
The Neumann expansion of $(\mathbf I-\mathbf W)^{-1}$ is the standard consequence of
$\rho(\mathbf W)<1$.  The column-stochastic structure makes the layer weights exact.  Since
$\mathbf W\mathbf 1=\varsigma\mathbf 1$,
and both $\boldsymbol{\gamma}$ and $\mathbf W$ are entrywise nonnegative with
$\boldsymbol{\gamma}^\top\mathbf 1=1$, the row vector $\boldsymbol{\gamma}^\top\mathbf W^\ell$ is
nonnegative and
\[
\bigl\|\boldsymbol{\gamma}^\top \mathbf W^\ell\bigr\|_1
=
\boldsymbol{\gamma}^\top \mathbf W^\ell\mathbf 1
=
\varsigma^{\,\ell}.
\]
The same computation applied to each row of the nonnegative matrix $\mathbf W^\ell$, whose row sums are
$\varsigma^\ell$ because $\mathbf W^\ell\mathbf 1=\varsigma^\ell\mathbf 1$, gives
$\|\mathbf W^\ell\|_\infty=\varsigma^\ell$.  H\"older's inequality then gives
$\bigl|\boldsymbol{\gamma}^\top \mathbf W^\ell\boldsymbol{\epsilon}\bigr|\le\varsigma^{\,\ell}\|\boldsymbol{\epsilon}\|_\infty$,
and summing the geometric tail yields \eqref{eq:finite_depth_error_bound_simple}.
\end{proof}

\subsubsection{\texorpdfstring{Proof of Lemma~\ref{lem:neumann_bounds_L}}{Proof of Lemma 2}}
\label{app:proof:lem:neumann_bounds_L}
\begin{proof}
The increment loads on consecutive depth blocks, whose difference the Neumann sum controls because every block is entrywise nonnegative. We establish in turn the moving-average representation of the increment, the variance bound $\phi^{[\mathcal L]}\le\phi^\ast$, the tail-risk ordering, and the deep-processing limit.

We first derive the moving-average representation, starting from the stationary superposition
\eqref{eq:superposition_infinite_vintage_L}.  Write
$\mathbf a_\tau^\top:=\boldsymbol{\gamma}^\top\mathbf W^{\tau\mathcal L}\mathbf S_{\mathcal L}$ for the
block exposure of a vintage of age $\tau$, so that
$\mathbf a_0^\top=\boldsymbol{\gamma}^\top\mathbf S_{\mathcal L}$.  Shifting the index of the lagged sum
by one date, the vintage $\boldsymbol{\epsilon}_{t-\tau}$ enters $y_{t-1}^{[\mathcal L]}$ with the
coefficient $\mathbf a_{\tau-1}$ for $\tau\ge 1$ and does not enter it at all for $\tau=0$.  Therefore,
\begin{align*}
\Delta y_t^{[\mathcal L]}
&=
\sum_{\tau=0}^{\infty}\mathbf a_\tau^\top\boldsymbol{\epsilon}_{t-\tau}
-\sum_{\tau=0}^{\infty}\mathbf a_\tau^\top\boldsymbol{\epsilon}_{t-1-\tau} \\
&=
\mathbf a_0^\top\boldsymbol{\epsilon}_{t}
+
\sum_{\tau=1}^{\infty}\bigl(\mathbf a_\tau-\mathbf a_{\tau-1}\bigr)^\top
\boldsymbol{\epsilon}_{t-\tau}.
\end{align*}

We next bound the increment variance.  By
independence across dates and $\Cov(\boldsymbol{\epsilon}_t)=\sigma^2\mathbf I$,
\[
\phi^{[\mathcal L]}
=
\sigma^2\Bigg(\|\mathbf a_0\|_2^2+\sum_{\tau=1}^{\infty}\|\mathbf a_\tau-\mathbf a_{\tau-1}\|_2^2\Bigg).
\]
Since $\mathbf W\ge 0$ and $\boldsymbol{\gamma}\ge 0$, every $\mathbf a_\tau$ is entrywise nonnegative, so
all pairwise inner products $\mathbf a_\tau^\top\mathbf a_{\tau'}$ are nonnegative.  Applied to adjacent
blocks this gives $\|\mathbf a_\tau-\mathbf a_{\tau-1}\|_2^2\le\|\mathbf a_\tau\|_2^2+\|\mathbf a_{\tau-1}\|_2^2$,
and hence, writing $S:=\sum_{\tau\ge0}\|\mathbf a_\tau\|_2^2$,
\[
\sum_{\tau=1}^{\infty}\|\mathbf a_\tau-\mathbf a_{\tau-1}\|_2^2
\;\le\;
\sum_{\tau\ge1}\|\mathbf a_\tau\|_2^2+\sum_{\tau\ge0}\|\mathbf a_\tau\|_2^2
=
2S-\|\mathbf a_0\|_2^2,
\]
so that $\phi^{[\mathcal L]}\le 2\sigma^2 S$.  Applied to all pairs, nonnegativity also gives
\[
S
\;\le\;
\Big\|\sum_{\tau=0}^{\infty}\mathbf a_\tau\Big\|_2^2 .
\]
The block exposures sum to the fully processed exposure by the tiling identity \eqref{eq:block_tiling},
\[
\sum_{\tau=0}^{\infty}\mathbf a_\tau^\top
=
\boldsymbol{\gamma}^\top\Bigg(\sum_{\tau=0}^{\infty}\mathbf W^{\tau\mathcal L}\mathbf S_{\mathcal L}\Bigg)
=
\boldsymbol{\gamma}^\top(\mathbf I-\mathbf W)^{-1}.
\]
Combining the three displays yields
\[
\phi^{[\mathcal L]}
\;\le\;
2\sigma^2\bigl\|\boldsymbol{\gamma}^\top(\mathbf I-\mathbf W)^{-1}\bigr\|_2^2.
\]
In the static benchmark, $y_t^\ast=\boldsymbol{\gamma}^\top(\mathbf I-\mathbf W)^{-1}\boldsymbol{\epsilon}_t$
is i.i.d.\ across $t$, hence
\[
\phi^\ast
=
\Var(\Delta y_t^\ast)
=
2\,\Var(y_t^\ast)
=
2\sigma^2\bigl\|\boldsymbol{\gamma}^\top(\mathbf I-\mathbf W)^{-1}\bigr\|_2^2,
\]
so $\phi^{[\mathcal L]}\le \phi^\ast$.

We now turn to the tail-risk ordering.  Both $\Delta y_t^{[\mathcal L]}$ and $\Delta y_t^\ast$ are
centered Gaussian scalars under Assumption~\ref{assu:innovations}.  For fixed $c>0$, if
$Z\sim\mathcal N(0,v)$ then $v\mapsto \Pr(Z<-c)$ is
increasing, so $\phi^{[\mathcal L]}\le\phi^\ast$ implies $\omega_c^{[\mathcal L]}\le \omega_c^\ast$.

We finally take the deep-processing limit.  The layer-weight identity
\eqref{eq:layer_weight_identity} of Lemma~\ref{lem:neumann_convergence_L} gives
$\|\mathbf a_\tau\|_1=\boldsymbol{\gamma}^\top\mathbf W^{\tau\mathcal L}\mathbf S_{\mathcal L}\mathbf 1
=\varsigma^{\tau\mathcal L}(1-\varsigma^{\mathcal L})/(1-\varsigma)$, so
$\mathbf a_\tau\to\mathbf 0$ for every $\tau\ge1$ as $\mathcal L\to\infty$, while
$\mathbf a_0^\top=\boldsymbol{\gamma}^\top\mathbf S_{\mathcal L}\to\boldsymbol{\gamma}^\top(\mathbf I-\mathbf W)^{-1}$
by the Neumann expansion \eqref{eq:neumann_series}.  The tail of the moving-average representation is controlled uniformly in $\mathcal L$, since by nonnegativity and $\|\cdot\|_2\le\|\cdot\|_1$,
\[
\sum_{\tau=2}^{\infty}\|\mathbf a_\tau-\mathbf a_{\tau-1}\|_2^2
\;\le\;
2\sum_{\tau=1}^{\infty}\|\mathbf a_\tau\|_1^2
=
\frac{2(1-\varsigma^{\mathcal L})^2}{(1-\varsigma)^2}\cdot\frac{\varsigma^{2\mathcal L}}{1-\varsigma^{2\mathcal L}}
\;\longrightarrow\;0 .
\]
The representation therefore converges in $L^2$ to
$\boldsymbol{\gamma}^\top(\mathbf I-\mathbf W)^{-1}(\boldsymbol{\epsilon}_t-\boldsymbol{\epsilon}_{t-1})
=\Delta y_t^\ast$, and $\phi^{[\mathcal L]}\to\phi^\ast$.
\end{proof}

\subsubsection{\texorpdfstring{Proof of Lemma~\ref{lem:modal_decomposition}}{Proof of Lemma 3}}
\label{app:proof:lem:modal_decomposition}
\begin{proof}
Define $s_{r,t}:=\sum_{\tau=0}^{t-1}\lambda_r^{\,\tau}\xi_{r,t-\tau}$, the mode-$r$ inner sum of
\eqref{eq:y_modal_exact} with every innovation dated one period later.  Splitting off the
$\tau=0$ term gives $s_{r,t}=\lambda_r s_{r,t-1}+\xi_{r,t}$ with $s_{r,0}=0$, and
\eqref{eq:y_modal_exact} with the same relabeling of innovation dates reads
$y_t=\sum_r c_r s_{r,t}$.  Because the innovations are i.i.d.\ across dates, the relabeling changes
no variance, tail, or reversal object, and it lines up the dynamic series with the contemporaneous
static benchmark $y_t^\ast=\boldsymbol\gamma^\top(\mathbf I-\mathbf W)^{-1}\boldsymbol{\epsilon}_t$.
Since $\mathbf v_1=\mathbf 1$ by \eqref{eq:W_spectral} and
$\boldsymbol\gamma^\top\mathbf 1=1$, the granular loading is $c_1=1$.  For a complex pair
$(\lambda_r,\overline{\lambda_r})$ the eigenvectors can be chosen conjugate, so the two modal terms
are conjugates and their sum is real.
\end{proof}

\subsubsection{\texorpdfstring{Proof of Theorem~\ref{thm:attenuation}}{Proof of Theorem 1}}
\label{app:proof:thm:attenuation}
\begin{proof}
In the static benchmark, \eqref{eq:Dy_static_lambda2} gives
$\Delta y_t^{\ast,\mathrm{re}}=\frac{b}{1-\lambda_2}(\eta_t-\eta_{t-1})$, so independence implies
\[
\Var(\Delta y_t^{\ast,\mathrm{re}})
=
\frac{b^2}{(1-\lambda_2)^2}\Var(\eta_t-\eta_{t-1})
=
\frac{2b^2\sigma^2}{(1-\lambda_2)^2},
\]
which is \eqref{eq:var_static_increment}.

In the dynamic economy, we have
$y_t^{\mathrm{re}}=b\,s_t$ with $s_t=\lambda_2 s_{t-1}+\eta_t$, so
\begin{equation}
\Delta y_t^{\mathrm{re}}
=
b(s_t-s_{t-1})
=
b\bigl((\lambda_2-1)s_{t-1}+\eta_t\bigr).
\label{eq:Dy_dynamic_lambda2_forproof}
\end{equation}
Since the stationary state $s_{t-1}=\sum_{\tau\ge 0}\lambda_2^{\,\tau}\eta_{t-1-\tau}$ is measurable with respect to past
innovations, it is independent of $\eta_t$.  Moreover,
$\Var(s_{t-1})=\sigma^2\sum_{\tau\ge 0}\lambda_2^{2\tau}=\sigma^2/(1-\lambda_2^2)$, so
\begin{align*}
\Var(\Delta y_t^{\mathrm{re}})
&=
b^2\Bigl((1-\lambda_2)^2\Var(s_{t-1})+\Var(\eta_t)\Bigr)
\\
&=
\frac{2b^2\sigma^2}{1+\lambda_2},
\end{align*}
which is \eqref{eq:var_dynamic_increment}.  Taking the ratio of \eqref{eq:var_dynamic_increment} and \eqref{eq:var_static_increment} yields \eqref{eq:attenuation_ratio}.  Finally,
$\mathcal R'(\lambda)=-(1-\lambda)(3+\lambda)/(1+\lambda)^2<0$ for $\lambda\in(0,1)$, and
$\mathcal R(\lambda)\to 0$ as $\lambda\uparrow 1$.
\end{proof}

\subsubsection{\texorpdfstring{Proof of Corollary~\ref{coro:granular_perron_rate}}{Proof of Corollary 1}}
\label{app:proof:coro:granular_perron_rate}
\begin{proof}
The granular channel is the scalar system of Theorem~\ref{thm:attenuation} under the
substitution $(b,\lambda_2,\eta_t)\mapsto(1,\,\lambda_1,\,g_t)$ with $g_t=\mathbf m^{\ast\top}\boldsymbol{\epsilon}_t$.
Here $\sigma_g^2$ cancels in the ratio.
\end{proof}

\subsubsection{\texorpdfstring{Proof of Theorem~\ref{thm:aggregator_geometry}}{Proof of Theorem 2}}
\label{app:proof:thm:aggregator_geometry}
\begin{proof}
By Lemma~\ref{lem:modal_decomposition},
$\Delta y_t=\sum_r c_r\bigl[(\lambda_r-1)s_{r,t-1}+\xi_{r,t}\bigr]$ under the stationary law.
Substituting
$s_{r,t-1}=\sum_{\tau\ge0}\lambda_r^{\,\tau}\xi_{r,t-1-\tau}$ and collecting the coefficient of
each date's innovation gives
\[
\Delta y_t
=
\Big(\sum_r\mathbf M_r\Big)^{\!\top}\boldsymbol{\epsilon}_t
-\sum_{\tau\ge0}\Big(\sum_r(1-\lambda_r)\lambda_r^{\,\tau}\mathbf M_r\Big)^{\!\top}\boldsymbol{\epsilon}_{t-1-\tau}.
\]
Independence across dates and $\Cov(\boldsymbol{\epsilon}_t)=\sigma^2\mathbf I$ give the first display, and $y_t^\ast=\boldsymbol\gamma^\top(\mathbf I-\mathbf W)^{-1}\boldsymbol{\epsilon}_t
=\bigl(\sum_r\mathbf M_r/(1-\lambda_r)\bigr)^\top\boldsymbol{\epsilon}_t$ with $\Var(\Delta y_t^\ast)=2\Var(y_t^\ast)$ gives the second.
\end{proof}

\subsubsection{\texorpdfstring{Proof of Proposition~\ref{prop:singular_aggregator}}{Proof of Proposition 1}}
\label{app:proof:prop:singular_aggregator}
\begin{proof}
From $\mathbf A\mathbf m^\ast=\mathbf m^\ast$ and $\mathbf W=\varsigma\mathbf A^\top$,
$\mathbf m^{\ast\top}\mathbf W=\varsigma(\mathbf A\mathbf m^\ast)^\top=\varsigma\,\mathbf m^{\ast\top}$,
and by induction $\mathbf m^{\ast\top}\mathbf W^{\tau}=\varsigma^{\tau}\mathbf m^{\ast\top}$ for every
$\tau\ge0$.  Substituting $\boldsymbol\gamma=\mathbf m^\ast$ into
Definition~\ref{def:aggregate_output} gives the autoregression in the statement, and
$\mathbf m^{\ast\top}(\mathbf I-\mathbf W)^{-1}=\sum_{\tau\ge0}\varsigma^{\tau}\mathbf m^{\ast\top}
=\mathbf m^{\ast\top}/(1-\varsigma)$ gives $y_t^\ast=g_t/(1-\varsigma)$.  The stream $\{g_t\}$ is
i.i.d.\ with finite variance $\sigma^2\|\mathbf m^\ast\|_2^2$, so the variance computation of
Theorem~\ref{thm:attenuation} applied to the scalar system
$(b,\lambda_2,\eta_t)\mapsto(1,\varsigma,g_t)$, which uses only i.i.d.\ innovations with finite variance, gives the ratio.
\end{proof}

\subsubsection{\texorpdfstring{Proof of Proposition~\ref{prop:tail_amp}}{Proof of Proposition 2}}
\label{app:proof:prop:tail_amp}
\begin{proof}
Both increments are centered Gaussian scalars, so $\omega_c=\Phi(-c/\sqrt{\phi})=\Phi(-x)$ and,
since $c/\sqrt{\phi^\ast}=x\sqrt{\phi/\phi^\ast}=x\sqrt{\kappa}$,
$\omega_c^\ast=\Phi(-x\sqrt{\kappa})$.  Taking the ratio gives \eqref{eq:tail_ratio_exact}.  For
\eqref{eq:tail_ratio_asymp}, apply the Mills-ratio asymptotic
$\Phi(-u)=(1+o(1))\,e^{-u^2/2}/(u\sqrt{2\pi})$ as $u\to\infty$ at $u=x\sqrt{\kappa}$ and at $u=x$:
\[
\frac{\omega_c^\ast}{\omega_c}
=
\frac{e^{-\kappa x^2/2}\,/\,(x\sqrt{\kappa}\,\sqrt{2\pi})}
     {e^{-x^2/2}\,/\,(x\sqrt{2\pi})}\,(1+o(1))
=
\frac{1}{\sqrt{\kappa}}\,
\exp\!\left\{\frac{1}{2}x^2(1-\kappa)\right\}(1+o(1)).
\]
\end{proof}

\subsection{\texorpdfstring{Proofs for Section~\ref{sec:eigenmodes}}{Proofs for Section 4}}

\subsubsection{\texorpdfstring{Lemma~\ref{lem:phi_quadratic_forms} and its proof}{Lemma 4 and its proof}}
\label{app:proof:lem:phi_quadratic_forms}

The weighted chi-square representations \eqref{eq:phi_stat_chisq} and \eqref{eq:phi_dyn_chisq} of the text are the content of the following lemma.

\begin{lemma}[Weighted chi-square forms of finite-$T$ realized volatility]
\label{lem:phi_quadratic_forms}
Fix $T\ge2$ and maintain Assumption~\ref{assu:innovations}. Then
\begin{align*}
\hat\phi^\ast
&\ \overset{d}{=}\
\sigma^2 \frac{b^2}{(1-\lambda_2)^2}\cdot\frac{1}{T-1}\sum_{j=1}^{T-1}\nu^*_j\,Z_j^2,
\\
\hat\phi
&\ \overset{d}{=}\
\sigma^2 b^2\cdot\frac{1}{T-1}\sum_{j=1}^{T-1}\nu_j\,Z_j^2,
\end{align*}
where $\{Z_j\}_{j=1}^{T-1}$ are i.i.d.\ standard normals and $\{\nu^*_j\}_{j=1}^{T-1}$ and $\{\nu_j\}_{j=1}^{T-1}$ are the eigenvalues of $\mathbf M^\ast$ and $\mathbf M$.
\end{lemma}

\begin{proof}
By \eqref{eq:phi_stat_quad_main}, \eqref{eq:phi_dyn_quad_main}, and the normalization \eqref{eq:u2_normalization}, under which $\boldsymbol{\eta}\overset{d}{=}\sigma\mathbf Z$ with $\mathbf Z\sim\mathcal N(0,I_T)$,
\[
(T-1)\,\hat\phi^\ast
\ \overset{d}{=}\
\frac{\sigma^2b^2}{(1-\lambda_2)^2}\,\|\mathbf D_T\mathbf Z\|_2^2,
\qquad
(T-1)\,\hat\phi
\ \overset{d}{=}\
\sigma^2b^2\,\|\mathbf D_T\mathbf K\mathbf Z\|_2^2 .
\]
The vectors $\mathbf D_T\mathbf Z$ and $\mathbf D_T\mathbf K\mathbf Z$ are centered Gaussian on $\mathbb R^{T-1}$ with covariance matrices $\mathbf M^\ast=\mathbf D_T\mathbf D_T^\top$ and $\mathbf M=\mathbf D_T\mathbf K\mathbf K^\top\mathbf D_T^\top$, both symmetric and positive semidefinite. Diagonalizing each covariance matrix by an orthogonal matrix turns the squared norm into a sum of the eigenvalues times the squares of independent standard normals,
\[
\|\mathbf D_T\mathbf Z\|_2^2 \ \overset{d}{=} \ \sum_{j=1}^{T-1}\nu^*_j Z_j^2,
\qquad
\|\mathbf D_T\mathbf K\mathbf Z\|_2^2 \ \overset{d}{=} \ \sum_{j=1}^{T-1}\nu_j Z_j^2,
\]
which gives the two representations.
\end{proof}

\subsubsection{\texorpdfstring{Proof of Theorem~\ref{thm:phi_finite_ordering}}{Proof of Theorem 3}}
\label{app:proof:thm:phi_finite_ordering}
\begin{proof}
The proof rests on a norm bound on the overlap operator, which is a truncated Neumann series in the nilpotent shift and therefore has operator norm below the static resolvent gain $(1-\lambda_2)^{-1}$. We show that (i) $\|\mathbf K\|_2\le(1-\lambda_2^{\,T})/(1-\lambda_2)$, (ii) $\mathbf M\preceq\|\mathbf K\|_2^2\,\mathbf M^\ast$ in the Loewner order, (iii) the ordered eigenvalues inherit the domination, which is \eqref{eq:finiteT_fosd_condition}, and (iv) a comonotone coupling of the two weighted chi-square laws gives the stochastic dominance.

We first bound the overlap operator.  Let $\mathbf L$ denote the lower shift on $\mathbb R^{T}$,
$[\mathbf L]_{tq}=\mathbf 1\{t=q+1\}$, which satisfies $\mathbf L^{T}=\mathbf 0$.  Since
$[\mathbf K]_{tq}=\lambda_2^{\,t-q}\mathbf 1_{\{q\le t\}}$, the overlap operator is the truncated Neumann
series in the shift,
\[
\mathbf K=\sum_{m=0}^{T-1}\lambda_2^{\,m}\,\mathbf L^{m} .
\]
The shift is a partial isometry, so $\|\mathbf L\|_2=1$ and hence $\|\mathbf L^{m}\|_2\le 1$ for every
$m$.  The triangle inequality then gives
\[
\|\mathbf K\|_2
\ \le\
\sum_{m=0}^{T-1}\lambda_2^{\,m}
\ =\
\frac{1-\lambda_2^{\,T}}{1-\lambda_2}
\ <\
\frac{1}{1-\lambda_2} .
\]
This proves (i).

We next pass to the Loewner order.  For any $\mathbf x\in\mathbb R^{T-1}$,
\[
\mathbf x^\top\mathbf M\mathbf x
=
\bigl\|\mathbf K^\top\mathbf D_T^\top\mathbf x\bigr\|_2^2
\ \le\
\|\mathbf K\|_2^2\,\bigl\|\mathbf D_T^\top\mathbf x\bigr\|_2^2
=
\|\mathbf K\|_2^2\;\mathbf x^\top\mathbf M^\ast\mathbf x,
\]
so that $\mathbf M\preceq\|\mathbf K\|_2^2\,\mathbf M^\ast$.  This proves (ii).

We now compare the two weight profiles.  Weyl's monotonicity theorem states that if
$\mathbf B-\mathbf A$ is positive semidefinite then $\nu_{(j)}(\mathbf A)\le\nu_{(j)}(\mathbf B)$ for
every $j$.  Applying it to (ii), and substituting the bound on $\|\mathbf K\|_2$ from (i), gives
\eqref{eq:finiteT_fosd_condition}.  This proves (iii).

We finally couple the two statistics.  By Lemma~\ref{lem:phi_quadratic_forms} and the
normalization \eqref{eq:u2_normalization}, each statistic is distributed as a weighted sum
of $T-1$ independent $\chi^2(1)$ variables,
\begin{align*}
\hat\phi^\ast
&\ \overset{d}{=}\
\sigma^2 \frac{b^2}{(1-\lambda_2)^2}\cdot\frac{1}{T-1}\sum_{j=1}^{T-1}\nu^\ast_{(j)}\,Z_j^2
\\
\hat\phi
&\ \overset{d}{=}\
\sigma^2 b^2\cdot\frac{1}{T-1}\sum_{j=1}^{T-1}\nu_{(j)}\,Z_j^2.
\end{align*}
Since $\mathbf M^\ast$ and $\mathbf M$ do not commute in general, the two statistics are not driven by a common set of $\chi^2(1)$ variables in one orthonormal basis, and (ii) does not upgrade to pathwise domination. The Loewner comparison in (ii) is between the two Gram matrices on $\mathbb R^{T-1}$, whereas pathwise domination would require $\mathbf K^\top\mathbf D_T^\top\mathbf D_T\mathbf K\preceq(1-\lambda_2)^{-2}\,\mathbf D_T^\top\mathbf D_T$ on $\mathbb R^{T}$, which fails because $\mathbf D_T\mathbf 1=\mathbf 0$ gives $\hat\phi^\ast=0$ while $\hat\phi>0$ whenever $\boldsymbol\eta\propto\mathbf 1$. The distribution of a weighted sum of independent $\chi^2(1)$ variables depends only on the multiset of weights, however, and not on the basis that produced them. We therefore construct two random variables on a single standard-normal vector
$\mathbf Z=(Z_1,\ldots,Z_{T-1})^\top\sim\mathcal N(0,I_{T-1})$,
\[
Q^\ast:=\frac{\sigma^2 b^2}{(1-\lambda_2)^2(T-1)}\sum_{j=1}^{T-1}\nu^\ast_{(j)}Z_j^2,
\qquad
Q:=\frac{\sigma^2 b^2}{T-1}\sum_{j=1}^{T-1}\nu_{(j)}Z_j^2,
\]
by assigning the $j$-th largest weight of each statistic to the same coordinate $Z_j^2$.  This
common-coordinate coupling has $Q^\ast\overset{d}{=}\hat\phi^\ast$ and $Q\overset{d}{=}\hat\phi$, and
\[
Q^\ast-Q
=
\frac{\sigma^2 b^2}{T-1}\sum_{j=1}^{T-1}
\left(\frac{\nu^\ast_{(j)}}{(1-\lambda_2)^2}-\nu_{(j)}\right)Z_j^2 .
\]
By \eqref{eq:finiteT_fosd_condition} each coefficient is bounded below by
$\lambda_2^{\,T}(2-\lambda_2^{\,T})\,\nu^\ast_{(j)}/(1-\lambda_2)^2$, which is strictly positive because
$\mathbf D_T$ has full row rank and hence $\nu^\ast_{(j)}>0$ for every $j$.  Therefore $Q^\ast>Q$
almost surely, and $\Pr(\hat\phi^\ast>x)=\Pr(Q^\ast>x)\ge\Pr(Q>x)=\Pr(\hat\phi>x)$ for every $x$,
which is $\hat\phi^\ast\succeq_{\mathrm{FOSD}}\hat\phi$.

The coupling compares the two marginal laws and says nothing about the original pair $(\hat\phi^\ast,\hat\phi)$ on its common innovation history. The difference of that pair is negative on a set of positive probability, since for $\boldsymbol\eta$ close to a nonzero constant vector the static differences are close to zero while the dynamic differences from the settled start are not. This proves (iv).

It remains to treat the stationary initialization. Under the stationary initialization the state vector $\mathbf s=(s_1,\ldots,s_T)^\top$ is centered
Gaussian with covariance $\sigma^2\mathbf C$, where $[\mathbf C]_{tq}=\lambda_2^{|t-q|}/(1-\lambda_2^2)$.
Every increment $\Delta y_t=b(s_t-s_{t-1})$ therefore has the stationary variance $\phi$, and
$\E[\hat\phi]=\phi$.  For the dominance, write $\hat\phi=(b^2/(T-1))\,\mathbf s^\top\mathbf D_T^\top\mathbf D_T\mathbf s$
and $\mathbf s\overset{d}{=}\sigma\mathbf C^{1/2}\mathbf Z$ with $\mathbf Z\sim\mathcal N(0,I_T)$.  The
argument of Lemma~\ref{lem:phi_quadratic_forms}, with $\mathbf C^{1/2}\mathbf Z$ in place of
$\mathbf K\mathbf Z$, represents $\hat\phi$ as $\sigma^2b^2/(T-1)$ times a weighted sum of independent
$\chi^2(1)$ variables whose weights are the eigenvalues $\nu^{\mathrm{st}}_{(1)}\ge\cdots\ge\nu^{\mathrm{st}}_{(T-1)}$
of $\mathbf M^{\mathrm{st}}:=\mathbf D_T\mathbf C\mathbf D_T^\top$.  By steps (iii) and (iv) above, it suffices to show that
$\mathbf M^{\mathrm{st}}\preceq c_{\mathrm{st}}\,\mathbf M^\ast$ for some constant
$c_{\mathrm{st}}<(1-\lambda_2)^{-2}$.  For $\mathbf x\in\R^{T-1}$,
\[
\mathbf x^\top\mathbf M^{\mathrm{st}}\mathbf x
=
(\mathbf D_T^\top\mathbf x)^\top\mathbf C\,(\mathbf D_T^\top\mathbf x)
\ \le\
\|\mathbf C\|_2\,\mathbf x^\top\mathbf M^\ast\mathbf x,
\]
so $c_{\mathrm{st}}=\|\mathbf C\|_2$ serves once we show $\|\mathbf C\|_2<(1-\lambda_2)^{-2}$.

The entries of $\mathbf C$ are the Fourier coefficients of the symbol
$g(\theta):=(1-2\lambda_2\cos\theta+\lambda_2^2)^{-1}$, that is,
$[\mathbf C]_{tq}=(2\pi)^{-1}\int_{-\pi}^{\pi}g(\theta)e^{i(t-q)\theta}\,d\theta$, because
$\sum_{k\in\mathbb Z}\lambda_2^{|k|}e^{ik\theta}=(1-\lambda_2^2)\,g(\theta)$.  Hence, for every
$\mathbf w\in\R^T$,
\[
\mathbf w^\top\mathbf C\,\mathbf w
=
\frac{1}{2\pi}\int_{-\pi}^{\pi}g(\theta)\,\Big|\sum_{t=1}^{T}w_te^{it\theta}\Big|^2 d\theta .
\]
The symbol attains its maximum $(1-\lambda_2)^{-2}$ at $\theta=0$ only, since $\lambda_2\in(0,1)$,
and for $\mathbf w\neq\mathbf 0$ the trigonometric polynomial $\sum_t w_te^{it\theta}$ vanishes at
finitely many $\theta$.  The integrand is therefore strictly below
$(1-\lambda_2)^{-2}|\sum_t w_te^{it\theta}|^2$ on a set of positive measure, and Parseval's identity
gives
\[
\mathbf w^\top\mathbf C\,\mathbf w
\ <\
\frac{1}{(1-\lambda_2)^{2}}\cdot\frac{1}{2\pi}\int_{-\pi}^{\pi}\Big|\sum_{t=1}^{T}w_te^{it\theta}\Big|^2 d\theta
=
\frac{\|\mathbf w\|_2^2}{(1-\lambda_2)^{2}} .
\]
This proves $\|\mathbf C\|_2<(1-\lambda_2)^{-2}$.  Weyl's monotonicity theorem then gives
$\nu^{\mathrm{st}}_{(j)}\le\|\mathbf C\|_2\,\nu^\ast_{(j)}<\nu^\ast_{(j)}/(1-\lambda_2)^2$ for every
$j$, because $\nu^\ast_{(j)}>0$, and the common-coordinate coupling of step (iv) gives the dominance.
\end{proof}

\subsubsection{\texorpdfstring{Proof of Theorem~\ref{thm:full_network_fosd}}{Proof of Theorem 4}}
\label{app:proof:thm:full_network_fosd}
\begin{proof}
The nonnegative exposures make the covariance matrix of the dynamic output path smaller, in the Loewner order, than the multiple of the identity that is the covariance matrix of the static path. We write both paths as linear maps of the innovation history, bound the dynamic covariance under each initialization, pass the bound through the differencing operator, and couple the two weighted chi-square laws.

Write $a_{k,j}:=\boldsymbol\gamma^\top\mathbf W^{k}\mathbf e_j\ge0$ for the exposure of the aggregate to
firm $j$'s innovation $k$ dates after it arrives, and
$B_j:=\sum_{k\ge0}a_{k,j}=\bigl(\boldsymbol\gamma^\top(\mathbf I-\mathbf W)^{-1}\bigr)_j$, so that
$\|\mathbf B\|_2^2=\|\boldsymbol\gamma^\top(\mathbf I-\mathbf W)^{-1}\|_2^2$.  Let
$\boldsymbol{\epsilon}^{(j)}:=(\epsilon_{j,1},\ldots,\epsilon_{j,T})^\top$ collect firm $j$'s innovations
over the window.  Under Assumption~\ref{assu:innovations} the vectors
$\boldsymbol{\epsilon}^{(1)},\ldots,\boldsymbol{\epsilon}^{(n)}$ are independent
$\mathcal N(\mathbf 0,\sigma^2\mathbf I_T)$, so the static path
$\mathbf y^\ast:=(y_1^\ast,\ldots,y_T^\ast)^\top=\sum_jB_j\boldsymbol{\epsilon}^{(j)}$ has covariance
matrix $\sigma^2\|\mathbf B\|_2^2\mathbf I_T$.  Primitivity gives a date $k_0$ with $\mathbf W^{k}>0$
entrywise for all $k\ge k_0$, hence $a_{k,j}>0$ for every $j$ and every $k\ge k_0$.

We first bound the dynamic covariance from a settled configuration.  With
$\boldsymbol{\epsilon}_t=\mathbf 0$ for $t\le0$ the dynamic path is
$\mathbf y:=(y_1,\ldots,y_T)^\top=\sum_j\mathbf K_j\boldsymbol{\epsilon}^{(j)}$ with
$\mathbf K_j:=\sum_{k=0}^{T-1}a_{k,j}\mathbf L^{k}$ and $\mathbf L$ the lower shift on $\mathbb R^{T}$,
which satisfies $\|\mathbf L^{k}\|_2\le1$.  Independence across firms gives
$\Cov(\mathbf y)=\sigma^2\sum_j\mathbf K_j\mathbf K_j^\top$, and the triangle inequality gives
$\|\mathbf K_j\|_2\le\sum_{k<T}a_{k,j}$, so
\[
\Cov(\mathbf y)\;\preceq\;\sigma^2 S_T\,\mathbf I_T,
\qquad
S_T:=\sum_j\Big(\sum_{k<T}a_{k,j}\Big)^2<\|\mathbf B\|_2^2,
\]
the strict inequality because $a_{k,j}>0$ for $k\ge\max\{k_0,T\}$.

We next bound the dynamic covariance from the stationary law.  The stationary path
$y_t=\sum_{k\ge0}\mathbf a_k^\top\boldsymbol{\epsilon}_{t-k}$, with
$\mathbf a_k:=(a_{k,1},\ldots,a_{k,n})^\top$, is the sum over $j$ of independent moving averages with
coefficients $\{a_{k,j}\}_{k\ge0}$, so for $\mathbf z\in\mathbb R^{T}$ its spectral representation
gives
\[
\mathbf z^\top\Cov(\mathbf y)\mathbf z
=
\frac{\sigma^2}{2\pi}\int_{-\pi}^{\pi}\Big|\sum_{t=1}^{T}z_te^{it\theta}\Big|^2
\sum_j\Big|\sum_{k\ge0}a_{k,j}e^{-ik\theta}\Big|^2d\theta .
\]
For $\theta\notin2\pi\mathbb Z$ the inner modulus is strictly below $B_j$, because the terms with
$k\ge k_0$ are all positive and the unit complex numbers $e^{-ik\theta}$, $k\ge k_0$, are not all
equal.  For $\mathbf z\neq\mathbf 0$ the trigonometric polynomial $\sum_tz_te^{it\theta}$ vanishes at
finitely many $\theta$, so the integrand is strictly below
$\|\mathbf B\|_2^2\,|\sum_tz_te^{it\theta}|^2$ on a set of positive measure, and Parseval's identity
gives $\mathbf z^\top\Cov(\mathbf y)\mathbf z<\sigma^2\|\mathbf B\|_2^2\|\mathbf z\|_2^2$.  Since the
unit sphere of $\mathbb R^{T}$ is compact, $\Cov(\mathbf y)\preceq\sigma^2 S^{\mathrm{st}}\mathbf I_T$
for a constant $S^{\mathrm{st}}<\|\mathbf B\|_2^2$.

We now pass the two bounds through the differencing operator.  Under either initialization,
$\Cov(\mathbf y)\preceq\sigma^2 S\,\mathbf I_T$ with $S<\|\mathbf B\|_2^2$, so
\[
\Cov(\mathbf D_T\mathbf y)
=
\mathbf D_T\Cov(\mathbf y)\mathbf D_T^\top
\;\preceq\;
\sigma^2 S\,\mathbf D_T\mathbf D_T^\top
\;=\;
\frac{S}{\|\mathbf B\|_2^2}\,\Cov(\mathbf D_T\mathbf y^\ast).
\]
Weyl's monotonicity theorem orders the eigenvalues,
\[
\nu_{(j)}\bigl(\Cov(\mathbf D_T\mathbf y)\bigr)
\ \le\ \frac{S}{\|\mathbf B\|_2^2}\,\nu_{(j)}\bigl(\Cov(\mathbf D_T\mathbf y^\ast)\bigr)
\ <\ \nu_{(j)}\bigl(\Cov(\mathbf D_T\mathbf y^\ast)\bigr),
\qquad j=1,\ldots,T-1 .
\]
The last inequality holds because $\mathbf D_T$ has full row rank, so every static eigenvalue
$\sigma^2\|\mathbf B\|_2^2\nu^\ast_{(j)}$ is positive.

We finally couple the two laws.  Both $(T-1)\hat\phi=\|\mathbf D_T\mathbf y\|_2^2$ and
$(T-1)\hat\phi^\ast=\|\mathbf D_T\mathbf y^\ast\|_2^2$ are squared norms of centered Gaussian vectors.
The argument of Lemma~\ref{lem:phi_quadratic_forms} therefore represents each as a weighted
sum of $T-1$ independent $\chi^2(1)$ variables with weights the eigenvalues of the corresponding
covariance matrix.  The common-coordinate coupling of step (iv) in the proof of
Theorem~\ref{thm:phi_finite_ordering}, applied to the two ordered weight profiles, gives two random
variables with the laws of $\hat\phi^\ast$ and $\hat\phi$ that are ordered almost surely, hence
$\hat\phi^\ast\succeq_{\mathrm{FOSD}}\hat\phi$.
\end{proof}

\subsection{\texorpdfstring{Proofs for Section~\ref{sec:spectral_gap_mediates_fat_tails}}{Proofs for Section 5}}

\subsubsection{\texorpdfstring{Proof of Proposition~\ref{prop:fat_tails_condition}}{Proof of Proposition 3}}
\label{app:proof:prop:fat_tails_condition}
\begin{proof}
Part (i) follows by differentiating the logarithms of $\phi=2\sigma^2b^2/(1+\lambda_2)$ and
$\phi^\ast=2\sigma^2b^2/(1-\lambda_2)^2$, and the sign of the static elasticity follows from $(b^2)'<0$ and $\lambda_2'<0$. For (ii),
$\phi_2/\phi_1=(b_2^2/b_1^2)(1+\lambda_{2,1})/(1+\lambda_{2,2})$ with $1+\lambda_{2,1}\ge1$ and
$1+\lambda_{2,2}<1+\varsigma$.
\end{proof}

\subsection{\texorpdfstring{Proofs for Appendix~\ref{app:reversal}}{Proofs for Appendix B}}

\subsubsection{\texorpdfstring{Proof of Proposition~\ref{prop:reversal_frequency}}{Proof of Proposition 4}}
\label{app:proof:prop:reversal_frequency}
\begin{proof}
Fix $\lambda_2\in(0,1)$.  By \eqref{eq:Dy_dynamic_lambda2} and \eqref{eq:Dy_static_lambda2} the pair
$(\Delta y_t,\Delta y_t^\ast)$ is a linear function of the jointly Gaussian innovations
$\{\eta_\tau\}_{\tau\le t}$, hence centered bivariate Gaussian.  In units of $b^2\sigma^2$, the
variances are, by \eqref{eq:var_dynamic_increment} and \eqref{eq:var_static_increment},
\[
V:=\Var(\Delta y_t)=\frac{2}{1+\lambda_2},
\qquad
V^\ast:=\Var(\Delta y_t^\ast)=\frac{2}{(1-\lambda_2)^2}.
\]
For the covariance, the stationary state satisfies $\E[s_{t-1}\eta_{t-1}]=\sigma^2$ and
$\E[s_{t-1}\eta_{t}]=0$, so
\[
C:=\Cov(\Delta y_t,\Delta y_t^\ast)
=\frac{1}{1-\lambda_2}\,
\E\Bigl[\bigl((\lambda_2-1)s_{t-1}+\eta_t\bigr)(\eta_t-\eta_{t-1})\Bigr]
=\frac{(1-\lambda_2)+1}{1-\lambda_2}
=\frac{2-\lambda_2}{1-\lambda_2},
\]
again in units of $b^2\sigma^2$.  Set $D:=\Delta y_t-\Delta y_t^\ast$ and
$S:=\Delta y_t+\Delta y_t^\ast$, so that the reversal event is
\[
\bigl\{|\Delta y_t|>|\Delta y_t^\ast|\bigr\}
=
\bigl\{(\Delta y_t-\Delta y_t^\ast)(\Delta y_t+\Delta y_t^\ast)>0\bigr\}
=
\{DS>0\}.
\]
The pair $(D,S)$ is centered bivariate Gaussian with
$\Var(D)=V+V^\ast-2C$, $\Var(S)=V+V^\ast+2C$, and $\Cov(D,S)=V-V^\ast$.  Direct computation gives
\[
V+V^\ast-2C=\frac{2\lambda_2^{2}(3-\lambda_2)}{(1+\lambda_2)(1-\lambda_2)^2},
\qquad
V+V^\ast+2C=\frac{2\bigl(\lambda_2^{3}-\lambda_2^{2}-2\lambda_2+4\bigr)}{(1+\lambda_2)(1-\lambda_2)^2},
\qquad
V-V^\ast=\frac{-2\lambda_2(3-\lambda_2)}{(1+\lambda_2)(1-\lambda_2)^2},
\]
so both variances are strictly positive on $(0,1)$ and the correlation of $D$ and $S$ is
\[
\rho(\lambda_2)
=
\frac{V-V^\ast}{\sqrt{(V+V^\ast)^2-4C^2}}
=
-\sqrt{\frac{3-\lambda_2}{\lambda_2^{3}-\lambda_2^{2}-2\lambda_2+4}}.
\]
Sheppard's formula for the quadrant probability of a centered bivariate Gaussian pair gives
$\Pr(DS>0)=\tfrac12+\tfrac1\pi\arcsin\rho$, and substituting $\rho(\lambda_2)$ with $\arcsin(-x)=-\arcsin x$ gives \eqref{eq:reversal_frequency}.

We now establish monotonicity and the limits.  Write
$g(\lambda):=(\lambda^{3}-\lambda^{2}-2\lambda+4)/(3-\lambda)$, so that $\rho^2=1/g$.
Differentiating,
\[
g'(\lambda)
=
\frac{-2(\lambda-1)(\lambda^{2}-4\lambda-1)}{(3-\lambda)^2},
\]
and on $(0,1)$ both $\lambda-1$ and $\lambda^{2}-4\lambda-1$ are negative, the latter because its
roots are $2\pm\sqrt5$.  Hence $g'<0$, so $g$ is strictly decreasing, $|\rho|$ is strictly
increasing, and $p$ is strictly decreasing on $(0,1)$.  At the endpoints, $g(0)=4/3$ gives
$|\rho|\to\sqrt3/2$ and $p\to\tfrac12-\tfrac1\pi\arcsin(\sqrt3/2)=\tfrac12-\tfrac13=\tfrac16$, while
$g(1)=1$ gives $|\rho|\to1$ and $p\to0$.

The factor $b^2\sigma^2$ cancels from $\rho$ and nothing depends on $n$, and the granular substitution reproduces the computation verbatim because $g_t=\mathbf m^{\ast\top}\boldsymbol{\epsilon}_t$ is Gaussian under Assumption~\ref{assu:innovations}.
\end{proof}

\subsubsection{\texorpdfstring{Proof of Corollary~\ref{coro:reversal_magnitude}}{Proof of Corollary 2}}
\label{app:proof:coro:reversal_magnitude}
\begin{proof}
The first inequality holds because $|\Delta y_t^\ast|\ge0$ makes
$\{|\Delta y_t|>|\Delta y_t^\ast|+\varepsilon\}$ a subset of $\{|\Delta y_t|>\varepsilon\}$.  The
second is Chebyshev's inequality with $\Var(\Delta y_t)=2b_n^2\sigma^2/(1+\lambda_2)$ from
\eqref{eq:var_dynamic_increment}, and it uses no Gaussianity.  The frequency statement is
Proposition~\ref{prop:reversal_frequency}, whose formula is free of $b_n$.
\end{proof}

\subsubsection{\texorpdfstring{Proof of Lemma~\ref{lem:mean_invariance}}{Proof of Lemma 5}}
\label{app:proof:lem:mean_invariance}
\begin{proof}
Fix $\mu\in\R$ and write $\boldsymbol{\epsilon}_t=\mu\mathbf 1+\tilde{\boldsymbol{\epsilon}}_t$ with
$\tilde{\boldsymbol{\epsilon}}_t$ centered.  By linearity of the stationary superposition,
\[
y_t(\mu)
=
\boldsymbol\gamma^\top\sum_{\tau\ge0}\mathbf W^{\tau}\bigl(\mu\mathbf 1+\tilde{\boldsymbol{\epsilon}}_{t-\tau}\bigr)
=
\mu\sum_{\tau\ge0}\varsigma^{\tau}
+y_t(0)
=
\frac{\mu}{1-\varsigma}+y_t(0).
\]
The middle equality uses the constant row sums $\mathbf W^{\tau}\mathbf 1=\varsigma^{\tau}\mathbf 1$
and $\boldsymbol\gamma^\top\mathbf 1=1$.  The static mapping shifts by the same constant,
$y_t^\ast(\mu)=\boldsymbol\gamma^\top(\mathbf I-\mathbf W)^{-1}(\mu\mathbf 1+\tilde{\boldsymbol{\epsilon}}_t)
=\mu/(1-\varsigma)+y_t^\ast(0)$.  This proves (i), and differencing removes the constants, proving
(ii).  For (iii), biorthogonality gives $\mathbf u_2^\top\mathbf v_1=0$ with $\mathbf v_1=\mathbf 1$,
so $\tilde{\mathbf u}_2^\top\mathbf 1=0$ and $\E[\eta_t]=\mu\,\tilde{\mathbf u}_2^\top\mathbf 1=0$.
Under the settled initialization of \eqref{eq:s_recursion} the
identity in (ii) holds up to the deterministic path $\mu\varsigma^{t-1}$ in the granular channel,
which vanishes geometrically and which the discussion after Definition~\ref{def:aggregate_output} treats.
\end{proof}

\subsection{\texorpdfstring{Proofs for Appendix~\ref{app:levels_vs_increments}}{Proofs for Appendix C}}

\subsubsection{\texorpdfstring{Proof of Lemma~\ref{lem:levels_increments_translation}}{Proof of Lemma 6}}
\label{app:proof:lem:levels_increments_translation}
\begin{proof}
Under \eqref{eq:app_levels_setup} and the stationary law, $s_t=\sum_{\tau\ge0}\lambda_2^{\,\tau}\eta_{t-\tau}$, so $\Var(s_t)=\sigma^2/(1-\lambda_2^2)$ and, since $\eta_t$ is uncorrelated with $s_{t-1}$, $\Cov(s_t,s_{t-1})=\lambda_2\Var(s_t)$. Hence $\Var(y_t)=b^2\sigma^2/(1-\lambda_2^2)$ and
\[
\Var(\Delta y_t)=2\Var(y_t)-2\Cov(y_t,y_{t-1})=2(1-\lambda_2)\Var(y_t)=\frac{2b^2\sigma^2}{1+\lambda_2}.
\]
In the static benchmark $\Var(y_t^\ast)=b^2\sigma^2/(1-\lambda_2)^2$ and, since $\{y_t^\ast\}$ is i.i.d.\ across $t$, $\Var(\Delta y_t^\ast)=2\Var(y_t^\ast)$. Taking ratios gives \eqref{eq:app_levels_vs_incr_ratio}.
\end{proof}

\subsection{\texorpdfstring{Proofs for Appendix~\ref{app:fat_tails_degree_approx}}{Proofs for Appendix D}}

\subsubsection{\texorpdfstring{Proof of Proposition~\ref{prop:fat_tail_moments_and_proxy}}{Proof of Proposition 5}}
\label{app:proof:prop:fat_tail_moments_and_proxy}
\begin{proof}
The moments are the integrals
\[
\E[d]=\int_1^{d_{\max}} x f(x)\,dx
=
\frac{\alpha}{\alpha-1}\,\frac{1-d_{\max}^{1-\alpha}}{1-d_{\max}^{-\alpha}},
\qquad
\E[d^2]=\int_1^{d_{\max}} x^2 f(x)\,dx
=
\frac{\alpha}{\alpha-2}\,\frac{1-d_{\max}^{2-\alpha}}{1-d_{\max}^{-\alpha}}
\]
for $\alpha\neq2$, and $\Var(d)=\E[d^2]-\E[d]^2$.  Imposing $d_{\max}=n^{1/\alpha}$ gives
$d_{\max}^{-\alpha}=n^{-1}$, $d_{\max}^{1-\alpha}=n^{(1-\alpha)/\alpha}$, and
$d_{\max}^{2-\alpha}=n^{(2-\alpha)/\alpha}$, which is \eqref{eq:degree_mean_var_alpha_n_compact_app}
with the factors \eqref{eq:psi_def_app}.  At $\alpha=2$ the second integral is
$C\int_1^{d_{\max}}x^{-1}\,dx=C\log d_{\max}=\log n/(1-n^{-1})$, and the first is the $\alpha=2$ case of
the display.  The loading follows from \eqref{eq:utilde_vtilde_alpha_app} with $\bar d_n$,
$\overline{d^2}_n$, and $V_n$ replaced by $\E[d]$, $\E[d^2]$, and $\Var(d)$, and
$\boldsymbol\gamma^\top\mathbf 1=1$.
\end{proof}

%==========================================================
\section{Finite-horizon reversals: exact frequency, vanishing magnitude, and one-sided innovations}
\label{app:reversal}

Theorem~\ref{thm:phi_finite_ordering} orders the finite-horizon distributions of static and dynamic aggregate volatility, and the ordering is distributional rather than pathwise. The dynamic economy can generate a larger swing than the fully processed benchmark built from the same innovation history, so that at some dates
\begin{equation}
|\Delta y_t|>|\Delta y_t^\ast|.
\label{eq:overshoot_event}
\end{equation}
We call \eqref{eq:overshoot_event} a finite-horizon reversal of the dynamic output increment. Throughout this appendix we work with the reallocation channel of Section~\ref{sec:ell_to_eigenmodes} under its stationary law and, as in Section~\ref{sec:eigenmodes}, drop the channel superscript.

A reversal reflects interference across vintages. In the dynamic economy new innovations arrive while past innovations are still being processed, so the inherited transient state carries unfinished imbalances forward, and the new shock can either offset or reinforce what is already in motion. The static benchmark evaluates the same innovations only after full processing through the resolvent, and so abstracts from the transient composition of intermediate flows at the date the shock arrives. Consider a date at which the propagation of an earlier shock places a large firm in the middle of an unusually high transient inflow of intermediates, and suppose that a new productivity shock raises that firm's productivity at the same date. Dynamic output can jump, because the firm amplifies the innovation while it operates with an elevated intermediate-input flow, whereas the static benchmark responds to the same innovation only through the firm's fully processed intermediate share.

By the increment representations \eqref{eq:Dy_dynamic_lambda2} and \eqref{eq:Dy_static_lambda2}, the reversal event \eqref{eq:overshoot_event} is
\begin{equation}
\bigl|(\lambda_2-1)s_{t-1}+\eta_t\bigr|
>
\frac{1}{1-\lambda_2}\,|\eta_t-\eta_{t-1}|.
\label{eq:overshoot_core_f}
\end{equation}
The inherited state enters the dynamic increment through the mean-reversion coefficient $\lambda_2-1$ and the new innovation at unit gain, while the difference of two adjacent innovations enters the static increment through the resolvent gain. A reversal occurs when the static difference $\eta_t-\eta_{t-1}$ is small relative to the dynamic combination, whether because the two adjacent innovations nearly coincide or because the new innovation reinforces the inherited state. How often this happens depends on the network only through $\lambda_2$, and we now compute the frequency exactly.

\begin{proposi}[The exact frequency of finite-horizon reversals]
\label{prop:reversal_frequency}
Maintain Assumption~\ref{assu:innovations} and the stationary law of the reallocation channel, and fix $\lambda_2\in(0,1)$. Then
\begin{equation}
\Pr\bigl(|\Delta y_t|>|\Delta y_t^\ast|\bigr)
=
p(\lambda_2)
:=
\frac{1}{2}-\frac{1}{\pi}
\arcsin\!\sqrt{\frac{3-\lambda_2}{\lambda_2^{3}-\lambda_2^{2}-2\lambda_2+4}},
\label{eq:reversal_frequency}
\end{equation}
and $p$ is strictly decreasing on $(0,1)$, with $p(\lambda_2)\to1/6$ as $\lambda_2\downarrow0$ and $p(\lambda_2)\to0$ as $\lambda_2\uparrow1$.
\end{proposi}

\noindent\emph{Proof in Appendix~\ref{app:proof:prop:reversal_frequency}.}

The frequency depends on $\lambda_2$ alone and is free of the loading $b$, the innovation scale $\sigma$, and the number of firms $n$. In the fast-mixing limit reversals mark one date in six. Over the calibrated range $\lambda_2\in[0.56,0.80]$ the frequency falls from about $0.108$ to $0.057$, between one date in nine and one in eighteen, and in the slow-mixing limit reversals disappear. Under the granular substitution $(b,\lambda_2,\eta_t)\mapsto(1,\lambda_1,g_t)$ the same formula gives the reversal frequency of the granular channel at $\lambda_1=\varsigma$. The frequency does not measure the size of a reversal, and the economy's scale enters through that size.

\begin{coro}[Vanishing reversal magnitudes in large economies]
\label{coro:reversal_magnitude}
Maintain the setting of Proposition~\ref{prop:reversal_frequency} and let $b=b_n\to0$ as $n\to\infty$. Then for every $\varepsilon>0$,
\[
\Pr\bigl(|\Delta y_t|>|\Delta y_t^\ast|+\varepsilon\bigr)
\ \le\
\Pr\bigl(|\Delta y_t|>\varepsilon\bigr)
\ \le\
\frac{2b_n^2\sigma^2}{(1+\lambda_2)\,\varepsilon^2}
\ \longrightarrow\ 0,
\]
while $\Pr(|\Delta y_t|>|\Delta y_t^\ast|)=p(\lambda_2)$ at every $n$.
\end{coro}

\noindent\emph{Proof in Appendix~\ref{app:proof:coro:reversal_magnitude}.}

The loading vanishes along a sequence of economies whenever the aggregation weights are delocalized, as under the degree proxy \eqref{eq:b_alpha_explicit_fattails}. So the size of the economy leaves the frequency of reversals untouched and shrinks their magnitude. Reversals occur at the same rate in an economy of ten firms as in one of ten million firms, but in the large economy they are too small to see.

\smallskip
We close the appendix with the role of the innovation mean. The paper's benchmark centers the innovations at zero, and we now show that nothing built from increments depends on that centering.

\begin{lemma}[Invariance of increments to innovation means]
\label{lem:mean_invariance}
Replace the zero mean of Assumption~\ref{assu:innovations} by $\E[\epsilon_{i,t}]=\mu$ for an arbitrary common mean $\mu\in\R$, keeping independence across dates and firms and the variance $\sigma^2$, with no distributional form assumed. Write $y_t(\mu)$ and $y_t^\ast(\mu)$ for the dynamic and static aggregates so generated, and $y_t(0)$, $y_t^\ast(0)$ for the processes driven by the centered innovations $\boldsymbol{\epsilon}_t-\mu\mathbf 1$. Then, in the stationary regime,
\begin{itemize}
\item[(i)] $y_t(\mu)=y_t(0)+\mu/(1-\varsigma)$ and $y_t^\ast(\mu)=y_t^\ast(0)+\mu/(1-\varsigma)$ for every $t$.
\item[(ii)] $\Delta y_t(\mu)=\Delta y_t(0)$ and $\Delta y_t^\ast(\mu)=\Delta y_t^\ast(0)$ pathwise.
\item[(iii)] $\E[\eta_t]=\mu\,\tilde{\mathbf u}_2^\top\mathbf 1=0$.
\end{itemize}
\end{lemma}

\noindent\emph{Proof in Appendix~\ref{app:proof:lem:mean_invariance}.}

Part (ii) says that every volatility, tail-risk, and reversal object built from increments coincides with that of the centered system. Part (iii) says that the projected innovation is centered even when every $\epsilon_{i,t}$ has the same sign, because biorthogonality with $\mathbf v_1=\mathbf 1$ forces $\tilde{\mathbf u}_2^\top\mathbf 1=0$. Every second-moment result of the paper therefore holds verbatim for innovations supported on the positive or on the negative half-line. The variance bound of Lemma~\ref{lem:neumann_bounds_L}, the attenuation factor of Theorem~\ref{thm:attenuation}, Corollary~\ref{coro:granular_perron_rate}, Proposition~\ref{prop:singular_aggregator}, and the magnitude bound of Corollary~\ref{coro:reversal_magnitude} use only the second moments of the increments, which part (ii) leaves unchanged. The results that compare probabilities, namely the tail ordering of Lemma~\ref{lem:neumann_bounds_L}, Proposition~\ref{prop:tail_amp}, the finite-window laws of Section~\ref{sec:eigenmodes}, and the reversal frequency of Proposition~\ref{prop:reversal_frequency}, use the Gaussian shape of the increments, and finite variance alone orders neither tails nor reversals. An innovation law with an atom, for instance, makes the static difference $\eta_t-\eta_{t-1}$ vanish with positive probability while the dynamic increment, whose inherited state has no atoms, does not. The reversal probability then acquires a floor that can exceed $p(\lambda_2)$, and the tail ordering can fail at small thresholds for the same reason.

Aggregation restores Gaussianity when the transient direction is delocalized. The projected innovation $\eta_t=\tilde{\mathbf u}_2^\top\boldsymbol{\epsilon}_t$ then sums many independent one-sided shocks with mixed-sign weights, mixed because $\tilde{\mathbf u}_2^\top\mathbf 1=0$. The stationary increment $b\bigl((\lambda_2-1)s_{t-1}+\eta_t\bigr)$ is then a linear combination of the independent primitive innovations with square-summable weights. If $\max_i|\tilde u_{2,i}|\to0$ along a sequence of economies, the Lindeberg condition holds for these weights, and the Lindeberg--Feller theorem drives the joint law of $(\Delta y_t,\Delta y_t^\ast)$ to the bivariate Gaussian, hence the tail ordering to that of Section~\ref{subsec:omega_dist_dom} and the reversal frequency to $p(\lambda_2)$. The frequency departs from $p(\lambda_2)$ when the mode-level innovation is itself one-sided, which is the case when the transient direction is localized on few firms. Corollary~\ref{coro:reversal_magnitude} stands in every case, because its magnitude bound uses only Chebyshev's inequality, so reversals of any fixed size vanish in large economies whatever the sign structure of the innovations.

%====================================================
\section{Levels versus increments aggregate volatility}
\label{app:levels_vs_increments}
%====================================================

The paper measures aggregate volatility with increments of log output, which is the object of much of the empirical literature and the one that isolates the interference mechanism, since differencing subtracts the components inherited from earlier vintages. The theoretical literature on production networks often works instead with the variance of log output levels. This appendix shows that the static--dynamic comparison survives the change of definition and records what the change alters.

Recall from Section~\ref{sec:ell_to_eigenmodes} the reallocation channel and its static counterpart, whose levels are, dropping the channel superscript as in Section~\ref{sec:eigenmodes},
\begin{equation}
y_t^\ast = \frac{b}{1-\lambda_2}\,\eta_t,
\qquad
y_t = b\,s_t,
\qquad
s_t=\lambda_2 s_{t-1}+\eta_t,
\label{eq:app_levels_setup}
\end{equation}
where $\{\eta_t\}$ is centered with $\Var(\eta_t)=\sigma^2$. The static level is the level that a single vintage attains under full processing, and the resolvent gain $(1-\lambda_2)^{-1}$ enlarges it as $\lambda_2$ rises. The dynamic level carries a stock of partially processed vintages forward. Levels measure that stock directly, while increments measure how much the stock changes from one date to the next, and the inherited component that differencing subtracts grows with persistence.

\begin{lemma}[Levels--increments translation]
\label{lem:levels_increments_translation}
Maintain \eqref{eq:app_levels_setup} under the stationary law. Then
\begin{equation}
\frac{\Var(y_t)}{\Var(y_t^\ast)}
=
\frac{1-\lambda_2}{1+\lambda_2},
\qquad
\frac{\Var(\Delta y_t)}{\Var(\Delta y_t^\ast)}
=
\frac{(1-\lambda_2)^2}{1+\lambda_2}.
\label{eq:app_levels_vs_incr_ratio}
\end{equation}
\end{lemma}

\noindent\emph{Proof in Appendix~\ref{app:proof:lem:levels_increments_translation}.}

Overlap attenuates the static benchmark under both definitions, and the increment ratio, which is the attenuation factor $\mathcal R(\lambda_2)$ of Theorem~\ref{thm:attenuation}, is $(1-\lambda_2)$ times the level ratio. A level-based definition therefore changes no conclusion of the paper that rests on the gap between static and dynamic volatility. What it changes is the comparative static of the dynamic object with respect to persistence. In levels the dynamic variance $b^2\sigma^2/(1-\lambda_2^2)$ rises with $\lambda_2$, because overlap keeps old vintages in the state longer and the stock of unfinished propagation grows. In increments the dynamic variance $2b^2\sigma^2/(1+\lambda_2)$ falls with $\lambda_2$, because higher persistence makes the inherited component larger relative to the one-step update, and differencing subtracts the inherited component. The contrast is largest as $\lambda_2\uparrow1$ in the scalar model \eqref{eq:app_levels_setup}, where the static variance diverges at the rate $(1-\lambda_2)^{-2}$ in both levels and increments, the dynamic level variance diverges at the slower rate $(1-\lambda_2)^{-1}$, and the dynamic increment variance stays bounded and tends to $b^2\sigma^2$. Within the production-network model $\lambda_2=\varsigma\lambda_2(\mathbf A)<\varsigma$, so that limit is reached only jointly with $\varsigma\to1$.

%====================================================
\section{Power-law tail exponent, eigenvector, and eigenvalue}
\label{app:fat_tails_degree_approx}
%====================================================

This appendix collects the auxiliary approximations used in Section~\ref{sec:spectral_gap_mediates_fat_tails} to make the role of degree heterogeneity explicit. The exact results of Sections~\ref{sec:volatility_leontief} and \ref{sec:eigenmodes} summarize the cross-sectional structure of shocks through the left and right eigenvectors of the dominant transient mode, which determine the loading $b$ and the innovation $\eta_t$, and degree heterogeneity enters there only through the shape it gives the eigenvectors.

\smallskip
Section~\ref{sec:spectral_gap_mediates_fat_tails} treats the tail exponent $\alpha$ as an explicit primitive, and $\alpha$ can then move aggregate volatility by shifting cross-sectional exposure and by shifting persistence through $\lambda_2=\lambda_2(\alpha)$. We first state the degree-based proxy for the dominant transient direction as an explicit assumption, give the moment formulas under the cutoff normalization used in the paper, and record where the proxy fails. We then collect the theoretical links between spectral gaps, bottlenecks, and mixing that suggest why fatter tails go with slower mixing.

%----------------------------------------------------
\subsection{A degree-based proxy under truncated power laws}
%----------------------------------------------------

Section~\ref{sec:spectral_gap_mediates_fat_tails} needs the dominant transient direction as a
function of the tail exponent.  No eigenvector equation delivers that direction from a degree
distribution, so we impose a proxy as an explicit assumption, separate it from the exact moment calculation, and record where the proxy can fail.

\begin{assu}[Degree proxy for the dominant transient direction]
\label{assu:degree_proxy}
Let $\mathbf d=(d_1,\ldots,d_n)^\top$ be a degree vector with realized mean
$\bar d_n:=n^{-1}\sum_i d_i$, realized second moment $\overline{d^2}_n:=n^{-1}\sum_i d_i^2$, and
realized variance $V_n:=\overline{d^2}_n-\bar d_n^{\,2}>0$.  We take the rescaled dominant transient
eigenvectors of \eqref{eq:u2_normalization} to be
\begin{equation}
\tilde{\mathbf u}_2
=
\frac{\mathbf d-\bar d_n\mathbf 1}{\sqrt{n\,V_n}},
\qquad
\tilde{\mathbf v}_2
=
\frac{\mathbf d-\bigl(\overline{d^2}_n/\bar d_n\bigr)\mathbf 1}{\sqrt{n\,V_n}}.
\label{eq:utilde_vtilde_alpha_app}
\end{equation}
\end{assu}

\noindent
The two vectors satisfy the normalizations of \eqref{eq:u2_normalization} exactly.  The identity
$\|\tilde{\mathbf u}_2\|_2^2=\sum_i(d_i-\bar d_n)^2/(nV_n)=1$ holds by the definition of $V_n$, and
expanding the inner product gives
$\sum_i(d_i-\bar d_n)\bigl(d_i-\overline{d^2}_n/\bar d_n\bigr)=n\overline{d^2}_n-n\overline{d^2}_n-n\bar d_n^{\,2}+n\overline{d^2}_n=nV_n$,
so $\tilde{\mathbf u}_2^\top\tilde{\mathbf v}_2=1$.

The two vectors also satisfy the two constraints
that biorthogonality places on any dominant transient pair when the stationary balance vector is
proportional to the degree vector.  The left constraint $\mathbf u_2^\top\mathbf v_1=\mathbf u_2^\top\mathbf 1=0$
holds because $\tilde{\mathbf u}_2$ is centered at the realized mean, and the right constraint
$\mathbf u_1^\top\mathbf v_2=\mathbf m^{\ast\top}\mathbf v_2=0$ holds because $\tilde{\mathbf v}_2$ is centered at
the size-weighted mean $\overline{d^2}_n/\bar d_n$, since $\mathbf d^\top\tilde{\mathbf v}_2\propto
n\overline{d^2}_n-n\overline{d^2}_n=0$.  Under locally even weights, $a_{ji}\approx1/d_i$ on the links
into buyer $i$, and matched in- and out-degrees, the stationary balance vector is the degree vector,
$\mathbf A\mathbf d=\mathbf d$ and $\mathbf m^\ast=\mathbf d/(n\bar d_n)$.  The right constraint is then the one Proposition~\ref{prop:singular_aggregator} requires, namely that the loading $b=\boldsymbol\gamma^\top\tilde{\mathbf v}_2$ vanishes at size weights.  Centering both vectors at $\bar d_n$ would satisfy the left constraint and
violate the right one, giving $\mathbf m^{\ast\top}\mathbf v_2\propto V_n\neq0$ and a nonzero loading at size weights.

The two constraints are necessary conditions on the pair and nothing more. An eigenvector equation would fix the direction within the transient subspace, whereas the two centerings leave open which transient mode the centered degree vector is closest to.  The discussion after
Proposition~\ref{prop:fat_tail_moments_and_proxy} gives a family of networks satisfying both regularities on which the centered degree vector lies in modes faster than the dominant one.  The assumption is therefore a
modelling choice for the comparative statics of Section~\ref{sec:spectral_gap_mediates_fat_tails},
and none of the exact channel results of Sections~\ref{sec:volatility_leontief}
and~\ref{sec:eigenmodes} depends on it.

\begin{proposi}[Moments of a truncated power law at the natural cutoff]
\label{prop:fat_tail_moments_and_proxy}
Let degrees be drawn from the continuous Pareto density $f(x)=Cx^{-(\alpha+1)}$ on $[1,d_{\max}]$
with $\alpha>1$ and $C=\alpha/(1-d_{\max}^{-\alpha})$, and let the cutoff be $d_{\max}=n^{1/\alpha}$, so
that $d_{\max}^{-\alpha}=n^{-1}$.  For $\alpha\neq2$ the first two moments are
\begin{equation}
\E[d]\ =\ \frac{\alpha}{\alpha-1}\,\psi_n^{(1)}(\alpha),
\qquad
\E[d^2]\ =\ \frac{\alpha}{\alpha-2}\,\psi_n^{(2)}(\alpha),
\qquad
\Var(d)\ =\ \E[d^2]-\E[d]^2,
\label{eq:degree_mean_var_alpha_n_compact_app}
\end{equation}
with truncation factors
\begin{equation}
\psi_n^{(k)}(\alpha):=\frac{1-n^{(k-\alpha)/\alpha}}{1-n^{-1}},
\qquad k=1,2.
\label{eq:psi_def_app}
\end{equation}
At $\alpha=2$ the moments are $\E[d]=2(1-n^{-1/2})/(1-n^{-1})$ and $\E[d^2]=\log n/(1-n^{-1})$.  Under
Assumption~\ref{assu:degree_proxy} with these population moments in place of the realized moments,
the loading of an admissible weight vector $\boldsymbol\gamma$ is
\begin{equation}
b(\alpha)=\boldsymbol\gamma^\top\tilde{\mathbf v}_2(\alpha)
\ =\
\frac{\E_{\boldsymbol\gamma}[d]-\E[d^2]/\E[d]}{\sqrt{n\,\Var(d)}},
\qquad
\E_{\boldsymbol\gamma}[d]:=\boldsymbol\gamma^\top\mathbf d .
\label{eq:b_alpha_population_app}
\end{equation}
\end{proposi}

\noindent\emph{Proof in Appendix~\ref{app:proof:prop:fat_tail_moments_and_proxy}.}

\smallskip
The proxy can fail because the centered degree vector need not be the dominant transient mode, and because the realized moments need not concentrate. On the first count, both regularities behind Assumption~\ref{assu:degree_proxy} can hold while the centered degree vector lies in faster modes. Let $\mathbf B$ be the adjacency matrix of a path of $n\ge3$ nodes with a self-loop at every node, so
that $\mathbf d=(2,3,\ldots,3,2)^\top$, and let $\mathbf A=\mathbf B\,\mathrm{diag}(\mathbf d)^{-1}$, which
is primitive and has matched in- and out-degrees and locally even shares.  The matrix $\mathbf A$ is
similar to the symmetric tridiagonal matrix
$\mathrm{diag}(\mathbf d)^{-1/2}\mathbf B\,\mathrm{diag}(\mathbf d)^{-1/2}$, whose off-diagonal entries
are positive.  By the oscillation theory of Jacobi matrices its eigenvalues are therefore real and
simple, and its $k$-th eigenvector, in decreasing order of eigenvalue, changes sign exactly $k-1$
times along the path.  Both $\mathbf B$ and $\mathbf d$ are invariant under reversal of the path, so every
eigenvector is either symmetric or antisymmetric under reversal, and the dominant transient
eigenvectors, which change sign once, are antisymmetric.  The degree vector is symmetric under
reversal, so both proxies in \eqref{eq:utilde_vtilde_alpha_app} are orthogonal to the dominant
transient pair and lie in the span of the reversal-symmetric transient modes, every one of which is
faster than the dominant one.  At uniform weights the true dominant loading is zero, while the proxy
gives $b=-\sqrt{V_n/n}/\bar d_n$, which is not zero.

On the second count, the realized moments in \eqref{eq:utilde_vtilde_alpha_app} can be replaced by the population moments of
Proposition~\ref{prop:fat_tail_moments_and_proxy} only when the empirical second moment
concentrates, and at the natural cutoff it does so only for $\alpha>2$.  For independent draws with
$1<\alpha<2$, $\E[d^2]\sim\frac{\alpha}{2-\alpha}n^{(2-\alpha)/\alpha}$ and
$\E[d^4]\sim\frac{\alpha}{4-\alpha}n^{(4-\alpha)/\alpha}$, so
\[
\frac{\Var\bigl(n^{-1}\sum_{i=1}^n d_i^2\bigr)}{\E[d^2]^2}
\ \longrightarrow\
\frac{(2-\alpha)^2}{\alpha(4-\alpha)}>0 ,
\]
so the largest degrees dominate the second moment, which does not concentrate.  For $\alpha>2$ the
same ratio tends to zero.  The comparative statics of
Section~\ref{sec:spectral_gap_mediates_fat_tails} therefore use the population moments as a
description of the realized moments that is exact in the limit for $\alpha>2$ and approximate below
it, while every normalization identity in \eqref{eq:utilde_vtilde_alpha_app} holds for the realized
moments at any $n$.

%----------------------------------------------------
\subsection{Relation between dominant transient eigenvalue and power-law tail exponent}
%----------------------------------------------------

When $\alpha$ is allowed to influence persistence, the comparative statics of Section~\ref{sec:spectral_gap_mediates_fat_tails} depend not only on how $\alpha$ reshapes exposure but also on how $\alpha$ reshapes the overlap spectrum through $\lambda_2=\lambda_2(\alpha)$. The theoretical connections between spectral gaps, bottlenecks, and mixing rates, together with results on large random graphs, suggest a direction for this mapping. Greater degree inequality can tighten bottlenecks, reduce the spectral gap, and therefore push $|\lambda_2|$ upward, which is slower mixing. We summarize those connections here.

\smallskip
The convergence of network propagation dynamics to long-run behavior is governed by the spectral
gap, so a larger $|\lambda_2|$ corresponds to slower mixing \citep{LevinPeresWilmer2009}.  Slow mixing is usually understood through bottlenecks. Conductance bounds and Cheeger-type inequalities relate low-conductance cuts to a small spectral gap, and hence to $|\lambda_2|$ close
to one \citep{JerrumSinclair1988,LawlerSokal1988}.  More generally, multicommodity-flow and
canonical-paths arguments show that congestion along a small set of `connector' nodes or edges can
control the worst-case mixing rate and thereby determine the slowest mode \citep{Sinclair1992}. When degree distributions become more unequal, propagation can become increasingly concentrated
around hubs, so that communication between large regions of the graph relies on a relatively small
set of connectors.  Concentration on hubs makes low-conductance bottlenecks more likely, shrinking the spectral gap
and pushing $|\lambda_2|$ upward.  For concrete results where such obstructions control mixing on
large random graphs, see \citet{BenjaminiKozmaWormald2006} and \citet{FountoulakisReed2007}.

%====================================================
\section{Transient spectrum of the reconstructed networks}
\label{app:abm_eigenvalue}
%====================================================

The closed forms of Sections~\ref{sec:volatility_leontief} to~\ref{sec:spectral_gap_mediates_fat_tails} assume that the dominant transient eigenvalue is real, positive, and simple (Assumption~\ref{assu:dominant_transient}), so that $\mathcal R(\lambda_2)$ is a number between zero and one. Directed production networks have circuits $\mathbf A$ that are not symmetric, so their transient eigenvalues need not be real and may be negative. This appendix records the realized spectrum of the reconstructed circuits of Section~\ref{sec:abm}, which fixes the calibrated range of Section~\ref{subsec:boe_vol}, and checks that the size-weighted volatility ratio is unmoved by that spectrum, as Proposition~\ref{prop:singular_aggregator} requires.

\smallskip
For each of $100$ reconstructed US networks of about $10^4$ firms we form the monetary circuit $\mathbf A$ and record its subdominant eigenvalue $\lambda_2(\mathbf A)$, the second by modulus, the Perron mode being real and equal to one. The corresponding transient rate for output is $\varsigma\lambda_2(\mathbf A)$, which is what enters $\mathcal R(\lambda_2)$.\footnote{The eigenvalue is read off the explicit spectrum of the propagation matrix via sparse Arnoldi iterations, which is feasible on the $\approx 10^4$-firm reconstructions but not on the $10^6$--$10^7$-firm operators, which we do not diagonalize. The scaling result of Section~\ref{sec:abm}, that the dynamic-to-static ratio is essentially flat in $n$, lets the robustness recorded here extrapolate to the large economy.}

Figure~\ref{fig:eigenvalue} summarizes the result. The left panel plots $\lambda_2(\mathbf A)$ in the complex plane, where the points sit overwhelmingly on the positive real axis, clustered around $0.7$--$1.0$, with only a handful off it. The right panel histograms $\mathrm{Re}\,\lambda_2(\mathbf A)$, which is positive for almost all networks. Across the $100$ networks, $\lambda_2(\mathbf A)$ is complex in only $5$ and has negative real part in only $7$, so the configuration assumed for the closed forms holds in almost every reconstruction. For the $86$ networks in the cluster, the implied output transient rate at $\varsigma=0.8$ lies in $[0.56,0.80)$, which is the range used for the channel rates in Section~\ref{subsec:boe_vol}. The remaining $14$ have a smaller positive or a negative subdominant eigenvalue. Primitivity of $\mathbf A$ excludes $\lambda_2(\mathbf A)=1$, so the rate is strictly below $\lambda_1=\varsigma$ as Assumption~\ref{assu:dominant_transient} requires.

\begin{figure}[H]
\centering
\includegraphics[width=\textwidth]{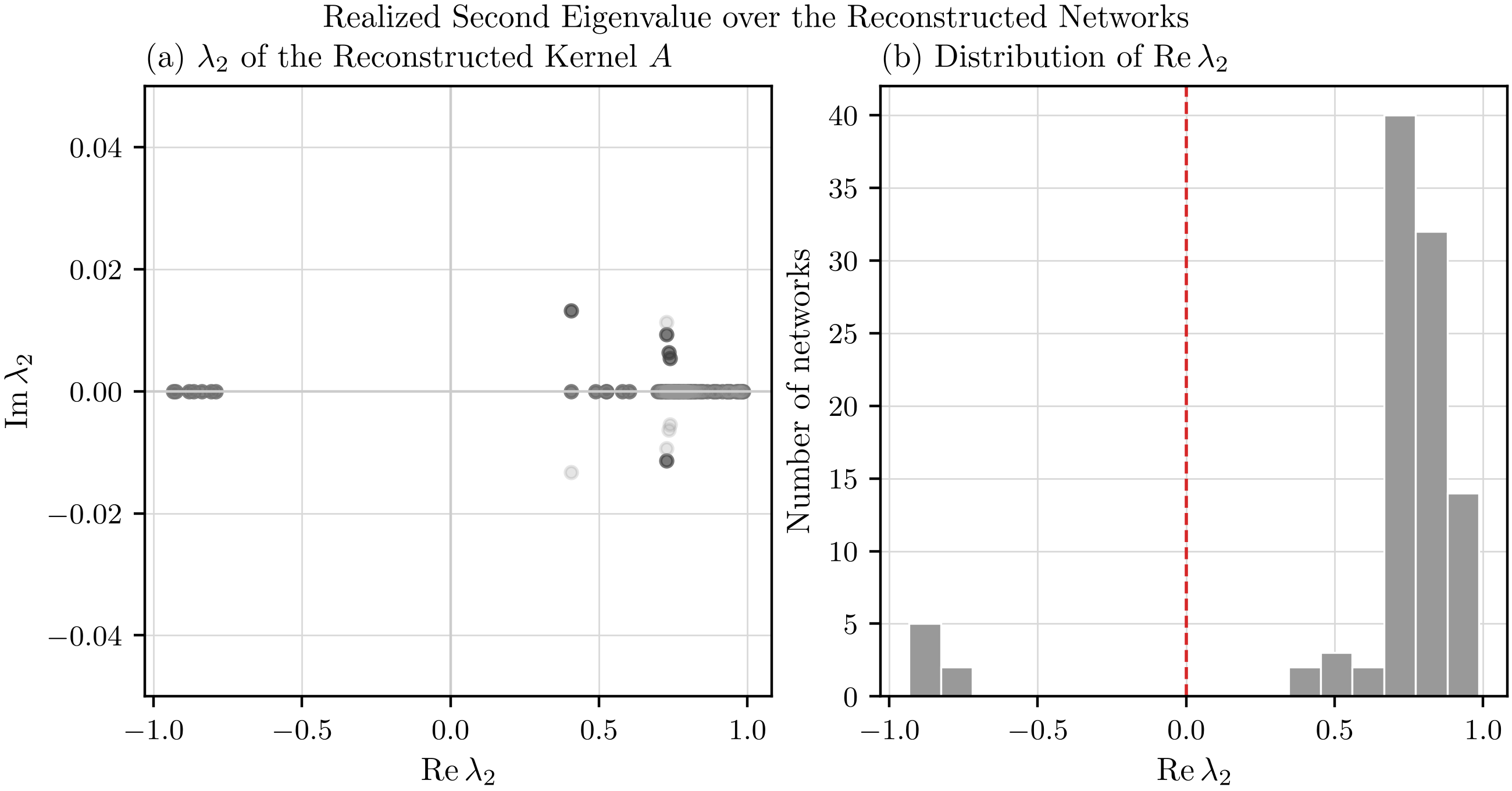}
\caption{Realized subdominant eigenvalue $\lambda_2(\mathbf A)$ of the monetary circuit across $100$
reconstructed US networks ($\approx 10^4$ firms).  Left: $\lambda_2(\mathbf A)$ in the complex plane (the panel labels $\mathbf A$ the reconstructed kernel), sitting largely on the positive real axis and clustered near $0.7$--$1.0$.  Right: histogram of
$\mathrm{Re}\,\lambda_2(\mathbf A)$, positive for almost all networks ($5$ of $100$ are complex, $7$
have negative real part).  The transient rate for output is $\varsigma$ times these values.}
\label{fig:eigenvalue}
\end{figure}

Figure~\ref{fig:agg_vol_robust} reports the corresponding volatilities, dynamic against static aggregate volatility, one point per network, colored by whether $\lambda_2(\mathbf A)$ is real or complex. Every point lies well below the $45^\circ$ line $\sqrt{\phi}=\sqrt{\phi^\ast}$, with the dynamic-to-static volatility ratio at most $0.166$ and median $0.148$, whether the eigenvalue is real or complex. The ratio is unchanged by the sign and phase of $\lambda_2$ because of the weights, since under size weighting Proposition~\ref{prop:singular_aggregator} removes every transient mode, real or complex, so the ratio cannot depend on where $\lambda_2$ sits in the unit disk. The ordering therefore survives complex and negative subdominant eigenvalues.

\begin{figure}[H]
\centering
\includegraphics[width=0.7\textwidth]{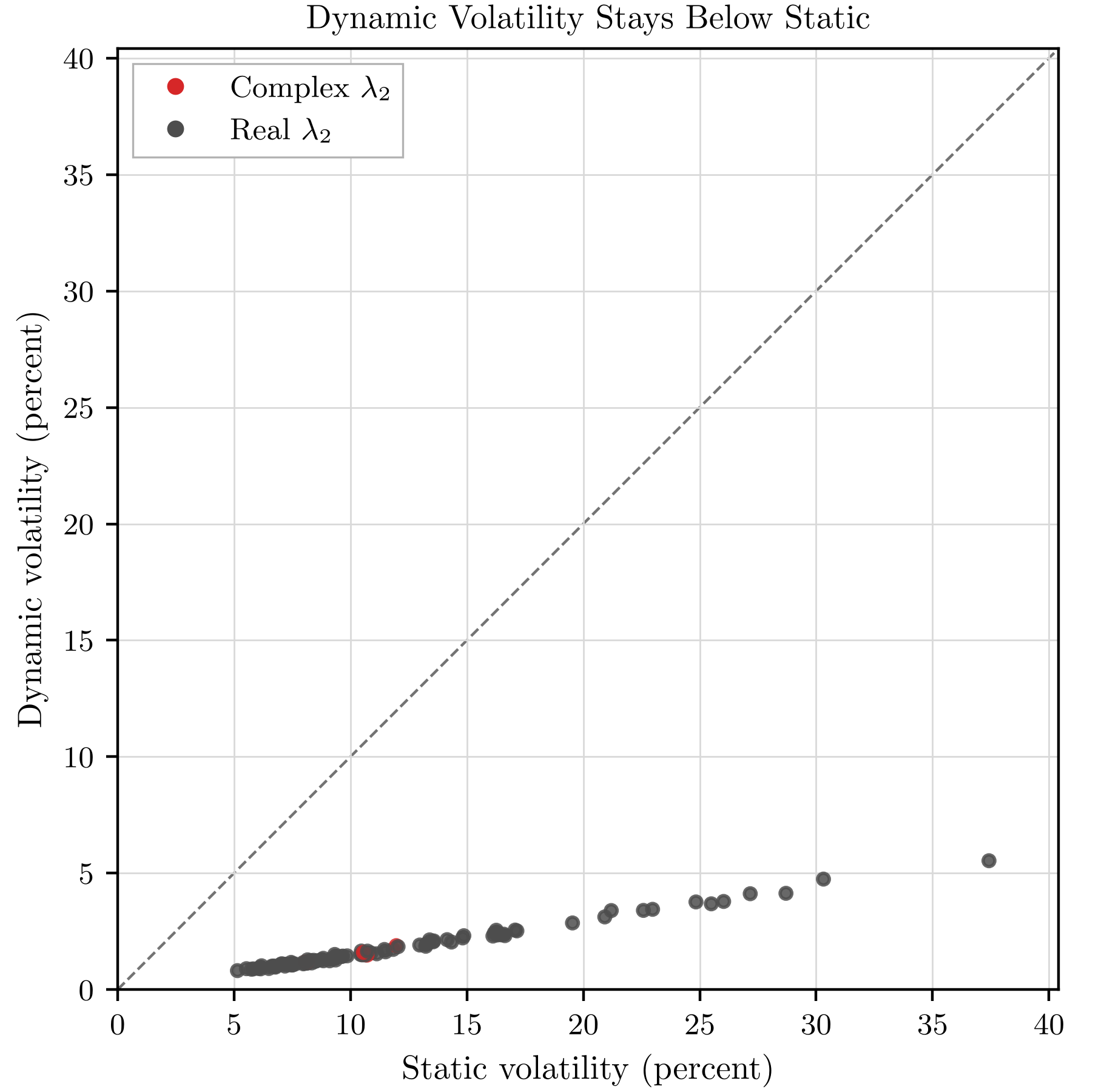}
\caption{Dynamic versus static aggregate volatility, one point per reconstructed US network
($\approx 10^4$ firms, $100$ networks), colored by whether the realized $\lambda_2(\mathbf A)$ is real
or complex.  All points lie below the diagonal $\sqrt{\phi}=\sqrt{\phi^\ast}$ (ratio at most $0.166$, median
$0.148$).}
\label{fig:agg_vol_robust}
\end{figure}

The robustness is specific to the fixed-operator environment of this paper. When the propagation operator is itself endogenous to prices, as under CES technologies where the network is remade each period, the sign and phase of $\lambda_2$ can govern whether overlapping responses cancel or reinforce, and the static--dynamic comparison is correspondingly more delicate.

%====================================================
\section{Notation}
\label{app:notation}
%====================================================

This appendix collects the symbols used throughout the paper.  As a general convention, scalars are
lowercase ($x$), vectors bold lowercase ($\mathbf x$), and matrices bold uppercase ($\mathbf X$).
Here $X^t$ denotes a power while $X_t$ and $X^{[t]}$ denote indexing.  We condition on the production
network throughout and suppress this conditioning in the notation.

\medskip
\noindent\textbf{Indices, sets, and horizons.}\par\smallskip
\begin{tabularx}{\linewidth}{@{}l@{\quad}>{\raggedright\arraybackslash}X@{}}
$N$ & set of firms, $N=\{1,\ldots,n\}$, with $n$ the number of firms \\
$i,j$ & firm indices (buyer $i$, supplier $j$) \\
$t$ & calendar date \\
$\tau$ & age of a vintage, the number of dates since a shock arrived \\
$\ell,\ \mathcal L$ & network layer (processing depth); $\mathcal L$ is the per-date depth of the $\mathcal L$-economy \\
$r$ & eigenmode index ($r=1$ Perron, $r=2$ dominant transient) \\
$T$ & sample window length for realized (finite-horizon) statistics \\
\end{tabularx}

\medskip
\noindent\textbf{Primitives and the economic environment.}\par\smallskip
\begin{tabularx}{\linewidth}{@{}l@{\quad}>{\raggedright\arraybackslash}X@{}}
$\mathbf A$ & input-share matrix, $\mathbf A=(a_{ji})$, column-stochastic; $a_{ji}\ge0$ is the share of firm $i$'s spending on input $j$, with $\sum_j a_{ji}=1$ \\
$\varsigma$ & returns to scale, common across firms, $\varsigma\in(0,1)$; the input elasticities sum to $\varsigma$ \\
$\boldsymbol\gamma$ & aggregation weights, an admissible weight vector, $\boldsymbol\gamma\ge0$, $\boldsymbol\gamma^\top\mathbf 1=1$; $\boldsymbol\gamma=\mathbf m^\ast$ are the size weights \\
$\mathbf m_t,\ \mathbf m^\ast$ & nominal balances and their stationary distribution, $\mathbf A\mathbf m^\ast=\mathbf m^\ast$, $\mathbf 1^\top\mathbf m^\ast=1$; $\|\mathbf m^\ast\|_2^2$ is the size Herfindahl \\
$\mathcal D_{i,t}$ & nominal demand faced by firm $i$ at date $t$, $\mathcal D_{i,t}=\sum_j a_{ij}m_{j,t}$ \\
$\boldsymbol\epsilon_t$ & log-productivity innovation at date $t$, $\boldsymbol\epsilon_t=(\epsilon_{i,t})$; i.i.d.\ across $t$ with $\boldsymbol\epsilon_t\sim\mathcal N(\mathbf 0,\sigma^2\mathbf I)$ (Assumption~\ref{assu:innovations}) \\
$\sigma^2$ & per-firm innovation variance \\
$z_{i,t}$ & Hicks-neutral productivity of firm $i$, $z_{i,t}=e^{\epsilon_{i,t}}$ \\
$q_{i,t},\ x_{ji,t}$ & firm $i$'s output and its use of intermediate input $j$ \\

$\mathbf p_t$ & price vector, $p_{i,t}=\mathcal D_{i,t}/q_{i,t}$ \\

$\mathbf b_0$ & deterministic constant vector $\mathbf b_0(\mathbf A,\varsigma)$ in the log-output recursion \\
\end{tabularx}

\medskip
\noindent\textbf{Propagation operator and spectrum.}\par\smallskip
\begin{tabularx}{\linewidth}{@{}l@{\quad}>{\raggedright\arraybackslash}X@{}}
$\mathbf W$ & production-network propagation operator, $\mathbf W=\varsigma\mathbf A^\top$, with $\rho(\mathbf W)=\lambda_1=\varsigma<1$ \\

$(\mathbf I-\mathbf W)^{-1}$ & Leontief inverse (resolvent), $\sum_{\ell\ge0}\mathbf W^\ell$ \\
$\lambda_r,\mathbf v_r,\mathbf u_r$ & eigenvalue with right and left eigenvectors, normalized $\mathbf u_r^\top\mathbf v_s=\delta_{rs}$ (biorthogonality) \\
$\lambda_1$ & Perron eigenvalue, $\lambda_1=\varsigma$, with $\mathbf v_1=\mathbf 1$ and $\mathbf u_1=\mathbf m^\ast$ \\
$c_r,\ \xi_{r,t},\ s_{r,t}$ & aggregate loading, modal innovation, and modal state of mode $r$: $c_r=\boldsymbol\gamma^\top\mathbf v_r$ (with $c_1=1$ for every admissible weight vector), $\xi_{r,t}=\mathbf u_r^\top\boldsymbol\epsilon_t$, $s_{r,t}=\lambda_r s_{r,t-1}+\xi_{r,t}$ \\
$\mathbf M_r$ & coefficient vector of mode $r$ in Theorem~\ref{thm:aggregator_geometry}, $\mathbf M_r=c_r\mathbf u_r$ \\
$\lambda_2(\mathbf A)$ & subdominant eigenvalue of the monetary circuit; $\lambda_2=\varsigma\lambda_2(\mathbf A)$ \\
$\lambda_2$ & dominant transient eigenvalue, $\max\{|\lambda_r|:r\ge2\}=\varsigma\lambda_2(\mathbf A)$; real, positive, simple (Assumption~\ref{assu:dominant_transient}) \\
$\lambda_3$ & next modulus, $\max\{|\lambda_r|:r\ge3\}$; $\lambda_3<\lambda_2$ by Assumption~\ref{assu:dominant_transient} \\
$\mathbf P_r$ & spectral projector of mode $r$, $\mathbf P_r=\mathbf v_r\mathbf u_r^\top$; $\mathbf P_r\mathbf P_s=\delta_{rs}\mathbf P_r$ \\
$\tilde{\mathbf u}_2,\ \tilde{\mathbf v}_2$ & rescaled dominant-transient eigenvectors, $\tilde{\mathbf u}_2=\mathbf u_2/\|\mathbf u_2\|$, $\tilde{\mathbf v}_2=\|\mathbf u_2\|\mathbf v_2$ \\
\end{tabularx}

\medskip
\noindent\textbf{Aggregate output and reduced-form state.}\par\smallskip
\begin{tabularx}{\linewidth}{@{}l@{\quad}>{\raggedright\arraybackslash}X@{}}
$y_t$ & dynamic aggregate output (shock-driven component of $\boldsymbol\gamma$-weighted log output) \\
$y_t^\ast$ & static (fully processed) aggregate output, $y_t^\ast=\boldsymbol\gamma^\top(\mathbf I-\mathbf W)^{-1}\boldsymbol\epsilon_t$ \\
$\mathbf S_{\mathcal L}$ & block operator of the first $\mathcal L$ layers, $\mathbf S_{\mathcal L}=\sum_{\ell<\mathcal L}\mathbf W^\ell$, with $\mathbf S_1=\mathbf I$ and $\sum_{\tau\ge0}\mathbf W^{\tau\mathcal L}\mathbf S_{\mathcal L}=(\mathbf I-\mathbf W)^{-1}$ \\
$y^\ast_{\mathcal L}$ & depth-$\mathcal L$ truncation of the fully processed mapping, $\boldsymbol\gamma^\top\mathbf S_{\mathcal L}\boldsymbol\epsilon$ \\
$y_t^{[\mathcal L]}$ & aggregate output in the $\mathcal L$-economy ($\mathcal L$ layers processed per date); a vintage of age $\tau$ contributes $\boldsymbol\gamma^\top\mathbf W^{\tau\mathcal L}\mathbf S_{\mathcal L}\boldsymbol\epsilon$ \\
$Y_t$ & aggregate quantity index, $\log Y_t=\sum_i\gamma_i\log q_{i,t}$ up to a constant \\
$\Delta y_t$ & one-step increment (growth rate), $y_t-y_{t-1}$; likewise $\Delta y_t^\ast,\ \Delta y_t^{[\mathcal L]}$ \\
$g_t$ & granular (Perron) innovation, $g_t=\mathbf m^{\ast\top}\boldsymbol\epsilon_t$; variance $\sigma_g^2=\sigma^2\|\mathbf m^\ast\|_2^2$ \\
$\eta_t$ & dominant-mode (reallocation) innovation, $\eta_t=\tilde{\mathbf u}_2^\top\boldsymbol\epsilon_t$; $\Var(\eta_t)=\sigma^2$ \\
$b$ & aggregate loading (visibility) of the dominant transient mode, $b=\boldsymbol\gamma^\top\tilde{\mathbf v}_2$ \\
$s_t$ & dominant-mode state, $s_t=\lambda_2 s_{t-1}+\eta_t$, $s_0=0$; $y_t^{\mathrm{re}}= b\,s_t$ \\
$y_t^{\mathrm{gr}},\ y_t^{\mathrm{re}}$ & granular channel and reallocation channel, the $r=1$ term and the rescaled $r=2$ term of \eqref{eq:y_perron_transient} \\
$\boldsymbol\eta,\ \mathbf s,\ \mathbf y,\ \mathbf y^\ast$ & innovations, transient state, and dynamic and static output over a window of length $T$, as vectors in $\R^T$ \\
$\mathbf K$ & overlap operator, $[\mathbf K]_{tq}=\lambda_2^{\,t-q}\mathbf 1_{\{q\le t\}}$; $\mathbf s=\mathbf K\boldsymbol\eta$ \\
$\mathbf D_T$ & first-difference matrix, of size $(T-1)\times T$ \\
$\mathbf M^\ast$ & static weighting matrix, $\mathbf M^\ast=\mathbf D_T\mathbf D_T^\top$; eigenvalues $\{\nu_j^\ast\}$ \\
$\mathbf M$ & dynamic weighting matrix, $\mathbf M=\mathbf D_T\mathbf K\mathbf K^\top\mathbf D_T^\top$; eigenvalues $\{\nu_j\}$ \\
$\nu_{(j)},\ \nu^\ast_{(j)}$ & the eigenvalues of $\mathbf M$ and $\mathbf M^\ast$ in descending order \\
$Z_j$ & independent standard normal variables in the weighted chi-square representations \\
\end{tabularx}

\medskip
\noindent\textbf{Volatility and tail-risk objects.}\par\smallskip
\begin{tabularx}{\linewidth}{@{}l@{\quad}>{\raggedright\arraybackslash}X@{}}
$\phi^\ast,\ \phi$ & population static and dynamic increment variances, $\lim_t\Var(\Delta y_t^\ast)$ and $\lim_t\Var(\Delta y_t)$; likewise $\phi^{[\mathcal L]}$ in the $\mathcal L$-economy \\
$\hat\phi^\ast,\ \hat\phi$ & realized (finite-$T$) sample volatilities \\
$\omega_c^{[\mathcal L]}$ & lower-tail probability $\Pr(\Delta y_t^{[\mathcal L]}<-c)$ in the $\mathcal L$-economy \\
$\omega_c^\ast,\ \omega_c$ & population lower-tail probabilities $\Pr(\Delta y_t^\ast<-c)$, $\Pr(\Delta y_t<-c)$ at threshold $c>0$ \\
$\hat\omega_c^\ast,\ \hat\omega_c$ & realized lower-tail frequencies over the window \\
$\mathcal R(\cdot)$ & attenuation factor $\mathcal R(\lambda)=(1-\lambda)^2/(1+\lambda)$ (fraction of static increment variance surviving overlap); reallocation $\mathcal R(\lambda_2)$, granular $\mathcal R(\lambda_1)=(1-\varsigma)^2/(1+\varsigma)$ \\
$\kappa$ & timing wedge of the reallocation channel, $\kappa=\phi/\phi^\ast\in(0,1)$; equals $\mathcal R(\lambda_2)$ under the dominant-transient reduction \\
$x$ & standardized threshold in Proposition~\ref{prop:tail_amp}, $x=c/\sqrt{\phi}$ \\
$p(\lambda_2)$ & reversal frequency $\Pr(|\Delta y_t|>|\Delta y_t^\ast|)$ of the channel (Proposition~\ref{prop:reversal_frequency}) \\
$\Phi$ & standard normal cumulative distribution function \\
\end{tabularx}

\medskip
\noindent\textbf{Degree distribution and power-law tail (Section~\ref{sec:spectral_gap_mediates_fat_tails}, Appendix~\ref{app:fat_tails_degree_approx}).}\par\smallskip
\begin{tabularx}{\linewidth}{@{}l@{\quad}>{\raggedright\arraybackslash}X@{}}
$\alpha$ & power-law tail exponent of the degree distribution (smaller $\alpha$ means fatter tails) \\
$\mathbf d$ & degree vector, $\mathbf d=(d_i)$; $\E[d],\ \Var(d)$ its mean and variance \\
$d_{\max}$ & degree cutoff, $d_{\max}\approx n^{1/\alpha}$ \\
$\psi_n^{(k)}(\alpha)$ & truncation factor, $\psi_n^{(k)}(\alpha)=(1-n^{(k-\alpha)/\alpha})/(1-n^{-1})$, $k=1,2$ \\
$b(\alpha),\ \lambda_2(\alpha)$ & loading and dominant transient eigenvalue as functions of the tail exponent \\
$\E_{\boldsymbol\gamma}[d],\ \Delta(\alpha)$ & mean degree under the observer's weights, $\boldsymbol\gamma^\top\mathbf d$, and the misalignment gap $\Delta(\alpha)=\E_{\boldsymbol\gamma}[d]-\E[d^2]/\E[d]$ \\
$\theta$ & tilt of the weight family $\gamma_i\propto d_i^{\,\theta}$ ($\theta=1$ gives size weights under the degree proxy) \\
$e(\alpha,\theta)$ & growth exponent of the misalignment gap in $n$ for $\alpha>2$ and $\theta>\alpha-1$ \\
\end{tabularx}

\medskip
\noindent\textbf{General mathematical notation.}\par\smallskip
\begin{tabularx}{\linewidth}{@{}l@{\quad}>{\raggedright\arraybackslash}X@{}}
$\E,\ \Var,\ \Cov$ & expectation, variance, covariance \\
$\mathbf 1$ & all-ones vector in $\R^n$ \\
$\rho(\cdot)$ & spectral radius \\
$\|\cdot\|_1,\ \|\cdot\|_2$ & $\ell_1$ and $\ell_2$ (Euclidean) norms; for a matrix, $\|\cdot\|_2$ is the operator norm \\
$\mathbf L$ & lower shift on $\R^T$, $[\mathbf L]_{tq}=\mathbf 1\{t=q+1\}$ \\
$\preceq$ & Loewner order on symmetric matrices \\
$\overset{d}{=}$ & equality in distribution \\
$\succeq_{\mathrm{FOSD}}$ & first-order stochastic dominance \\
\end{tabularx}

\end{document}